\documentclass[a4paper,fleqn]{cas-sc}

\usepackage[numbers]{natbib}
\usepackage{xcolor}
\usepackage{cancel}
\usepackage{tikz}
\usepackage{tikz-dimline}
\usepackage{mathtools}
\usepackage{grffile}
\usepackage{amsthm}

\def\Gb{\Gamma_{\scriptsize\mbox{b}}}
\def\Gi{\Gamma_{\scriptsize\mbox{in}}}
\def\Go{\Gamma_{\scriptsize\mbox{out}}}
\def\Gf{\Gamma_{\scriptsize\mbox{fs}}}
\def\Gs{\Gamma_{\scriptsize\mbox{str}}}
\def\Gbh{\Gamma_{\scriptsize\mbox{b},h}}
\def\Gih{\Gamma_{\scriptsize\mbox{in},h}}
\def\Goh{\Gamma_{\scriptsize\mbox{out},h}}
\def\Gfh{\Gamma_{\scriptsize\mbox{fs},h}}
\def\Gsh{\Gamma_{\scriptsize\mbox{str},h}}
\def\tGsh{\tilde{\Gamma}_{\scriptsize\mbox{str},h}}
\def\Ls{\Lambda_{\scriptsize\mbox{str}}}
\def\Lfs{\Lambda_{\scriptsize\mbox{fs,str}}}
\def\Lsh{\Lambda_{\scriptsize\mbox{str},h}}
\def\Lj{\Lambda_{\scriptsize\mbox{j}}}
\def\Ekf{E_{\scriptsize\mbox{kin,flow}}}
\def\Epf{E_{\scriptsize\mbox{pot,flow}}}
\def\Eks{E_{\scriptsize\mbox{kin,str}}}
\def\Ees{E_{\scriptsize\mbox{ela,str}}}
\def\Et{E_{\scriptsize\mbox{total}}}
\def\Eth{E_{\scriptsize\mbox{total,h}}}
\def\uin{\text{u}_{\scriptsize\mbox{in}}}
\def\uout{\text{u}_{\scriptsize\mbox{out}}}
\def\u{\mathbf{u}}
\def\x{\mathbf{x}}
\def\n{\mathbf{n}}
\def\nL{\mathbf{n}_\Lambda}
\def\C{\mathbf{C}}
\def\nphi{\nabla\phi}
\def\Eq#1{(\ref{#1})}
\def\V{\mathcal{V}}
\def\hatV{\hat{\V}}

\def\Drho{D_\rho}
\def\gammaNB{\gamma_{\scriptsize\mbox{NB}}}
\def\betaNB{\beta_{\scriptsize\mbox{NB}}}
\newtheorem{assumption}{Assumption}
\newtheorem{remark}{Remark}
\newtheorem{definition}{Definition}
\newtheorem{theorem}{Theorem}
\newtheorem{corollary}{Corollary}
\newtheorem{proposition}{Proposition}
\newcommand{\vertiii}[1]{{\left\vert\kern-0.25ex\left\vert\kern-0.25ex\left\vert #1 
		\right\vert\kern-0.25ex\right\vert\kern-0.25ex\right\vert}}
\newcommand{\vertiiidg}[1]{{\left\vert\kern-0.25ex\left\vert\kern-0.25ex\left\vert #1 
		\right\vert\kern-0.25ex\right\vert\kern-0.25ex\right\vert_{\scriptsize DG}}}
\newcommand{\average}[1]{\ensuremath{\langle#1\rangle} }
\newcommand{\jump}[1]{\ensuremath{[\![#1]\!]} }

\begin{document}
\let\WriteBookmarks\relax
\def\floatpagepagefraction{1}
\def\textpagefraction{.001}

\shorttitle{A monolithic FE formulation for hydroelastic analysis of VLFS}
\shortauthors{O. Colom\'es et~al.}
%\begin{frontmatter}

\title [mode = title]{A monolithic Finite Element formulation for the hydroelastic analysis of Very Large Floating Structures}                      

%\tnotemark[1,2]
%\tnotetext[1]{This document is the results of the research project funded by the National Science Foundation.}
%\tnotetext[2]{The second title footnote which is a longer text matter    to fill through the whole text width and overflow into another line in the footnotes area of the first page.}

%% Oriol %%
\author[1]{Oriol Colom\'es}[type=author,
                        auid=000,
                        bioid=1,
                        prefix=,
                        role=,
                        orcid=0000-0002-5552-9695]
\cormark[1]
\ead{j.o.colomesgene@tudelft.nl}
\ead[url]{www.oriolcolomes.com}
\credit{Conceptualization of the formulation, analysis, software, numerical results, writing - original draft preparation}
\address[1]{Faculty of Civil Engineering and Geosciences, Delft University of Technology, Stevinweg 1, 2628 CN, Delft, The Netherlands}

%% Francesc %%
\author[2]{Francesc Verdugo}[type=author,
	auid=000,
	bioid=1,
	prefix=,
	role=,
	orcid=0000-0003-3667-443X]
\ead{fverdugo@cimne.upc.edu}
\ead[url]{www.francescverdugo.com}
\credit{software, writing}
\address[2]{CIMNE -- Centre Internacional de M\`etodes Num\`erics a l'Enginyeria, Esteve Terrades 5, 08860 Castelldefels, Spain.}
\newcommand{\fvcom}[1]{{\color{blue}@francesc: #1}}

%% Ido %%
\author[3]{Ido Akkerman}[type=author,
auid=000,
bioid=1,
prefix=,
role=,
orcid=0000-0002-5937-0300]
\ead{i.akkerman@tudelft.nl}
\ead[url]{http://homepage.tudelft.nl/4x7d5}
\credit{analysis, writing}
\address[3]{Faculty of Mechanical, Maritime and Materials Engineering, Delft University of Technology, Mekelweg 2, 2628 CD, Delft, The Netherlands}

%\author[2,4]{Han Theh Thanh}[style=chinese]
%\author[2,3]{CV Rajagopal}[%
%   role=Co-ordinator,
%   suffix=Jr,
%\ead{cvr3@sayahna.org}
%\ead[URL]{www.sayahna.org}
%\credit{Data curation, Writing - Original draft preparation}
%\address[2]{Sayahna Foundation, Jagathy, Trivandrum 695014, India}

%\author%
%[1,3]
%{Rishi T.}
%\cormark[2]
%\fnmark[1,3]
%\ead{rishi@stmdocs.in}
%\ead[URL]{www.stmdocs.in}
%\address[3]{STM Document Engineering Pvt Ltd., Mepukada,
%    Malayinkil, Trivandrum 695571, India}

\cortext[cor1]{Corresponding author}
%\fntext[fn1]{This is the first author footnote. but is common to third
%  author as well.}
%\fntext[fn2]{Another author footnote, this is a very long footnote and
%  it should be a really long footnote. But this footnote is not yet
%  sufficiently long enough to make two lines of footnote text.}

%\nonumnote{This note has no numbers. In this work we demonstrate $a_b$
%  the formation Y\_1 of a new type of polariton on the interface
%  between a cuprous oxide slab and a polystyrene micro-sphere placed
%  on the slab.
%  }

\begin{abstract}
	 In this work we present a novel monolithic Finite Element Method (FEM) for the hydroelastic analysis of Very Large Floating Structures (VLFS) with arbitrary shapes that is stable, energy conserving and overcomes the need of an iterative algorithm. The new formulation enables a fully monolithic solution of the linear free-surface flow, described by linear potential flow, coupled with floating thin structures, described by the Euler-Bernoulli beam or Poisson-Kirchhoff plate equations. 
	 The formulation presented in this work is general in the sense that solutions can be found in the frequency and time domains, it overcomes the need of using elements with $ C^1 $ continuity by employing a continuous/discontinuous Galerkin (C/DG) approach, and it is suitable for Finite Elements of arbitrary order.
	 We show that the proposed approach can accurately describe the hydroelastic phenomena of VLFS with a variety of tests, including structures with elastic joints, variable bathymetry and arbitrary structural shapes.
\end{abstract}

%%\includegraphics{figs/grabs.pdf}

%\item Novel monolithic Finite Element formulation for the hydroelastic analysis of thin structures coupled with potential flow
%\item The formulation is defined for both, $ \mathcal{C}^1 $ and $ \mathcal{C}^0 $, Finite Elements
%\item The formulation can be implemented in multi-purpose Finite Elements libraries. In this work we show a practical implementation using Gridap.jl
%\item We proof of stability, convergence and energy conservation properties of the formulation, with the corresponding numerical results
%\item The formulation can be used for 2 and 3-dimensional problems with arbitrary geometries including elastic joints


\maketitle

\section{Introduction}
Floating offshore structures are of great interest for many applications. A particular type of floating structures are the so called Very Large Floating Structures (VLFS). One can find several examples of VLFS \cite{watanabe2004hydroelastic}, such as floating airports \cite{isobe1999research,zhang2017connection}, floating breakwaters \cite{nagata1998prediction,cheng2022wave}, floating solar energy installations \cite{trapani2015review,sahu2016floating}, or even futuristic floating modular cities \cite{takeuchi2022mega,dai2022modular}. The study of the behavior of VLFS is, therefore, relevant for a wide variety of industries and scientific disciplines. %The response of VLFS is governed by the coupling of hydrodynamic and elastic theories, in what is known as hydroelasticity. 
One of the main characteristics of VLFS is that the overall structural stiffness is relatively low, behaving like elastic thin plates. In addition, due to their large dimensions, the incoming waves are typically relatively short compared with the structure length. Thus, due to the low stiffness combined with short incoming waves, the response of VLFS is governed by a strong coupling between inertial, hydrodynamic and elastic responses, what is known as hydroelastic response.   

The study of hydroelastic phenomena entails several challenges, namely the strong coupling between the elastic deformation of the structure and its hydrodynamic response, the analysis of the structural response under the effect of nonlinear waves, the characterization of the behavior of finite structures or the nonlinear interaction between flow and structure, see \cite{korobkin2011mathematical}. During the last decades, several techniques have been developed to analyse hydroelastic phenomena for VLFS under the effect of waves. We refer the reader to \cite{chen2006review} for an in depth review on different methods used for the hydroelastic analysis of VLFS. Some studies have been carried out based on experimental analysis of floating elastic platforms, see for instance~\cite{yago1996hydoroelastic,liu2002time,schreier2020experimental}. However, experimental studies are limited in terms of structural size and wave conditions. Other studies are based on analytical or semi-analytical approaches, see~\cite{tsubogo1999dispersion,ohkusu2004hydroelastic,andrianov2005hydroelasticity,andrianov2006hydroelastic,xu2022theoretical}, where the fundamental behaviour of floating elastic structures is assessed, assuming infinite or finite floating platforms with regular shapes. Again, analytical approaches are limited to the study of VLFS with regular shapes, \textit{e.g.} rectangular or circular platforms. Hence, the use of numerical techniques is essential for the analysis of the hydroelastic behavior of VLFS of finite size with irregular shapes and subject to a variety of wave input conditions. 

One of the most popular numerical approaches for the hydroelastic analysis of VLFS relies on the mode expansion framework for linear potential flow theory, based on the Boundary Element Method. Assuming a negligible structure draft, this approach models the effect of the structure by a pressure distribution on the free surface. The dynamic pressure distribution is approximated by a set of panels with different accuracies, \textit{e.g.} piece-wise constant, linear, quadratic  or cubic, see~\cite{maeda1995hydroelastic,yasuzawa1997dynamic,hamamoto19983d,kashiwagi1998b}. Most of the works based on the panel method, or mode expansion, assume regular shaped structures. For the analysis of irregular shaped VLFS, the Finite Element Method (FEM) becomes predominant in the literature due to its suitability to model the structural behaviour of arbitrary shaped platforms. Some works have developed a coupling strategy between BEM and FEM frameworks, see for instance~\cite{hamamoto19963d,hamamoto19983d,wang2004higher,shirkol2018coupled,pal2018fully}.
The study of the hydroelastic phenomena in VLFS can be done in either the frequency domain or the time domain. The former assumes a linear response of the transient effects described as a time-harmonic motion, see for instance \cite{karperaki2021hydroelastic,liu2020dmm,ding2019simplified}. The time-domain analysis avoids the linear assumption, making it a suitable approach for the analysis of steep wave fronts and cases with highly nonlinear effects, \cite{liu2002time,karperaki2016time}. In this work we develop a formulation that is suitable for both approaches, frequency and time domain.

VLFS hydroelasticity is inherently a Fluid-Structure Interaction (FSI) problem, that is a  fluid flow problem coupled with a structural elasticity problem. There are two main frameworks that can be used to solve this type of problems: a monolithic approach, where a unified set of coupled governing equations are solved, and a partitioned approach, where the equations for the fluid and for the structure are solved separately and a coupling strategy is used to ensure compatibility between the two solutions. 
One coupling strategy that can be used for a partitioned scheme is a weak coupling (also called a staggered scheme) strategy. Here, a new fluid solution is solved using the structural displacement of the previous time step.  While the obtained fluid solution is used to determine the loading on the structure. One of the advantages of partitioned schemes is that, if the software is available separately for the flow and structural deformation, one can reuse them and just develop an interface that exchanges the coupling variables. However, stable behavior of this coupling method is not guaranteed due to the added-mass effect~\cite{causin2005added}. In case of incompressible fluids, selecting a smaller time step does not resolve this problem \cite{vanBrummelen2009}. 
 
Alternatively, one could use a strong coupling strategy, where  the compatibility conditions are enforced simultaneously. This is done by iteratively solving fluid and structural updates until convergence is achieved. 
For a strongly coupled the aforementioned instabilities due added-mass effect can be avoided by using sufficient relaxation in the iteration procedure. However,  this may lead to excessive iterations even if the two separate problems are linear, see for instance~\cite{kyoung2006fem,kim2009analysis,weir2011nonlinear}. 
Otherwise, monolithic approaches are stable and do not require additional iterations to resolve the coupling. Mayor drawback is the fact that \textit{ad hoc} software has to be developed in case it is  not available for the coupled problem at hand. A particular feature of monolithic VLFS simulations is that they lead to a mixed-dimensional problem, \textit{i.e.} a problem described by a system of Partial Differential Equations (PDEs) defined in domains of different topological dimension. The solution of mixed-dimensional PDEs in a FEM framework presents an additional challenge, related to the coupling of Finite Element (FE) spaces defined in different dimensions, see \cite{DaversinCatty2021}.

In this work we present a novel monolithic FEM framework for the hydroelastic analysis of VLFS with arbitrary shapes that is stable, energy conserving and overcomes the need of an iterative algorithm. The formulation is based on the monolithic FEM approach proposed by Akkerman \textit{et al.} in~\cite{akkerman2020isogeometric} for linear free-surface potential flow, which is here extended to floating structures modeled as Euler-Bernoulli beams or Poisson-Kirchhoff plates. Moreover, the formulation presented in this manuscript is general in the sense that solutions can be found in the frequency and time domains. In addition, the proposed formulation overcomes the need of using elements with $ C^1 $ continuity, \textit{i.e.} continuity of rotations on the element boundaries, by using a continuous/discontinuous Galerkin (C/DG) approach formulated in~\cite{engel2002continuous}. Note that we restrict this manuscript to the analysis of linear problems, \textit{i.e.} linear potential flow coupled with linear Euler-Bernoulli/Poisson-Kirchhoff structural formulations. However, the framework would still hold for nonlinear potential flow theory and/or nonlinear structural models. For the sake of completeness, in this manuscript we focus on a detailed analysis of the formulation for linear problems, keeping its extension to nonlinear problems as a future work.

The manuscript is organized as follows: in Section~\ref{sec:problem_setting} we describe the problem setting, with the definition of the governing equations for the fluid, structure in two and three-dimensional cases, as well as the coupling conditions. In Section~\ref{sec:formulation} we develop the novel monolithic formulation for VLFS for the most general case, \textit{i.e.} a floating Poisson-Kirchhoff plate in a three-dimensional domain, giving the expression of the fully discrete problem for both, frequency and time domains. Section~\ref{sec:numerical_analysis} is dedicated to the numerical analysis of the method, where we prove consistency, energy conservation, stability and convergence statements. The numerical results are shown in Section~\ref{sec:numerical_results}, where we analyse the behavior of the method for a variety of cases in two and three dimensions and for the frequency and time domains. The final conclusions are given in Section~\ref{sec:conclusions}.

\section{Problem setting}\label{sec:problem_setting}
Let us consider a thin structure floating in a fluid. We denote the fluid domain as $ \Omega $, bounded by the bottom surface, $ \Gb $, the inlet surface, $ \Gi $, the outlet surface, $ \Go $, the free surface, $ \Gf $, and the interface with the floating structure, $ \Gs $. The floating structure is bounded by $\delta\Gs\equiv\Lfs$, a set of entities of dimension $d{-}2$ with $d$ the topological dimension of $\Omega$. We might also consider the case in which the structure has a set of internal joints with different structural properties, denoted as $ \Lj $. The geometry of the idealized problem is given in Figure~\ref{fig:domain}. 
\begin{figure}[pos=h!]
	\centering
	\includegraphics[width=0.6\textwidth]{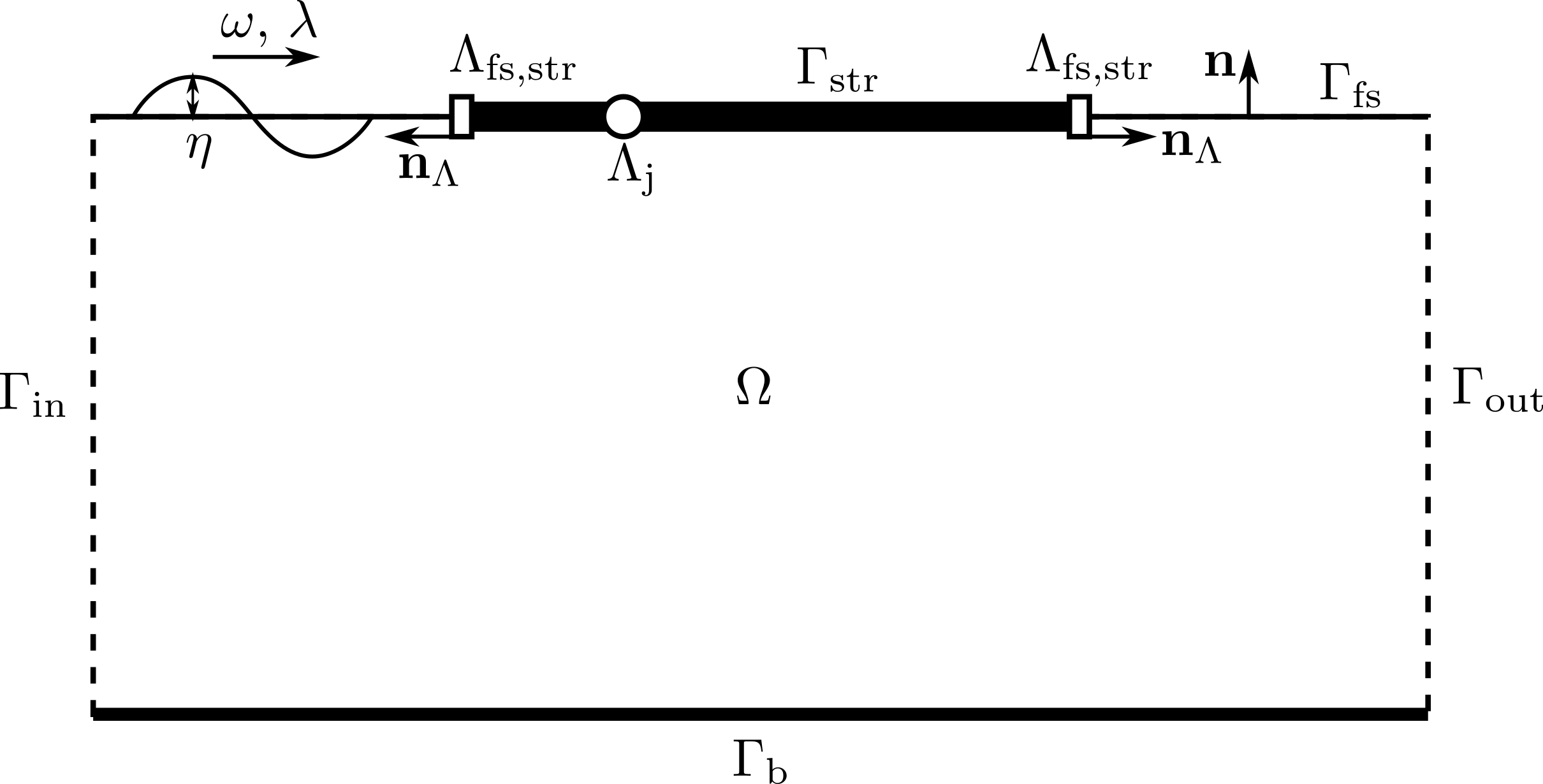}
	\caption{Sketch with the definition of the geometrical entities.}
	\label{fig:domain}	
\end{figure}
We use  $ \n $ for the  normal vector to the surface, while
$ \nL $ to indicate the normal to the joint on the structure plane. 
% Remove?? to confusing?? In one-dimensional problems, where only the structure is modeled, one could omit $ \nL $ from the formulation.
% Remove?? to confusing??  However, since this formulation is defined on a 2-dimensional domain, \textit{i.e} $ \Omega\in\mathbb{R}^2 $, we need the normal vector $ \nL $ to obtain the function derivative along the structure. 
Note also that the normal $ \nL $ appears naturally in the formulation when integrating by parts on the structural domain, see Section~\ref{subsec:weak_form}.

Let us also consider the following assumptions:
\begin{assumption}
	\label{ass:pflow}
	The fluid in $ \Omega $ is inviscid, incompressible and irrotational and can be described well by  (linear) potential flow.
	We also assume that there is no cavitation, \textit{i.e.} detachment of the structure with respect to the fluid.
\end{assumption}
\begin{assumption}
	\label{ass:linear}
	The incoming waves have a small steepness and can be well described well by (linear) Airy wave theory.
\end{assumption}
\begin{assumption}
\label{ass:beam}
The floating structures are thin and can be modeled by the linear Euler-Bernoulli equations in the case of a 2-dimensional domain and by the linear Poisson-Kirchhoff theory for the 3-dimensional domain. 
\end{assumption}
\subsection{Linear potential flow theory}
Let us denote the fluid velocity as $ \u:\Omega\rightarrow\mathbb{R}^d $. From the incompressibility condition assumed in Assumption \ref{ass:pflow}, we know that
\begin{equation}\label{eq:incompressibility}
	\nabla\cdot\u=0\quad\mbox{in }\Omega.
\end{equation}
We also know that for an inviscid and irrotational flow, there exists a potential field, $ \phi:\Omega\rightarrow\mathbb{R} $, that satisfies
\begin{equation}\label{eq:potential}
	\u=\nabla\phi.
\end{equation}
Combining equations \Eq{eq:incompressibility} and \Eq{eq:potential}, we reach the governing equation for a potential flow:
\begin{equation}\label{eq:pflow}
	\Delta\phi = 0\quad\mbox{in }\Omega.
\end{equation}
Equation \Eq{eq:pflow} is supplemented with the appropriate boundary and interface conditions. Before detailing such conditions, we introduce some notation. Let us denote by $ \eta $ the surface elevation with respect to the surface of the fluid at rest, and by $ \n $ the normal vector to any surface pointing outwards, see Figure~\ref{fig:domain}. We also use the notation $ (\cdot)_t$  and $ (\cdot)_{tt}$ for the first and second order  time derivative. Using this notation, the kinematic boundary conditions are given by
\begin{subequations}
	\label{eq:kinematic_bc}
	\begin{align}
		\label{eq:kinematic_bc_b}
		\n\cdot\nphi&=0\quad\mbox{on }\Gb,\\
		\label{eq:kinematic_bc_i}
		\n\cdot\nphi&=\uin\quad\mbox{on }\Gi,\\
		\label{eq:kinematic_bc_o}
			\n\cdot\nphi&=\uout\quad\mbox{on }\Go,\\
		\label{eq:kinematic_bc_fs}
		\n\cdot\nphi&=\eta_t\quad\mbox{on }\Gf\cup\Gs.
	\end{align}
\end{subequations}
Equation \Eq{eq:kinematic_bc_b} is enforcing a \textit{no penetration} boundary condition, where the velocity normal to the boundary is zero. In equations \Eq{eq:kinematic_bc_i} and \Eq{eq:kinematic_bc_o} we enforce a prescribed inlet and outlet normal velocities, $ \uin $ and $ \uout $, respectively. At the free surface and at the beam interface, equation \Eq{eq:kinematic_bc_fs}, we enforce that the time derivative of the fluid surface is equal to the normal component of the velocity. Note that this last condition is only valid under Assumption \ref{ass:linear}, where due to small wave steepness the normal component of the velocity is well approximated by the vertical velocity component. When this assumption is not valid, equation \Eq{eq:kinematic_bc_fs} should be replaced by $ \phi_z-\nabla\phi\nabla\eta=\eta_t $, where $ \phi_z $ is the directional derivative of $ \phi $ in the vertical direction.

In addition, the fluid satisfies the dynamic boundary conditions at the free surface and at the interface between the fluid and the structure. This condition is given by the Bernoulli's equation for pressure in a potential flow, which at the free surface and interface boundary reads
\begin{equation}\label{eq:dynamic_bc}
	p=-\rho_w\phi_t-g\rho_w\eta-\frac{1}{2}\left(\nphi\right)^2\quad\mbox{on }\Gf\cup\Gs.
\end{equation}
Where $ p $ is the pressure, $ \rho_w $ the fluid density and $ g=9.81\ \mbox{m/s$ ^2 $} $ the gravity acceleration. Under Assumption \ref{ass:linear}, the quadratic terms in equation \Eq{eq:dynamic_bc} can be neglected, leading to the linearized condition
\begin{equation}\label{eq:dynamic_linear_bc}
p=-\rho_w\phi_t-g\rho_w\eta\quad\mbox{on }\Gf\cup\Gs.
\end{equation}
At the free surface, we enforce that the pressure is equal to the atmospheric pressure, which we assume to be zero, \textit{i.e.} $ p=p_a=0 $. At the fluid-structure interface, the pressure is in equilibrium with the structure dynamics. Using Assumption \ref{ass:beam}, \textit{i.e.} no cavitation and thin beam/plate theories, the structural motion is governed either by the Euler-Bernoulli equation form beams or the Poisson-Kirchoff plate theory. In what follows we describe the two cases.

\subsection{Euler-Bernoulli beam theory}\label{subsec:euler_bernoulli}

We start by defining the 2-dimensional case, where the floating structure is modeled as a 1-dimensional Euler-Bernoulli beam. In that case, we have that the surface elevation $\eta:\Gs\rightarrow\mathbb{R}$, which from Assumption~\ref{ass:beam} is equivalent to the beam deflection, satisfies
\begin{equation}\label{eq:beam_eqn}
	\rho_bh_b\eta_{tt} + D\Delta^2\eta=p\quad\mbox{on }\Gs.
\end{equation}
With $ \rho_b $ the structure density, $ h_b $ the structure thickness and $ D $ the structural rigidity, given in terms of the Young modulus $E$ and the moment of inertia as$ D\eqdef EI=Eh_b^3/12 $.  Here, $ D $ is assumed to be constant.  Adding equations \Eq{eq:dynamic_linear_bc} and \Eq{eq:beam_eqn} together and using the fact that the atmospheric pressure is zero at the free surface, we have that the potential field and the fluid surface elevation satisfy the following conditions at the fluid free surface and fluid-structure interface
\begin{subequations}
	\label{eq:dynamic_bc_f_s}
	\begin{align}
		\label{eq:dynamic_linear_bc_f}
		\phi_t+g\eta=0\quad\mbox{on }\Gf,\\
		\label{eq:dynamic_linear_bc_s}
		d_0\eta_{tt} + \Drho\Delta^2\eta+\phi_t+g\eta=0\quad\mbox{on }\Gs.
	\end{align}
\end{subequations}

Where $ d_0\eqdef\frac{\rho_bh_b}{\rho_w} $ is the draft (submerged section) of the structure and $ D_\rho\eqdef\frac{D}{\rho_w} $. Equation \Eq{eq:pflow} together with the kinematic boundary conditions given in \Eq{eq:kinematic_bc} and the dynamic boundary conditions given in \Eq{eq:dynamic_bc_f_s}, define the systems of equations that will be analysed in subsequent sections.

In addition, we consider that the structure is allowed to have free motion at the boundaries, \textit{i.e.} zero moment, $M\eqdef D\Delta\eta=0$, and shear $V\eqdef\nabla(D\Delta\eta)\cdot\nL=0$. We also assume that in the case of having joints, these will act as linear rotational springs, meaning that the mean moment at the join will depend linearly on the relative rotation angle, with a spring constant of $k_\varphi$. Since there are not any external point moment or shear force applied to the joint, the shear and moment will be continuous at these points. These conditions can be summarized as
\begin{subequations}
	\label{eq:boundary_conditions_beam}
	\begin{align}
	\label{eq:zero_moment}
	D\Delta{\eta} &= 0&\mbox{on }\Lfs,\\
	\label{eq:zero_shear}
	\nabla(D\Delta\eta)\cdot\nL&=0&\mbox{on }\Lfs,\\
	\label{eq:moment_mean}	
	\average{D\Delta{\eta}} &= -k_\varphi\jump{\nabla\eta\cdot\nL}&\mbox{on }\Lj,\\
	\label{eq:zero_moment_jump}	
	\jump{D\Delta{\eta}\nL} &= 0&\mbox{on }\Lj,\\
	\label{eq:zero_shear_jump}
	\jump{\nabla(D\Delta\eta)\cdot\nL}&=0&\mbox{on }\Lj.	
	\end{align}
\end{subequations}

% moved to begin of chapter Note that we use the notation $ \nL $ instead of $ \n $ to differentiate the normal to the joint on the structure plane, $ \nL $, from the normal vector to the surface, $ \n $. In one-dimensional problems, where only the structure is modeled, one could omit $ \nL $ from the formulation. However, since this formulation is defined on a 2-dimensional domain, \textit{i.e} $ \Omega\in\mathbb{R}^2 $, we need the normal vector $ \nL $ to obtain the function derivative along the structure. Note also that the normal $ \nL $ appears naturally in the formulation when integrating by parts on the structural domain, see Section~\ref{subsec:weak_form}.

\subsection{Poisson-Kirchhoff plate theory}\label{subsec:poisson_kirchhoff}
For the 3-dimensional case, the floating structure is modeled as a 2-dimensional thin plate governed by the Poisson-Kirchhoff theory. In this case, the plate deflection, \textit{i.e.} surface elevation, follows the relation
\begin{equation}\label{eq:plate_eqn}
	\rho_bh_b\eta_{tt} + \nabla^2\colon(\C\colon\nabla^2\eta)=p\quad\mbox{on }\Gs.
\end{equation}
Where $ \nabla^2 $ is the Hessian operator and $ \C $ is a fourth-order symmetric tensor containing the elastic coefficients. Note that the problem is defined in $ \mathbb{R}^3 $, so the Hessian operator will be a $ 3\times3 $ matrix, while the elastic tensor $ \C $ has dimension $ 2 $ on each component. Assuming that the structure remains contained in a plane aligned with the coordinate system, \textit{e.g.} normal to the $ z $-direction, one can extend the 2-dimensional definition of the tensor as follows
\begin{equation}\label{eq:elastic_tensor}
	C_{ijkl}=\begin{cases}
		\frac{h_b^3}{12}[\mu(\delta_{ik}\delta_{jl}+\delta_{il}\delta_{jk})+\lambda\delta_{ij}\delta_{kl}]\qquad1\leq i,j,k,l\leq 2,\\
		0\qquad\mbox{otherwise}.
	\end{cases}
\end{equation}
In equation \Eq{eq:elastic_tensor}, $ \delta_{\alpha\beta} $ is the Kronecker delta, $ \mu\eqdef\frac{E}{2(1+\nu)} $ and $ \lambda\eqdef \frac{\nu E}{(1-\nu^2)} $, with $ E $ the Young's modulus and $ \nu $ the Poisson's ratio. 

Following the same procedure as in Section~\ref{subsec:euler_bernoulli}, we introduce equation~\Eq{eq:plate_eqn} into equation~\Eq{eq:dynamic_linear_bc}, resulting into the set of governing equations on the free and structure surfaces that read
\begin{subequations}
	\label{eq:dynamic_bc_f_s_3d}
	\begin{align}
		\label{eq:dynamic_linear_bc_f_3d}
		\phi_t+g\eta=0\quad\mbox{on }\Gf,\\
		\label{eq:dynamic_linear_bc_s_3d}
		d_0\eta_{tt} + \nabla^2\colon(\C_\rho\colon\nabla^2\eta)+\phi_t+g\eta=0\quad\mbox{on }\Gs.
	\end{align}
\end{subequations}
Where we have used the notation $ C_{\rho,ijkl}\eqdef \frac{C_{ijkl}}{\rho_w} $. Again, we consider that the structure is free of reactions on its boundaries $\Lfs$, \textit{i.e.} zero moment and zero shear. We also assume that we can have linear rotational springs on 1-dimensional joints, $ \Lj $, where the normal is linearly dependent on the rotation to the joint, with a spring constant of $ k_\phi $, assumed scalar for simplicity. These conditions are summarized in the following set of equations
 \begin{subequations}
 	\label{eq:boundary_conditions_plate}
 	\begin{align}
 		\label{eq:zero_moment_plate}
 		(\C_\rho:\nabla^2\eta)\cdot\nL &= 0&\mbox{on }\Lfs,\\
 		\label{eq:zero_shear_plate}
 		(\nabla\cdot(\C_\rho:\nabla^2\eta))\cdot\nL&=0&\mbox{on }\Lfs,\\
 		\label{eq:moment_mean_plate}	
 		\average{\C_\rho:\nabla^2\eta} &= -k_\varphi\jump{\nabla\eta\otimes\nL}&\mbox{on }\Lj,\\
 		\label{eq:zero_moment_jump_plate}	
 		\jump{(\C_\rho:\nabla^2\eta)\cdot\nL} &= 0&\mbox{on }\Lj,\\
 		\label{eq:zero_shear_jump_plate}
 		\jump{(\nabla\cdot(\C_\rho:\nabla^2\eta))\cdot\nL}&=0&\mbox{on }\Lj.	
 	\end{align}
 \end{subequations}

Note that for 2-dimensional problems the Poisson-Kirchhoff formulation presented in equations~\Eq{eq:dynamic_bc_f_s_3d}-\Eq{eq:boundary_conditions_plate} collapses to the Euler-Bernoulli formulation given in~\Eq{eq:dynamic_bc_f_s}-\Eq{eq:boundary_conditions_beam}. Hereinafter we will restrict the derivations for the most general Poisson-Kirchhoff case. We refer the reader to Appendix~\ref{sec:appendix_A} for the equivalent formulation for the floating the Euler-Bernoulli beam case.

\section{A monolithic Finite Element formulation}\label{sec:formulation}
\subsection{Weak form}\label{subsec:weak_form}
In order to derive the weak form of the problem given by equations \Eq{eq:pflow}, \Eq{eq:kinematic_bc} and \Eq{eq:dynamic_bc_f_s}, we first introduce some notation that will be used hereinafter. Let us denote by $L^r(\Omega)$, $1\leq r<\infty$, the spaces of functions such that their $r$-th power is absolutely integrable in $\Omega$. For the case in which $r=2$, we have a Hilbert space with inner product 
\begin{equation}
\label{eq:scalar_product_O}
(u,v)_\Omega\eqdef\int_\Omega u(\x) \, v(\x)d\Omega
\end{equation}
and induced norm $\|u\|_{L^2(\Omega)}\equiv\|u\|_\Omega\eqdef(u,u)_\Omega^{1/2}$. Abusing of the notation, the same symbol as in (\ref{eq:scalar_product_O}) will be used for the integral of the product of two functions, even if these are not in $L^2(\Omega)$, and both for scalar and vector fields. The space of functions whose distributional derivatives up to order $m$ are in $L^2(\Omega)$ are denoted by $H^m(\Omega)$. We will focus on the case of $m=1$, which is also a Hilbert space. Given a Banach space $X$,  $L^r(0,T;X)$ is the space of time dependent functions such that their $X$-norm is in $L^r(0,T)$, being $ [0,T] $ the time interval where such functions are defined. In addition, we will also use the inner product on a given boundary, $ \Gamma \subset \Omega $, defined as
\begin{equation}
\label{eq:scalar_product_G}
(u,v)_\Gamma\eqdef\int_\Gamma u(\x) \, v(\x)d\Gamma,
\end{equation}
with the associated norm $\|u\|_{L^2(\Gamma)}\equiv\|u\|_\Gamma\eqdef(u,u)_\Gamma^{1/2}$.

Let $ \V\eqdef L^2(0,T;H^2(\Omega)) $ be a functional space, $ \V_{\Gf} $ the trace space of $ \V $ on the free surface $ \Gf $, i.e. $ \V_{\Gf}\eqdef\{v|_{\Gf}:v\in\V\} $, and $ \V_{\Gs} $ the trace space of $ \V $ on the structure $ \Gs $. The weak form of the problem reads: find $ [\phi,\kappa,\eta] \in\V\times\V_\Gamma$ such that
\begin{equation}\label{eq:weak_form}
B([\phi,\kappa,\eta],[w,v,u])=L([w,v,u])\quad\forall[w,v,u]\in\V\times\V_{\Gf}\times\V_{\Gs}.
\end{equation}
Where the bilinear form, assuming that the structure is continuous and satisfies the boundary conditions \eqref{eq:zero_moment_plate}-\eqref{eq:zero_shear_plate}, is given by
\begin{align}\label{eq:bilinear_modified}
B([\phi,\kappa,\eta],[w,v,u])\eqdef&(\nphi,\nabla w)_\Omega - (\kappa_t,w)_{\Gf} + \beta\left(\phi_t+g\kappa,\alpha_fw+v\right)_{\Gf} \\\nonumber
-&(\eta_t,w)_{\Gs} +\left(d_0\eta_{tt} + \phi_t+g\eta,u\right)_{\Gs} + \left( \C_\rho:\nabla^2\eta,\nabla^2 u\right)_{\Gs},
\end{align}
with $ \alpha_f $ and $ \beta$ scaling parameters introduced for stability and dimensional consistency purposes, as discussed in Section \ref{sec:numerical_analysis}, defined in equations \Eq{eq:alpha_f}. Without loss of generality, we assume that there is no external loading acting on the structure, leading to
\begin{equation}\label{eq:linear}
L([w,v])\eqdef(\uin,w)_{\Gi}  + (\uout,w)_{\Go} .
\end{equation}

For the sake of completeness, we proceed with the description of the steps followed to reach the bilinear form \Eq{eq:bilinear_modified}. The first row is obtained by multiplying equation \Eq{eq:potential} against the test function, $ w $, integrating over the domain $ \Omega $, integrating by parts and replacing the normal velocity at the boundaries by the respective kinematic boundary condition as stated in \Eq{eq:kinematic_bc}. 

Following the monolithic approach described in \cite{akkerman2020isogeometric}, the last term in the first row of \Eq{eq:bilinear_modified} incorporates the dynamic boundary condition on the free surface, equation \Eq{eq:dynamic_linear_bc_f} and \Eq{eq:dynamic_linear_bc_f_3d}. Here we multiply such condition against a modified test function, $ \alpha_fw+v $, integrate over the free surface boundary $ \Gf $ and weight this contribution by a parameter, $ \beta$. The term $ \alpha_fw $ is added to the test function $ v $ to guarantee coercivity of the system, see Section \ref{subsec:continuous}. The main difference with respect to \cite{akkerman2020isogeometric} is that in the cited work the authors select $ \beta=\frac{1}{2} $ as a fixed parameter, while here we analyse the relation between $ \alpha_{f} $ and $ \beta $ that results in a stable formulation.

Similarly, the second and third terms of the second row of \Eq{eq:bilinear_modified} enforces the dynamic boundary condition on the beam surface, equation \Eq{eq:dynamic_linear_bc_s} and \Eq{eq:dynamic_linear_bc_s_3d}. In this case, however, the fourth order term is integrated by parts twice. In this process, we assume that the functions belong to $ H^2(\Omega) $, so that the rotations are continuous across the structure. In addition, the lower-dimensional integrals that appear in the integration by parts cancel when we enforce the boundary conditions \eqref{eq:zero_moment_plate}-\eqref{eq:zero_shear_plate}, \textit{i.e.}
\begin{align*}
\left((\nabla\cdot(\C_\rho:\nabla^2\eta))\cdot\nL,u\right)_{\Lfs}=0,\\
\left((\C_\rho:\nabla^2\eta)\cdot\nL,\nabla u\right)_{\Lfs}=0.
\end{align*}

\newpage
\begin{remark}[Structure with joints]
	The weak problem given in \Eq{eq:weak_form} assumes the variational space $ \V_{\Gs} $ belongs to the trace space of $ H^2(\Gs) $. However, when we have structural joints, the rotations (deflection gradients) are not continuous at the joint location. Thus, the space of functions $ \eta $ can be defined as the space of trace functions of $ H^2(\Omega) $ on $ \Gs\setminus\Lj $, with continuous surface elevation, but discontinuous gradients on $ \Lj $. %{\color{red} better way of defining space $ \V_{\Gamma} $}.
	With this definition of the space $ \V_{\Gs} $, integrating by parts the fourth order term appearing in \eqref{eq:dynamic_linear_bc_s_3d}, we have that
	\begin{align}\label{eq:by_parts_joint}
	\left(\nabla^2:(\C_\rho:\nabla^2\eta),u\right)_{\Gs}=
	&-\left(\nabla\cdot(\C_\rho:\nabla^2\eta),\nabla u\right)_{\Gs}
	+\left((\nabla\cdot(\C_\rho:\nabla^2\eta))\cdot\nL,u\right)_{\Ls}\\\nonumber
	&+\left(\jump{(\nabla\cdot(\C_\rho:\nabla^2\eta))\cdot\nL},\average{u}\right)_{\Lj}
	+\left(\average{\nabla\cdot(\C_\rho:\nabla^2\eta)},\jump{u\nL}\right)_{\Lj}\\\nonumber
	\overset{\eqref{eq:zero_shear_plate},\eqref{eq:zero_shear_jump_plate},\jump{u\n_{\Lj}}=0}{=}
	&\left(\C_\rho:\nabla^2\eta,\nabla^2 u\right)_{\Gs}-
	\left(\left(\C_\rho:\nabla^2\eta\right)\cdot\nL,\nabla u\right)_{\Ls}\\\nonumber
	&-\left(\jump{\left(\C_\rho:\nabla^2\eta\right)\cdot\nL},\average{\nabla u}\right)_{\Lj} 
	- \left(\average{\C_\rho:\nabla^2\eta},\jump{\nabla u\otimes\nL}\right)_{\Lj}\\\nonumber
	\overset{\eqref{eq:zero_moment_plate},\eqref{eq:zero_moment_jump_plate},\eqref{eq:moment_mean_plate}}{=}&\left(\C_\rho:\nabla^2\eta,\nabla^2 u\right)_{\Gs} + \left(k_\rho\jump{\nabla\eta\otimes\nL},\jump{\nabla u\otimes\nL}\right)_{\Lj}.
	\end{align}
	
	With $k_\rho\eqdef\frac{k_\varphi}{\rho_w}$. Then, the bilinear form equivalent to \eqref{eq:bilinear_modified} will read
	\begin{align}\label{eq:bilinear_modified_joints}
	B_{\scriptsize\mbox{j}}([\phi,\kappa,\eta],[w,v,u])\eqdef B([\phi,\kappa,\eta],[w,v,u])+\left(k_\rho\jump{\nabla\eta\otimes\nL},\jump{\nabla u\otimes\nL}\right)_{\Lj}.
	\end{align}
\end{remark}

Since \eqref{eq:bilinear_modified} is a particular case of \eqref{eq:bilinear_modified_joints} for $ \Lj=\emptyset $, hereinafter we will use the later bilinear form. 

\subsection{Spatial discretization}
Let us consider a FE partition $\Omega_h$ of the domain $\Omega$ from which we can construct conforming finite dimensional spaces for the potential $\V_{h} \subset \V$, for the surface elevation at the free surface $\V_{\Gf,h}\subset \V_{\Gf}$ and for the surface elevation at the structure $\V_{\Gs,h}\subset \V_{\Gs}$. We denote by $ \mathcal{E}_h $ the set of facets (entities of one dimension lower than the dimension of $ \Omega $) generated by the FE partition $\Omega_h$. We define as $ \Gamma_h $ the set of facets of $ \mathcal{E}_h $ that lie on the boundary of $ \Omega $, $ \partial\Omega $, i.e. $ \Gamma_h\eqdef\mathcal{E}_h\cap\partial\Omega $. Following this notation, we also define the discrete boundary parts $ \Gbh\eqdef\mathcal{E}_h\cap\Gb $, $ \Gih\eqdef\mathcal{E}_h\cap\Gi $, $ \Goh\eqdef\mathcal{E}_h\cap\Go $, $ \Gfh\eqdef\mathcal{E}_h\cap\Gf $ and $ \Gsh\eqdef\mathcal{E}_h\cap\Gs $. In addition, we denote by $ \Lsh $ the set of edges or points between facets of $ \Gsh $ that are not joints and do not belong to the boundary of $\Gs $, see Figure~\ref{fig:domain_discretized}. We assume that in the case that the structure has one or multiple joints, these will lie in an edge or point between the FE partition facets $ \mathcal{E}_h $.
\begin{figure}[pos=h!]
	\centering
	\includegraphics[width=0.6\textwidth]{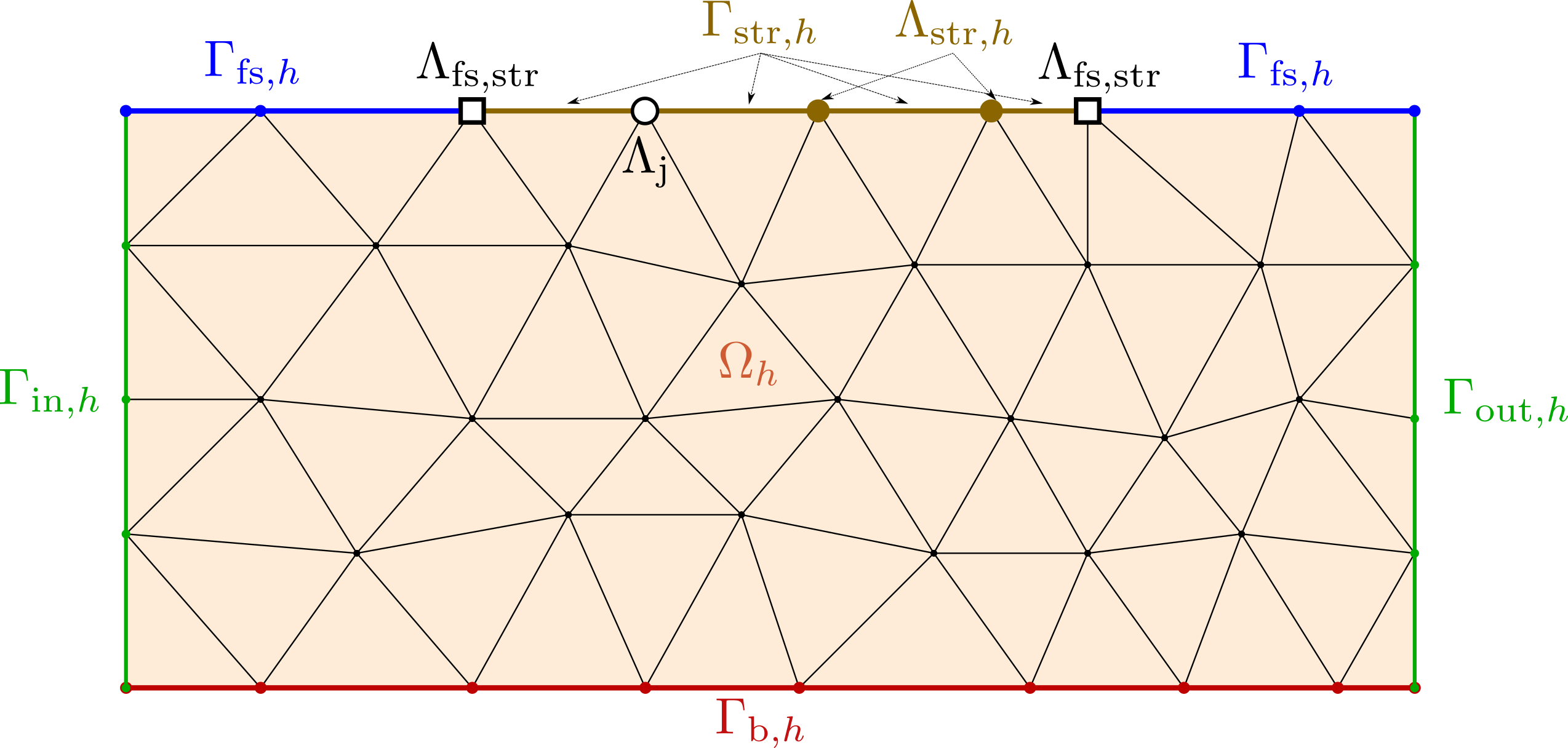}
	\caption{Sketch with the definition of the discrete geometrical entities.}	
	\label{fig:domain_discretized}
\end{figure}

Using this notation, the Galerkin FE formulation equivalent to \Eq{eq:weak_form} reads: find $ [\phi_h,\kappa_h,\eta_h] \in\V_h\times\V_{\Gf,h}\times\V_{\Gs,h}$ such that
\begin{equation}\label{eq:galerkin_form}
	B_h([\phi_h,\kappa_h,\eta_h],[w_h,v_h,u_h])=L_h([w_h,v_h,u_h])\quad\forall[w_h,v_h,u_h]\in\V_h\times\V_{\Gf,h}\times\V_{\Gs,h},
\end{equation}
where the bilinear form is given by
%\begin{align}\label{eq:bilinear_modified_h}
%	B_h([\phi_h,\eta_h],[w_h,v_h]):=&(\nabla\phi_h,\nabla w_h)_{\Omega_h} - (\eta_{h,t},w_h)_{\Gfh\cup\Gsh}\\\nonumber
%	+&\beta_f\left(\phi_{h,t}+g\eta_h,\alpha_fw_h+v_h\right)_{\Gfh} \\\nonumber
%	+&\beta_b\left(\frac{\rho_bh_b}{\rho_w}\eta_{h,tt} + \phi_{h,t}+g\eta_h,\alpha_bw_h+v_h\right)_{\Gsh} + \beta_b\frac{D}{\rho_w}\left( \Delta\eta_h,\alpha_b\Delta w_h+\Delta v_h\right)_{\Gsh}\\\nonumber
%	+&2\beta_b\alpha_b\frac{D}{\rho_w}(\Delta\phi_h,\Delta w_h)_{\Gsh},
%\end{align}
\newpage
\begin{align}\label{eq:bilinear_modified_h}
	B_h([\phi_h,\kappa_h,\eta_h],[w_h,v_h,u_h])\eqdef
	&(\nabla\phi_h,\nabla w_h)_{\Omega_h} - (\kappa_{h,t},w_h)_{\Gfh} 
	+ \beta\left(\phi_{h,t}+g\kappa_h,\alpha_fw_h+v_h\right)_{\Gfh}\\\nonumber
	-&(\eta_{h,t},w_h)_{\Gsh} 
	+\left(d_0\eta_{h,tt} + \phi_{h,t}+g\eta_h,u_h\right)_{\Gsh} \\\nonumber
	+&\left( \C_\rho:\nabla^2\eta_h,\nabla^2 u_h\right)_{\Gsh} +\left(k_\rho\jump{\nabla\eta_h\otimes\nL},\jump{\nabla u_h\otimes\nL}\right)_{\Lj},
\end{align}
and the linear right hand side is
\begin{equation}\label{eq:linear_h}
	L_h([w_h,v_h])\eqdef(\uin,w_h)_{\Gih}  + (\uout,w_h)_{\Goh} .
\end{equation}

Note that the formulation defined by \Eq{eq:galerkin_form}-\Eq{eq:linear_h} assumes that the FE spaces are defined in $ H^2(\Gs\setminus\Lj) $, which requires $ \mathcal{C}^1 $ continuity across elements in the structure, i.e continuous gradients between elements. This condition is satisfied by certain FE types such as the Morley~\cite{morley1971constant}, Argyris~\cite{argyris1968tuba} or NURBS-based FE~\cite{cottrell2006isogeometric}. Nonetheless, we propose an alternative formulation that can be generalized to $ H^1(\Gs) $ FE defined, for instance, by Lagrange polynomials. This formulation is based on a Continuous/Discontinuous Galerkin (C/DG) approach for fourth order operators, \cite{engel2002continuous}, where the discrete functions are continuous at the element nodes, but the gradient is discontinuous. The continuity of first order derivatives is weakly enforced via an interior penalty approach.

The resulting C/DG formulation reads:  find $ [\phi_h,\kappa_h,\eta_h] \in\hatV_h\times\hatV_{\Gf,h}\times\hatV_{\Gs,h}$ such that
\begin{equation}\label{eq:cgdg_form}
	\hat{B}_h([\phi_h,\kappa_h,\eta_h],[w_h,v_h])=L_h([w_h,v_h,u_h])\quad\forall[w_h,v_h,u_h]\in\hatV_h\times\hatV_{\Gf,h}\times\hatV_{\Gs,h},
\end{equation}
where
%\begin{align}\label{eq:bilinear_modified_h_cgdg}
%\hat{B}_h([\phi_h,\eta_h],[w_h,v_h])\eqdef&(\nabla\phi_h,\nabla w_h)_{\Omega_h} - (\eta_{h,t},w_h)_{\Gfh\cup\Gsh}\\\nonumber
%+&\beta\left(\phi_{h,t}+g\eta_h,\alpha_fw_h+v_h\right)_{\Gfh} +\left(d_0\eta_{h,tt} + \phi_{h,t}+g\eta_h,v_h\right)_{\Gsh}\\\nonumber
%+&\left( \Drho\Delta\eta_h,\Delta v_h\right)_{\Gsh} -\left( \average{\Drho\Delta\eta_h},\jump{\nabla v_h\cdot\n_{\Lsh}}\right)_{\Lsh} - \left(\jump{\Drho\nabla\eta_h\cdot\n_{\Lsh}},\average{\Delta v}\right)_{\Lsh}\\\nonumber +&\frac{\gamma}{h}\left(\jump{\Drho\nabla\eta_h\cdot\n_{\Lsh}},\jump{\nabla v_h\cdot\n_{\Lsh}}\right)_{\Lsh} + \left(k_\rho\jump{\nabla\eta_h\cdot\n_{\Lj}},\jump{\nabla v_h\cdot\n_{\Lj}}\right)_{\Lj}.
%\end{align}
\begin{align}\label{eq:bilinear_modified_h_cgdg}
\hat{B}_h([\phi_h,\kappa_h,\eta_h],[w_h,v_h,u_h])\eqdef
&B_h([\phi_h,\kappa_h,\eta_h],[w_h,v_h,u_h])\\\nonumber
-&\left( \average{\C_\rho:\nabla^2\eta_h},\jump{\nabla u_h\otimes\nL}\right)_{\Lsh} 
- \left(\jump{\nabla\eta_h\otimes\nL},\average{\C_\rho:\nabla^2u_h}\right)_{\Lsh}\\\nonumber
+&\frac{\gamma\hat{D}_{\rho}}{h}\left(\jump{\nabla\eta_h\otimes\nL},\jump{\nabla u_h\otimes\nL}\right)_{\Lsh},
\end{align}
With $ h $ the characteristic element size, $ \hat{D}_{\rho}\eqdef\frac{EI}{(1-\nu^2)} $ a constant that depends on the material properties and $ \gamma $ a constant that can be tuned to guaranty stability, see Section~\ref{sec:numerical_analysis}. The first term of the second row in \Eq{eq:bilinear_modified_h_cgdg} is required for consistence. It originates from the integration by parts of the fourth order term in the plate equation after assuming continuity of test functions across boundary elements and weakly enforcing
\begin{subequations}\label{eq:weak_continuity}
\begin{align}
	\jump{(\nabla\cdot(\C_\rho:\nabla^2\eta_h))\cdot\nL}=0&\qquad\mbox{on $ \Lsh $},\\
	\jump{(\C_\rho:\nabla^2\eta_h)\cdot\nL}=0&\qquad\mbox{on $ \Lsh $},
\end{align}
\end{subequations} 
\textit{i.e.} continuity of shear forces and moments across structure elements. 
The second term in the second row enforces dual consistency, and improves overall accuracy of the method. 
It penalizes the jump of surface elevation gradients across structural elements. 
The term appearing in the last row of \Eq{eq:bilinear_modified_h_cgdg} is added for stability purposes, while also penalizing the jump of gradients across the structural elements.

Hereinafter we will work with the formulation defined by equations \Eq{eq:cgdg_form}-\Eq{eq:bilinear_modified_h_cgdg}. The variational spaces $ \hat{\V}_h $, $ \hatV_{\Gf,h} $ and $ \hatV_{\Gs,h} $ will be given by
\begin{align}
\hatV_h&\eqdef\left\{w_h\in\mathcal{C}^0(\Omega):\ w_h\vert_K\in\mathbb{P}_r(K),\forall K\in\Omega_h\right\},\\
\hatV_{\Gf,h}&\eqdef\left\{w_h\vert_E:\ w_h\in\hatV_h,\forall E\in\Gf\right\},\\
\hatV_{\Gs,h}&\eqdef\left\{w_h\vert_E:\ w_h\in\hatV_h,\forall E\in\Gsh\right\},
\end{align}
where $ \mathbb{P}_r(K) $ is the space of Lagrange polynomials of degree $ r\geq 2 $ in an element $ K $.

\subsection{Time discretization}
The forms given in \Eq{eq:galerkin_form} or \Eq{eq:cgdg_form} describe a semi-discrete system of equations, i.e. discrete in space and continuous in time, resulting in a system of 2nd order ordinary differential equations (ODEs). In this work we consider two approaches to find the solution to the transient problem:  a frequency domain  and a time domain approach. 

\subsubsection{Frequency domain}\label{subsubsec:freq_domain}

The analysis of free surface flows in the Airy wave theory, described by linear potential flow, is suitable for a frequency domain formulation. This is also the case of dynamic analysis of linear structural response. When defining the formulation in the frequency domain, the response, \textit{i.e} surface elevation and velocity potential, is assumed to be harmonic. That is, a space and time dependent variable, $ \xi(\x,t) $, can be given in terms of a prescribed frequency, $ \omega $, and a time-independent variable, $ \bar{\xi}(\x) $ as
\begin{equation}\label{eq:freq_transform}
	\xi(\x,t)=\bar{\xi}(\x)\exp(-i\omega t).
\end{equation}

Under this assumption, given an incoming wave frequency, $ \omega $, the fully discrete problem reduces to find the set of time-independent and complex-valued fields, in our case $ \bar{\phi}(\x) $, $ \bar{\kappa}(\x) $ and $ \bar{\eta}(\x) $. To simplify notation, in this subsection we will neglect the bar and hereinafter assume that $ \eta\eqdef\bar{\eta} $, $ \kappa\eqdef\bar{\kappa} $ and $ \phi\eqdef\bar{\phi} $. Hence, the final discrete problem in the frequency domain will read: find $ [\phi_h,\kappa_h,\eta_h] \in\hatV^\omega_h\times\hatV^\omega_{\Gf,h}\times\hatV^\omega_{\Gs,h}$ such that
\begin{equation}\label{eq:cgdg_form_freq_domain}
\hat{B}^\omega_h([\phi_h,\kappa_h,\eta_h],[w_h,v_h,u_h])=L^\omega_h([w_h,v_h,u_h])\quad\forall[w_h,v_h,u_h]\in\hatV^\omega_h\times\hatV^\omega_{\Gf,h}\times\hatV^\omega_{\Gs,h},
\end{equation}
with
\begin{align}\label{eq:bilinear_modified_h_cgdg_freq_domain}
\hat{B}^\omega_h([\phi_h,\kappa_h,\eta_h],[w_h,v_h,u_h])\eqdef
&(\nabla\phi_h,\nabla w_h)_{\Omega_h} 
+ (i\omega\kappa_{h},w_h)_{\Gfh}
+\beta\left(g\kappa_h-i\omega\phi_{h},\alpha^\omega_fw_h+v_h\right)_{\Gfh} \\\nonumber
+&(i\omega\eta_{h},w_h)_{\Gsh}
+\left((g-\omega^2d_0)\eta_{h} - i\omega\phi_{h},u_h\right)_{\Gsh}\\\nonumber
+&\left( \C_\rho:\nabla^2\eta_h,\nabla^2 u_h\right)_{\Gsh} 
+ \left(k_\rho\jump{\nabla\eta_h\otimes\nL},\jump{\nabla u_h\otimes\nL}\right)_{\Lj} \\\nonumber
-&\left( \average{\C_\rho:\nabla^2\eta_h},\jump{\nabla u_h\otimes\nL}\right)_{\Lsh} 
- \left(\jump{(\nabla\eta_h\otimes\nL},\average{\C_\rho:\nabla^2 u_h}\right)_{\Lsh}\\\nonumber
+&\frac{\gamma\hat{D}_{\rho}}{h}\left(\jump{\nabla\eta_h\otimes\nL},\jump{\nabla u_h\otimes\nL}\right)_{\Lsh} ,
\end{align}
and
\begin{equation}\label{eq:linear_h_freq_domain}
L^\omega_h([w_h,v_h,u_h])\eqdef L_h([w_h,v_h,u_h]) .
\end{equation}

Here, the stabilization parameter $ \alpha^\omega_f $ appearing in \eqref{eq:bilinear_modified_h_cgdg_freq_domain} is defined as $ \alpha^\omega_f \eqdef\frac{-i\omega}{g}\frac{1-\beta}{\beta} $, see Section~\ref{sec:numerical_analysis} for further details on the justification of this definition.

It is important to highlight that the FE spaces $ \hatV^\omega_h $  and $\hatV^\omega_{\Gamma,h}$ are finite-dimensional spaces of complex-valued functions, composed by real-valued shape functions and complex-valued degrees of freedom.
%composed by real-valued shape functions, $ \zeta_i:\Omega_h\longrightarrow\mathbb{R} $, and complex-valued degrees of freedom, $ \eta_{h,i}\in\mathbb{C} $ and  $ \phi_{h,i}\in\mathbb{C} $, for all $ i=1,...,N_{\scriptsize\mbox{DOF}} $. Assuming that we use the same functions $ \zeta_i $ for both velocity potential and surface elevation spaces, the FE functions would read as follows
%\begin{align*}
%\phi_h(\x)=\sum_{i=1}^{N_{\scriptsize\mbox{DOF}}}\zeta_i(\x)\phi_{h,i},\qquad
%\kappa_h(\x)=\sum_{i=1}^{N_{\scriptsize\mbox{DOF}}}\zeta_i(\x)\kappa_{h,i},\qquad
%\eta_h(\x)=\sum_{i=1}^{N_{\scriptsize\mbox{DOF}}}\zeta_i(\x)\eta_{h,i}.
%\end{align*}

\subsubsection{Time domain}\label{subsubsec:time_domain}

	As an alternative to the frequency domain approach, instead of assuming an harmonic response, one can discretize in time the semi-discrete system given by equation~\Eq{eq:cgdg_form}. Here, since we have a second-order ODE, we use the so called Newmark-beta time discretization scheme~\cite{newmark1959method}. Let us consider a uniform discretization of the time domain, with a constant time step size $ \Delta t $. At a given time step $ n+1 $, with $ t^{n+1}=t^n+\Delta t $, an unknown, $ x^{n+1} $, and its first time derivative,   are defined by
\begin{align}
\label{eq:newmar_velocity}
x_t^{n+1} &\eqdef x_t^{n}+\Delta t\left[(1-\gammaNB)x_{tt}^n+\gammaNB x_{tt}^{n+1}\right],\\
\label{eq:newmar_unknown}
x^{n+1} &\eqdef x^{n}+\Delta t x_t^{n}+\Delta t^2\left[\left(\frac{1}{2}-\betaNB\right)x_{tt}^n+\betaNB x_{tt}^{n+1}\right],
\end{align}
where $ \gammaNB $ and $ \betaNB $ are two coefficients that determine the stability and accuracy of the scheme~\cite{goudreau1973evaluation}. In this work we will use the pair $ \gammaNB=0.5 $ and $ \betaNB=0.25 $, which results in a second order accurate and unconditionally stable scheme.

Doing some manipulations, we can obtain an expression for the first and second time derivatives at $ n+1 $. 
These derivatives  depend on the unknown solution $ x^{n+1} $, and the known solution and  derivatives, $ \left\{x^n,x_t^{n},x_{tt}^n\right\}  $,
\begin{align}
\label{eq:newmark_velocity_2}
	x_t^{n+1} &= \delta_t \left(x^{n+1}-x^n\right) + \frac{1-\gammaNB}{\betaNB}x_t^n + \Delta t \frac{1-\gammaNB}{2\betaNB}x_{tt}^n,\\
	\label{eq:newmark_acceleration}
	x_{tt}^{n+1} &= \delta_{tt}\left(x^{n+1}-x^n\right) -\frac{1}{\betaNB\Delta t}x_t^n + \frac{1-2\betaNB}{2\betaNB}x_{tt}^n
\end{align} 
where $ \delta_t=\frac{\gammaNB}{\betaNB\Delta t} $ and $ \delta_{tt}=\frac{1}{\betaNB\Delta t^2} $. 

Using the time discretization given by equations \Eq{eq:newmark_velocity_2}-\Eq{eq:newmark_acceleration} for the potential and surface elevation time derivatives, we obtain the following fully discrete problem in the time domain: find $ [\phi_h^{n+1},\kappa_h^{n+1},\eta_h^{n+1}] \in\hatV_h\times\hatV_{\Gf,h}\times\hatV_{\Gs,h}$ such that
\begin{equation}\label{eq:cgdg_t_form}
\hat{B}_h^{n+1}([\phi_h^{n+1},\kappa_h^{n+1},\eta_h^{n+1}],[w_h,v_h,u_h])=L_h^{n+1}([w_h,v_h,u_h])\quad\forall[w_h,v_h,u_h]\in\hatV_h\times\hatV_{\Gf,h}\times\hatV_{\Gs,h},
\end{equation}
where
\begin{align}\label{eq:bilinear_modified_h_t_cgdg}
\hat{B}_h^{n+1}([\phi_h^{n+1},\kappa_h^{n+1}\eta_h^{n+1}],[w_h,v_h,u_h])\eqdef
&(\nabla\phi_h^{n+1},\nabla w_h)_{\Omega_h} 
- (\delta_t\kappa_{h}^{n+1},w_h)_{\Gfh}\\\nonumber
+&\beta\left(\delta_t\phi_{h}^{n+1}+g\kappa_h^{n+1},\alpha_fw_h+v_h\right)_{\Gfh}\\\nonumber
-&(\delta_t\eta_{h}^{n+1},w_h)_{\Gsh}
+\left(\delta_{tt}d_0\eta_{h}^{n+1} + \delta_t\phi_{h}^{n+1}+g\eta_h^{n+1},u_h\right)_{\Gsh} \\\nonumber 
+&\left( \C_\rho:\nabla^2\eta_h^{n+1},\nabla^2 u_h\right)_{\Gsh}
+\left(k_\rho\jump{\nabla\eta_h^{n+1}\otimes\nL},\jump{\nabla u_h\otimes\nL}\right)_{\Lj}\\\nonumber
-&\left( \average{\C_\rho:\nabla^2\eta_h^{n+1}},\jump{\nabla u_h\otimes\nL}\right)_{\Lsh} \\\nonumber
-&\left(\jump{\nabla\eta_h^{n+1}\otimes\n_{\Lsh}},\average{\C_\rho:\nabla^2 u_h}\right)_{\Lsh}\\\nonumber
+&\frac{\gamma\hat{D}_{\rho}}{h}\left(\jump{\nabla\eta_h^{n+1}\otimes\nL},\jump{\nabla u_h\otimes\nL}\right)_{\Lsh},
\end{align}
and
\begin{align}\label{eq:linear_h_t}
L_h^{n+1}([w_h,v_h,u_h])\eqdef&(\uin,w_h)_{\Gih}  + (\uout,w_h)_{\Goh} \\\nonumber
 -&\left(\delta_t\kappa_{h}^n - \frac{1-\gammaNB}{\betaNB}\kappa_{h,t}^n - \Delta t \frac{1-\gammaNB}{2\betaNB}\kappa_{h,tt}^n,w_h\right)_{\Gfh}\\\nonumber
+&\beta\left(\delta_t\phi_{h}^n- \frac{1-\gammaNB}{\betaNB}\phi_{h,t}^n - \Delta t \frac{1-\gammaNB}{2\betaNB}\phi_{h,tt}^n,\alpha_fw_h+v_h\right)_{\Gfh} \\\nonumber
-&\left(\delta_t\eta_{h}^n - \frac{1-\gammaNB}{\betaNB}\eta_{h,t}^n - \Delta t \frac{1-\gammaNB}{2\betaNB}\eta_{h,tt}^n,w_h\right)_{\Gsh}\\\nonumber
+&\left(d_0\left(\delta_{tt}\eta_h^n+ \frac{1}{\betaNB\Delta t}\eta_{h,t}^n - \frac{1-2\betaNB}{2\betaNB}\eta_{h,tt}^n\right),u_h\right)_{\Gsh} \\\nonumber
+&\left(\delta_t\phi_{h}^n- \frac{1-\gammaNB}{\betaNB}\phi_{h,t}^n - \Delta t \frac{1-\gammaNB}{2\betaNB}\phi_{h,tt}^n,u_h\right)_{\Gsh}.
\end{align}

In equation \eqref{eq:bilinear_modified_h_t_cgdg}, the stabilization parameter is defined as $ \alpha_f\eqdef\frac{\delta_t}{g}\frac{(1-\beta)}{\beta} $ for the system to be stable, see Section~\ref{sec:numerical_analysis}.

\section{Numerical analysis}\label{sec:numerical_analysis}

In this section we prove statements of consistency, coercivity, boundedness and energy conservation for the discrete formulation proposed in this work. In Section~\ref{subsec:continuous} we will first demonstrate these properties for the formulation without discontinuities, i.e. equation~\eqref{eq:galerkin_form}. After, in Section~\ref{subsec:discontinuous}, we will extend the analysis to include the C/DG formulation~~\eqref{eq:cgdg_form}.

%\begin{itemize}
%	\item definition of a norm
%	\item coercivity
%	\item boundedness
%	\item Cea
%	\item error estimate	
%	\item Energy conservation
%\end{itemize}

\subsection{Preliminary definitions and theorems}\label{ssec:prelim}
Let us establish some definitions and theorems that will be later used in the numerical analysis of the formulation.

%\begin{definition} For ease of notation, we define the notation, 	 
%    \begin{equation}\label{eq:norm_domain}
%	 \left\| \cdot \right\|_{L_2(X)} = 	 \left\| \cdot \right\|_X. 
%	 \end{equation}
%\end{definition}

\begin{definition}
	For any $w$ we define the $ H^1 $-norm in $ \Omega $ as 
	 \begin{equation}\label{eq:H1_def}
		\left\| w \right\|^2_{H^1(\Omega)} =  \left\| w \right\|^2_{\Omega} + \left\| \nabla w \right\|^2_{\Omega} .
	\end{equation}
\end{definition}

\begin{corollary} For $w\in H^1(\Omega)$ we can bound the gradient. That is,
\begin{equation}\label{eq:H1_equiv}
\|\nabla w\|_{\Omega_h}\leq\|w\|_{H^1(\Omega_h)},\qquad\forall w\in H^1(\Omega).
%\|\cdot\|_{\Omega_h}\leq\|\cdot\|_{H^1(\Omega_h)},\qquad\forall w\in H^1(\Omega).
\end{equation}
\begin{proof}
The statement is a direct consequence of the definition of the $H_1{\Omega}$-norm.
\end{proof}
\end{corollary}

\begin{definition}
	For any $w\in H^1(\Omega)$ define by extension the trace operator $ \gamma_{\partial\Omega}:H^1(\Omega)\rightarrow L^2(\partial\Omega) $ such that  
	 \begin{equation}\label{eq:trace}
		\gamma_{\partial\Omega} w=w|_{\partial\Omega},\forall w \in C^\infty(\Omega).
	\end{equation}
\end{definition}

\begin{theorem}[Trace theorem of Sobolev spaces]\label{theorem:trace} %1
	Let $ \Omega $ be a bounded simply connected Lipschitz domain. Then, the trace operator $ \gamma_{\partial\Omega} $ is a bounded linear operator from $ H^1(\Omega) $ to $ L^2(\partial\Omega) $. That is, 
	\begin{equation}\label{eq:sob_trace}
	 \left\|\gamma_{\partial\Omega}w\right\|_{\partial\Omega}\leq C_{\partial\Omega} \left\|w\right\|_{H^1(\Omega)}. 
	 	\end{equation}
	With $ C_{\partial\Omega} $ a constant that only depends on $ \partial\Omega $.
\end{theorem}
See \cite{ding1996proof} for a proof of Theorem~\ref{theorem:trace}.
\begin{theorem}\label{theorem:generalized_poincare} %2
	Let $ \Omega $ be a bounded connected Lipschitz domain and $ f $ be a linear form from $ H^1(\Omega) $ 
	with a non-zero restriction on non-zero constant functions.
	%whose restriction on constant functions is not zero. 
	Then, there is a constant $ C_\Omega > 0 $ such that 
	\begin{equation}\label{key}
		C_\Omega\|w\|_{H^1(\Omega)}\leq \|\nabla w\|_{\Omega} + |f(w)|, \qquad\forall w\in H^1(\Omega).
	\end{equation}
\end{theorem}

Let us define the function $ f $ appearing in Theorem~\ref{theorem:generalized_poincare} as 
\[ f(w) \eqdef \alpha_{f}\beta\delta_t \|\gamma_{\Gf}w\|_{\Gf}  , \] 
with $ \gamma_{\Gf} $ the trace operator as defined in equation~\eqref{eq:trace}. Note that $ f $ is a linear form on $ H^1(\Omega) $ and its restriction on non-zero constant functions is non-zero, which holds for any $ \alpha_{f},\beta,\delta_t>0 $ and any open boundary portion $ \Gf\subseteq\partial\Omega $ with non-zero measure, \textit{i.e.} $ |\Gf|\eqdef \mbox{meas}(\Gf)>0 $. Then, we have that 
\begin{equation}\label{eq:f_bound}
	|f(w)| = \alpha_f\beta\delta_t  \|\gamma_{\Gf}w\|_{\Gf} =\alpha_f\beta\delta_t\|w\|_{\Gf}, \qquad\forall w\in H^1(\Omega).
\end{equation}

\begin{corollary} For $w\in H^1(\Omega)$ we can bound the $H^1(\Omega)$-norm as follows,
\begin{equation}\label{eq:seminorm_bound}
 C_\Omega\|w\|_{H^1(\Omega)}\leq
	\|\nabla w\|_{\Omega}+\alpha_f\beta\delta_t\|w\|_{\Gf},\qquad\forall w\in H^1(\Omega).
\end{equation}
\begin{proof}
Introducing~\eqref{eq:f_bound} into Theorem~\ref{theorem:generalized_poincare} proofs the statement.
\end{proof}
\end{corollary}

\begin{theorem}\label{theorem:trace_inequality} %3
	Let $ \Omega $ be a bounded simply connected Lipschitz domain and $ \Omega_e $ an element of the FE triangulation of $\Omega$ with characteristic element size $ h_e $. Then there is a constant $C_I$ such that 
	\begin{equation}
		\|w\|^2_{\partial\Omega_e}\leq C_I\left(h_e^{-1}\|w\|^2_{\Omega_e}+h_e\|\nabla\|^2_{\Omega_e}\right).
	\end{equation}
	See \cite{arnold1982interior} for more details.
\end{theorem}

\subsection{Continuous formulation}\label{subsec:continuous}
We start this section by demonstrating the consistency and energy conservation of the semi-discrete form~\eqref{eq:galerkin_form}. That is, we assume that we use a set of FE spaces such that $ \eta_h\in H^2(\Gs\setminus\Lj) $, i.e. functions with continuous gradients.

\begin{proposition}[Consistency]\label{prop:consistency}
	The semi-discrete problem~\eqref{eq:galerkin_form} is consistent. That is, the exact solution $ [\phi,\kappa,\eta]\in\V\times\V_{\Gf}\times\V_{\Gs}$ satisfies the approximate problem
	\begin{equation}\label{eq:galerkin_form_exact}
		B_h([\phi,\kappa,\eta],[w_h,v_h,u_h])=L_h([w_h,v_h,u_h])\quad\forall[w_h,v_h,u_h]\in\V_h\times\V_{\Gf,h}\times\V_{\Gs,h}.
	\end{equation}
\end{proposition}
\begin{proof}
	The consistency statement results from integrating by parts on each element the terms $ (\nabla\phi,\nabla w_h)_{\Omega_h} $ and $ \left( \C_\rho:\nabla^2\eta,\nabla^2 u_h\right)_{\Gsh} $ appearing in \eqref{eq:bilinear_modified_h}, using the strong form of the equations \eqref{eq:potential} and \eqref{eq:plate_eqn}, and boundary conditions \eqref{eq:kinematic_bc}, \eqref{eq:dynamic_bc} and \eqref{eq:boundary_conditions_plate}.
%	we have
%	\begin{align}\label{eq:consistency}
%		&B_h([\phi,\kappa,\eta],[w_h,v_h,u_h])-L_h([w_h,v_h,u_h])=\\\nonumber
%		& (-\Delta\phi,w_h)_{\Omega_h} 
%		+ (\nabla\phi\cdot\n - \uin,w_h)_{\Gih} 
%		+ (\nabla\phi\cdot\n,w_h)_{\Gbh} 
%		+ (\nabla\phi\cdot\n - \uout,w_h)_{\Goh} \\\nonumber
%		& + (\nabla\phi\cdot\n - \kappa_t,w_h)_{\Gfh}
%		+ (\nabla\phi\cdot\n - \eta_t,w_h)_{\Gsh}
%		+\beta\left(\phi_t+g\kappa,\alpha_fw_h+v_h\right)_{\Gfh}\\\nonumber
%		%
%		& +\left(d_0\eta_{tt} + \phi_t+g\eta,u_h\right)_{\Gsh}  
%		+\left( \nabla^2:\C_\rho:\nabla^2\eta,u_h\right)_{\Gsh}\\\nonumber
%		+&\left( \jump{\left(\nabla\cdot(\C_\rho:\nabla^2\eta)\right)\cdot\nL},u_h\right)_{\Lsh\cup\Lj}
%		+\left( \jump{\left(\C_\rho:\nabla^2\eta\right)\cdot\nL},\nabla u_h\right)_{\Lsh}
%		+\left( \jump{\left(\C_\rho:\nabla^2\eta\right)\cdot\nL},\average{\nabla u_h}\right)_{\Lj}
%		+\left( \average{\C_\rho:\nabla^2\eta},\jump{\nabla u_h\otimes\nL}\right)_{\Lj}\\\nonumber
%		+&\left(k_\rho\jump{\nabla\eta_h\otimes\nL},\jump{\nabla u_h\otimes\nL}\right)_{\Lj}
%	\end{align}
\end{proof}

\begin{proposition}[Energy conservation]\label{prop:energy}
	The semi-discrete problem~\eqref{eq:galerkin_form} is energy conserving for any $ \beta $ such that $ 0<\beta<1$. That is, 
	\begin{equation}\label{eq:energy}
		\frac{d\Et}{dt} = 0.
	\end{equation}
	With
	\begin{align}
		\Et\eqdef&\Ekf+\Epf+\Eks+\Ees,\\
		\Ekf\eqdef&\frac{1}{2}\left\|\nabla\phi\right\|^2_\Omega,\\
		\Epf\eqdef&\frac{g}{2}\left(\left\|\kappa\right\|^2_{\Gf}+\left\|\eta\right\|^2_{\Gs}\right),\\
		\Eks\eqdef&\frac{1}{2}\left\|d_0^{1/2}\eta_t\right\|^2_{\Gs},\\
		\Ees\eqdef&\frac{1}{2}\left\|\C^{1/2}:\nabla^2\eta\right\|^2_{\Gs}+\frac{1}{2}\left\|k^{1/2}_\rho\jump{\nabla\eta_h\otimes\nL}\right\|^2_{\Lj}.
	\end{align}
\end{proposition}

\begin{proof}
	Let us select the set of test functions as $ [w_h,v_h,u_h]=\left[\phi_{h,t},\frac{1}{\beta}\kappa_{h,t}-\alpha_f\phi_{h,t},\eta_{h,t}\right]\in\V_h\times\V_{\Gf,h}\times\V_{\Gs,h} $. Introducing them into \eqref{eq:galerkin_form} we have that, for $ \uin=\uout=0 $, the following statement holds
	\begin{align}
		0=&B_h([\phi_h,\kappa_h,\eta_h],[\phi_{h,t},\frac{1}{\beta}\kappa_{h,t}-\alpha_f\phi_{h,t},\eta_{h,t}])-L_h([\phi_{h,t},\frac{1}{\beta}\kappa_{h,t}-\alpha_f\phi_{h,t},\eta_{h,t}])\\\nonumber
		=&(\nabla\phi_h,\nabla \phi_{h,t})_{\Omega_h} - (\kappa_{h,t},\phi_{h,t})_{\Gfh} 
		+ \beta\left(\phi_{h,t}+g\kappa_h,\alpha_f\phi_{h,t}+\frac{1}{\beta}\kappa_{h,t}-\alpha_f\phi_{h,t}\right)_{\Gfh}\\\nonumber
		&-(\eta_{h,t},\phi_{h,t})_{\Gsh} 
		+\left(d_0\eta_{h,tt} + \phi_{h,t}+g\eta_h,\eta_{h,t}\right)_{\Gsh} \\\nonumber
		&+\left( \C_\rho:\nabla^2\eta_h,\nabla^2 \eta_{h,t}\right)_{\Gsh} +\left(k_\rho\jump{\nabla\eta_h\otimes\nL},\jump{\nabla \eta_{h,t}\otimes\nL}\right)_{\Lj}\\\nonumber
		=&\frac{1}{2}\frac{d}{dt}\left\|\nabla\phi_h\right\|^2_{\Omega_h} 
		+ \frac{g}{2}\frac{d}{dt}\left\|\kappa_h\right\|^2_{\Gfh}
		+\frac{d_0}{2}\frac{d}{dt}\left\|\eta_{h,t}\right\|^2_{\Gsh}+\frac{g}{2}\frac{d}{dt}\left\|\eta_h\right\|^2_{\Gsh} \\\nonumber
		&+\frac{1}{2}\frac{d}{dt}\left\|\C^{1/2}:\nabla^2\eta\right\|^2_{\Gs}+\frac{1}{2}\frac{d}{dt}\left\|k^{1/2}_\rho\jump{\nabla\eta_h\otimes\nL}\right\|^2_{\Lj}=\frac{d\Et}{dt}.
	\end{align}
Note, that the selection  for $v_h$ imposes a mild compatibility requirement on the discretization spaces in order for the proof to hold, viz.
$\gamma_{\Gf}(\V) \in \V_{\Gf}$.

\end{proof}

Let us now consider the fully discrete problem in the time domain given by equation~\eqref{eq:cgdg_t_form}. To simplify notation we will omit the super-index related to the time step $ (\cdot)^{n+1} $. We also note that the analysis is done for the fully discrete formulation in the time domain, but the same derivations also hold for the frequency domain. In the later case, instead of the constants $ \delta_t $ and $ \delta_{tt} $, we have $ -i\omega $ and $ -\omega^2 $, respectively. The fully discrete bilinear form in time domain for the CG case is given by
%\newpage
\begin{align}\label{eq:fully_discrete_b_cg}
	B_h([\phi_h,\kappa_h,\eta_h],[w_h,v_h,u_h])\eqdef
	&(\nabla\phi_h,\nabla w_h)_{\Omega_h} 
	- (\delta_t\kappa_{h},w_h)_{\Gfh}
	+\beta\left(\delta_t\phi_{h}+g\kappa_h,\alpha_fw_h+v_h\right)_{\Gfh}\\\nonumber
	-&(\delta_t\eta_{h},w_h)_{\Gsh}
	+\left(\delta_{tt}d_0\eta_{h} + \delta_t\phi_{h}+g\eta_h,u_h\right)_{\Gsh}  
	+\left( \C_\rho:\nabla^2\eta_h,\nabla^2 u_h\right)_{\Gsh}\\\nonumber
	+&\left(k_\rho\jump{\nabla\eta_h\otimes\nL},\jump{\nabla u_h\otimes\nL}\right)_{\Lj}.
\end{align}

Let us define the stabilization parameter $ \alpha_f $ as
\begin{equation}\label{eq:alpha_f}
	\alpha_f\eqdef\frac{(1-\beta)\delta_t}{\beta g}.
\end{equation}
This choice for $ \alpha_f $ ensures in the coercivity proof  the second term in Eq (\ref{eq:fully_discrete_b_cg}) cancels a similar term originating from the third term.

\begin{corollary} Selecting $\alpha_f $ as in equation~\eqref{eq:alpha_f} and introducing it into \eqref{eq:fully_discrete_b_cg} results in the following bilinear form,
\begin{align}\label{eq:fully_discrete_b_cg_2}
	B_h([\phi_h,\kappa_h,\eta_h],[w_h,v_h,u_h])=
	&(\nabla\phi_h,\nabla w_h)_{\Omega_h} \\\nonumber
	-& \beta\delta_t(\kappa_{h},w_h)_{\Gfh}
	+\frac{(1-\beta)\delta_t^2}{g}\left(\phi_{h},w_h\right)_{\Gfh}
	+\beta\delta_t\left(\phi_{h},v_h\right)_{\Gfh} 
	+\beta g\left(\kappa_h,v_h\right)_{\Gfh}
	\\\nonumber
	-&\delta_t(\eta_{h},w_h)_{\Gsh}
	+(\delta_{tt}d_0+g)\left(\eta_{h},u_h\right)_{\Gsh}  
	+\delta_t\left(\phi_{h},u_h\right)_{\Gsh}  \\\nonumber
	+&\left( \C_\rho:\nabla^2\eta_h,\nabla^2 u_h\right)_{\Gsh}
	+\left(k_\rho\jump{\nabla\eta_h\otimes\nL},\jump{\nabla u_h\otimes\nL}\right)_{\Lj}.
\end{align}
\end{corollary}

\begin{proposition}[Coercivity]\label{prop:coercivity}
If  $ \alpha_f $ is given by \eqref{eq:alpha_f}, the bilinear form~\eqref{eq:fully_discrete_b_cg} is coercive for any $ \beta $ such that $ 0<\beta<1$. That is, there exists a constant $ C>0 $ such that 
	\begin{equation}\label{eq:coercivity}
		B_h([w_h,v_h,u_h],[w_h,v_h,u_h])\geq C_c \vertiii{[w_h,v_h,u_h]}^2,\qquad\forall[w_h,v_h,u_h]\in\V_h\times\V_{\Gf,h}\times\V_{\Gs,h},
	\end{equation}
	with 
	\begin{align}\label{eq:norm}
		\vertiii{[w_h,v_h,u_h]}^2 \eqdef& \|w_h\|^2_{H^1(\Omega_h)} 
		+\left\|\left(\beta g\right)^{1/2}v_h\right\|^2_{\Gfh}
		+\left\|\left(\delta_{tt}d_0+g\right)^{1/2}u_h\right\|^2_{\Gsh} 
		\\\nonumber &
		+\left\|\C^{1/2}_\rho:\nabla^2u_h\right\|^2_{\tGsh}
		+\left\|k_\rho^{1/2}\jump{\nabla u_h\otimes\nL}\right\|^2_{\Lj}.
	\end{align}
\end{proposition}

\begin{proof}
	From the bilinear form defined in equation \eqref{eq:fully_discrete_b_cg_2} and using  equation \eqref{eq:seminorm_bound}, we can write
	\begin{align}\label{eq:coercivity1}
		B_h([w_h,v_h,u_h],[w_h,v_h,u_h])\eqdef
		&\|\nabla w_h\|^2_{\Omega_h} 
		+\left(\frac{(1-\beta)\delta_t^2}{g}\right)\left\|w_h\right\|^2_{\Gfh} \\\nonumber &
		+\left\|\left(\beta g\right)^{1/2}v_h\right\|^2_{\Gfh}
		+\left\|\left(\delta_{tt}d_0+g\right)^{1/2}u_h\right\|^2_{\Gsh}\\\nonumber
		&+\left\| \C^{1/2}_\rho:\nabla^2u_h\right\|^2_{\Gsh}
		+\left\|k_\rho^{1/2}\jump{\nabla u_h\otimes\nL}\right\|^2_{\Lj}\\\nonumber
		\geq&\ C_\Omega\|w_h\|^2_{H^1(\Omega_h)} \\\nonumber
		&+\left\|\left(\beta g\right)^{1/2}v_h\right\|^2_{\Gfh}
		+\left\|\left(\delta_{tt}d_0+g\right)^{1/2}u_h\right\|^2_{\Gsh} \\\nonumber
		&+\left\| \C^{1/2}_\rho:\nabla^2u_h\right\|^2_{\Gsh}
		+\left\|k_\rho^{1/2}\jump{\nabla u_h\otimes\nL}\right\|^2_{\Lj}\\\nonumber
		\geq &\ C_c\vertiii{[w_h,v_h,u_h]}^2.
	\end{align}
	Defining the coercivity constant as $ C_c\eqdef\min(C_\Omega,1)>0 $ proves Proposition \ref{prop:coercivity}.
\end{proof}

\begin{proposition}[Boundedness]\label{prop:bound}
 If  $ \alpha_f $ is given by \eqref{eq:alpha_f}, the bilinear form~\eqref{eq:fully_discrete_b_cg} is bounded for any $ \beta $ such that $ 0<\beta<1$. That is, there exists a constant $ C_b>0 $ such that 
	\begin{equation}\label{eq:bound}
		B_h([\phi_h,\kappa_h,\eta_h],[w_h,v_h,u_h])\leq C_b\vertiii{[\phi_h,\kappa_h,\eta_h]} \vertiii{[w_h,v_h,u_h]} ,\quad\forall[\phi_h,\kappa_h,\eta_h],[w_h,v_h,u_h]\in\V_h\times\V_{\Gf,h}\times\V_{\Gs,h},
	\end{equation}
	with $\vertiii{\cdot}$ defined in eq \eqref{eq:norm}.
\end{proposition}

\begin{proof}
Let us define the minimum draft $ d_0^{\min} $ as the minimum value of the draft for any point in the structure, i.e. $ d_0^{\min}\eqdef\min_{\x\in\Gs}d_0(\x) $. Starting from the bilinear form defined in equation~\eqref{eq:fully_discrete_b_cg_2} and using Schwarz inequality, we have that
\begin{align}\label{eq:boundedness1}
	B_h([\phi_h,\kappa_h,\eta_h],[w_h,v_h,u_h])\leq
	&\left\|\nabla\phi_h\right\|_{\Omega_h}\left\|\nabla w_h\right\|_{\Omega_h}
	+ \delta_t \left(\frac{\beta}{g}\right)^{1/2}\left\|(\beta g)^{1/2}\kappa_{h}\right\|_{\Gfh}\left\|w_h\right\|_{\Gfh} \\\nonumber
	+&(1-\beta)\delta_t^2 g^{-1}\left\|\phi_{h}\right\|_{\Gfh}\left\|w_h\right\|_{\Gfh}
	+\delta_t\left(\frac{\beta}{g}\right)^{1/2}\left\|\phi_{h}\right\|_{\Gfh}\left\|(\beta g)^{1/2}v_h\right\|_{\Gfh} \\\nonumber
	+&\left\|(\beta g)^{1/2}\kappa_h\right\|_{\Gfh}\left\|(\beta g)^{1/2}v_h\right\|_{\Gfh}
	\\\nonumber
	+&\delta_t(\delta_{tt}d_0^{\min}+g)^{-1/2}\left\|(\delta_{tt}d_0+g)^{1/2}\eta_{h}\right\|_{\Gsh}\left\|w_h\right\|_{\Gsh} \\\nonumber
	+&\left\|(\delta_{tt}d_0+g)^{1/2}\eta_{h}\right\|_{\Gsh}\left\|(\delta_{tt}d_0+g)^{1/2}u_h\right\|_{\Gsh}  \\\nonumber
	+&\delta_t(\delta_{tt}d_0^{\min}+g)^{-1/2}\left\|\phi_{h}\right\|_{\Gsh}\left\|(\delta_{tt}d_0+g)^{1/2}u_h\right\|_{\Gsh}  \\\nonumber
	+&\left\| \C_\rho^{1/2}:\nabla^2\eta_h\right\|_{\Gsh}\left\|\C_\rho^{1/2}:\nabla^2 u_h\right\|_{\Gsh} \\\nonumber
	+&\left\|k_\rho^{1/2}\jump{\nabla\eta_h\otimes\nL}\right\|_{\Gsh}\left\|k_\rho^{1/2}\jump{\nabla u_h\otimes\nL}\right\|_{\Lj}.
\end{align}
%\begin{align}\label{eq:boundedness1}
%	\hat{B}_h([\phi_h,\eta_h],[w_h,v_h])\leq
%	&\left\|\nabla\phi_h\right\|_{\Omega_h}\left\|\nabla w_h\right\|_{\Omega_h} 
%	+ \delta_t \left(\beta g\right)^{-1} \left\|\left(\beta g\right)^{1/2}\eta_{h}\right\|_{\Gfh}\left\|w_h\right\|_{\Gfh}	\\\nonumber
%	+& \delta_t\left(\delta_{tt}d^{min}_0+ g\right)^{-1}\left\|\left(\delta_{tt}d_0+ g\right)^{1/2}\eta_{h}\right\|_{\Gsh}\left\|w_h\right\|_{\Gsh} 
%	\\\nonumber
%	+&\frac{(1-\beta)\delta^2_t}{g}\left\|\phi_{h}\right\|_{\Gfh}\left\|w_h\right\|_{\Gfh}	
%	+(1-\beta)\delta_t\left(\beta g\right)^{-1} \left\|\left(\beta g\right)^{1/2}\eta_h\right\|_{\Gfh}\left\|w_h\right\|_{\Gfh} 
%	\\\nonumber
%	%
%	+&\beta\delta_t\left(\beta g\right)^{-1} \left\|\phi_{h}\right\|_{\Gfh}\left\|\left(\beta g\right)^{1/2}v_h\right\|_{\Gfh}
%	+\left\|\left(\beta g\right)^{1/2} \eta_h\right\|_{\Gfh}\left\|\left(\beta g\right)^{1/2}v_h\right\|_{\Gfh}\\\nonumber
%	+&\left\|\left(\delta_{tt}d_0+g\right)^{1/2}\eta_{h}\right\|_{\Gsh}\left\|\left(\delta_{tt}d_0+g\right)^{1/2}v_h\right\|_{\Gsh} 	
%	\\\nonumber 
%	+& \delta_t \left(\delta_{tt}d_0+g\right)^{-1}v\left\|\phi_{h}\right\|_{\Gsh}\left\|\left(\delta_{tt}d_0+g\right)^{1/2}v_h\right\|_{\Gsh} 
%\\\nonumber 
%	+&\left\|\C^{1/2}_\rho:\nabla^2\eta_h\right\|_{\Gsh}\left\|\C^{1/2}_\rho:\nabla^2 v_h\right\|_{\Gsh} 
%	%
%\end{align}
%TOT HIER
Using equations \eqref{eq:H1_equiv} and \eqref{eq:sob_trace}, and noting that $ \|w\|_{\Gf}\leq\|w\|_{\partial\Omega} $ and $ \|w\|_{\Gs}\leq\|w\|_{\partial\Omega} $, we find that
%2. $ \|\cdot\|_{\Omega_h}\leq\|\cdot\|_{H^1(\Omega_h)} $ and Trace theorem~\ref{theorem:trace}
\begin{align}\label{eq:boundedness2}
	B_h([\phi_h,\kappa_h,\eta_h],[w_h,v_h,u_h])\leq
	&\left\|\phi_h\right\|_{H^1(\Omega_h)}\left\|w_h\right\|_{H^1(\Omega_h)}
	+ C_{\partial\Omega}\delta_t \left(\frac{\beta}{g}\right)^{1/2}\left\|(\beta g)^{1/2}\kappa_{h}\right\|_{\Gfh}\left\|w_h\right\|_{H^1(\Omega_h)} \\\nonumber
	+&C_{\partial\Omega}^2(1-\beta)\delta_t^2 g^{-1}\left\|\phi_{h}\right\|_{H^1(\Omega_h)}\left\|w_h\right\|_{H^1(\Omega_h)} \\\nonumber
	+&C_{\partial\Omega}\delta_t\left(\frac{\beta}{g}\right)^{1/2}\left\|\phi_{h}\right\|_{H^1(\Omega_h)}\left\|(\beta g)^{1/2}v_h\right\|_{\Gfh} \\\nonumber
	+&\left\|(\beta g)^{1/2}\kappa_h\right\|_{\Gfh}\left\|(\beta g)^{1/2}v_h\right\|_{\Gfh}
	\\\nonumber
	+&C_{\partial\Omega}\delta_t(\delta_{tt}d_0^{\min}+g)^{-1/2}\left\|(\delta_{tt}d_0+g)^{1/2}\eta_{h}\right\|_{\Gsh}\left\|w_h\right\|_{H^1(\Omega_h)} \\\nonumber
	+&\left\|(\delta_{tt}d_0+g)^{1/2}\eta_{h}\right\|_{\Gsh}\left\|(\delta_{tt}d_0+g)^{1/2}u_h\right\|_{\Gsh}  \\\nonumber
	+&C_{\partial\Omega}\delta_t(\delta_{tt}d_0^{\min}+g)^{-1/2}\left\|\phi_{h}\right\|_{H^1(\Omega_h)}\left\|(\delta_{tt}d_0+g)^{1/2}u_h\right\|_{\Gsh}  \\\nonumber
	+&\left\| \C_\rho^{1/2}:\nabla^2\eta_h\right\|_{\Gsh}\left\|\C_\rho^{1/2}:\nabla^2 u_h\right\|_{\Gsh} \\\nonumber
	+&\left\|k_\rho^{1/2}\jump{\nabla\eta_h\otimes\nL}\right\|_{\Gsh}\left\|k_\rho^{1/2}\jump{\nabla u_h\otimes\nL}\right\|_{\Lj}.
\end{align}
%\begin{align}\label{eq:boundedness2}
%	\hat{B}_h([\phi_h,\eta_h],[w_h,v_h])\leq
%&\left\|\phi_h\right\|_{H^1(\Omega_h)}\left\|w_h\right\|_{H^1(\Omega_h)} 
%	+ \delta_t \left(\beta g\right)^{-1} C_{\partial\Omega}\left\|\left(\beta g\right)^{1/2}\eta_{h}\right\|_{\Gfh}\left\|w_h\right\|_{H^1(\Omega_h)} 	\\\nonumber
%	+& \delta_t\left(\delta_{tt}d^{min}_0+ g\right)^{-1}C_{\partial\Omega}\left\|\left(\delta_{tt}d_0+ g\right)^{1/2}\eta_{h}\right\|_{\Gsh}\left\|w_h\right\|_{H^1(\Omega_h)} 
%	\\\nonumber
%	+&\frac{(1-\beta)\delta^2_t}{g}C_{\partial\Omega}^2\left\|\phi_{h}\right\|_{H^1(\Omega_h)} \left\|w_h\right\|_{H^1(\Omega_h)} 
%	\\\nonumber
%	+&(1-\beta)\delta_t\left(\beta g\right)^{-1} C_{\partial\Omega}\left\|\left(\beta g\right)^{1/2}\eta_h\right\|_{\Gfh}\left\|w_h\right\|_{H^1(\Omega_h)}  
%	\\\nonumber
%	%
%	+&\beta\delta_t\left(\beta g\right)^{-1} C_{\partial\Omega}\left\|\phi_{h}\right\|_{H^1(\Omega_h)} \left\|\left(\beta g\right)^{1/2}v_h\right\|_{\Gfh}
%	+\left\|\left(\beta g\right)^{1/2} \eta_h\right\|_{\Gfh}\left\|\left(\beta g\right)^{1/2}v_h\right\|_{\Gfh}\\\nonumber
%	+&\left\|\left(\delta_{tt}d_0+g\right)^{1/2}\eta_{h}\right\|_{\Gsh}\left\|\left(\delta_{tt}d_0+g\right)^{1/2}v_h\right\|_{\Gsh} 	
%	\\\nonumber 
%	+& \delta_t \left(\delta_{tt}d_0+g\right)^{-1}C_{\partial\Omega}\left\|\phi_{h}\right\|_{H^1(\Omega_h)}\left\|\left(\delta_{tt}d_0+g\right)^{1/2}v_h\right\|_{\Gsh} 
%\\\nonumber 
%	+&\left\|\C^{1/2}_\rho:\nabla^2\eta_h\right\|_{\Gsh}\left\|\C^{1/2}_\rho:\nabla^2 v_h\right\|_{\Gsh} 
%	%
%\end{align}
Which can be written in matrix vector form as follows
\begin{align}\label{eq:boundedness3}
	B_h([\phi_h,\kappa_h,\eta_h],[w_h,v_h,u_h])\leq
	\left(\begin{matrix}
		\|\phi_h\|_{H^1(\Omega_h)}\\
		\left\|\left(\beta g\right)^{1/2}\kappa_h\right\|_{\Gfh}\\
		\left\|\left(\delta_{tt}d_0+g\right)^{1/2}\eta_h\right\|_{\Gsh} \\
		\left\|\C^{1/2}_\rho:\nabla^2\eta_h\right\|_{\tGsh}\\
		\left\|k_\rho^{1/2}\jump{\nabla \eta_h\otimes\nL}\right\|_{\Lj}
	\end{matrix} \right)^T
		\mathbf{A}
	\left(\begin{matrix}
		 \|w_h\|_{H^1(\Omega_h)}\\
		\left\|\left(\beta g\right)^{1/2}v_h\right\|_{\Gfh}\\
		\left\|\left(\delta_{tt}d_0+g\right)^{1/2}u_h\right\|_{\Gsh} \\
		\left\|\C^{1/2}_\rho:\nabla^2u_h\right\|_{\tGsh}\\
		\left\|k_\rho^{1/2}\jump{\nabla u_h\otimes\nL}\right\|_{\Lj}
	\end{matrix} \right),
\end{align}
with $ \mathbf{A} $ a symmetric matrix defined by
\begin{align}\label{eq:boundedness4}
		\mathbf{A} =
	\left(\begin{matrix}
	1+\frac{(1-\beta)\delta^2_t}{g}C_{\partial\Omega}^2&
	C_{\partial\Omega}\delta_t \left(\frac{\beta}{g}\right)^{1/2} &
	C_{\partial\Omega}\delta_t(\delta_{tt}d_0^{\min}+g)^{-1/2}&
	0&0 \\
	C_{\partial\Omega}\delta_t \left(\frac{\beta}{g}\right)^{1/2}&1&0&0&0 \\
	C_{\partial\Omega}\delta_t(\delta_{tt}d_0^{\min}+g)^{-1/2}& 0&1&0&0 \\
	0&0&0&1&0 \\
	0&0&0&0&1
	\end{matrix} \right).
\end{align}	
%\begin{align}\label{eq:boundedness4}
%\mathbf{A} =
%\left(\begin{matrix}
%	1+\frac{(1-\beta)\delta^2_t}{g}C_{\partial\Omega}^2& 	 \delta_t \left(\beta g\right)^{-1} C_{\partial\Omega} & \delta_t \left(\delta_{tt}d_0+g\right)^{-1}C_{\partial\Omega}&0 \\
%	\delta_t \left(\beta g\right)^{-1} C_{\partial\Omega} + (1-\beta)\delta_t\left(\beta g\right)^{-1} C_{\partial\Omega}& 1&0&0 \\
%	\delta_t\left(\delta_{tt}d^{min}_0+ g\right)^{-1}C_{\partial\Omega}& 0&1&0 \\
%	0& 0&0&1 
%\end{matrix} \right).
%\end{align}	
Using arguments analogous in the proof of the Gershgorin circle theorem, if we define the boundedness constant $ C_b $ as the maximum eigenvalue of the matrix $ \mathbf{A} $, that is
\begin{align}
C_b =\max_i \left ( \sum_j | \mathbf{A}_{ij}| \right ),
\end{align}	
we arrive at the bound \eqref{eq:bound}.
\end{proof}

Using the Lax-Milgram theorem, see for instance \cite{ern2004theory}, together with Proposition \ref{prop:coercivity} and Proposition \ref{prop:bound}, we can conclude that the semi-discrete problem \eqref{eq:galerkin_form} has a unique solution.

%{\color{red}
%\subsubsection*{Accuracy:}
%\begin{proposition}
%	Assume that the penalization parameter $ \gamma $ is such that satisfies $ \gamma>2C_I $. For any scalar $ \beta $ with $ 0<\beta<1 $, the bilinear form~\eqref{eq:bilinear_modified_2} is bounded. That is, there exists a constant $ C_b>0 $ such that 
%\end{proposition}
%
%\begin{proof}
%
%\begin{align}
%B([\phi,\eta],[w,v])\eqdef&(\nphi,\nabla w)_\Omega - (\eta_t,w)_{\Gf\cup\Gs}
%+(1-\beta)\left(\phi_t+g\eta, \frac{\delta_t}{g}  w\right)_{\Gf} \\\nonumber
%+&\beta\left(\phi_t+g\eta, v\right)_{\Gf} 
%+\left(d_0\eta_{tt} + \phi_t+g\eta,v\right)_{\Gs} + \left( \C_\rho:\nabla^2\eta,\nabla^2 v\right)_{\Gs},
%\end{align}
%\end{proof}
%}

\subsection{Discontinuous formulation}\label{subsec:discontinuous}
Let us now analyse the CDG formulation as stated in the semi-discrete form~\eqref{eq:cgdg_form}. 
\begin{proposition}[Consistency of the C/DG formulation]\label{prop:consistency_cdg}
	The semi-discrete problem~\eqref{eq:cgdg_form} is consistent. That is, the exact solution $ [\phi,\kappa,\eta]\in\V\times\V_{\Gf}\times\V_{\Gs}$ satisfies the approximate problem
	\begin{equation}\label{eq:cgdg_t_form_exact}
		\hat{B}_h([\phi,\kappa,\eta],[w_h,v_h,u_h])=L_h([w_h,v_h,u_h])\quad\forall[w_h,v_h,u_h]\in\hatV_h\times\hatV_{\Gf,h}\times\hatV_{\Gs,h}.
	\end{equation}
\end{proposition}
\begin{proof}
	The consistency statement \eqref{eq:cgdg_t_form_exact} results from the same reasoning as in Proposition \ref{prop:consistency}, noting that the terms involving $ \jump{\nabla \eta\otimes\nL} $  appearing in equation \eqref{eq:bilinear_modified_h_cgdg} vanish since the solution gradients are continuous across element boundaries not belonging to $ \Lj $.
\end{proof}

The fully discrete C/DG bilinear form in the time domain, again omitting the super-index related to the time step $ (\cdot)^{n+1} $, reads
\begin{align}\label{eq:fully_discrete_b_cgdg}
	\hat{B}_h([\phi_h,\kappa_h,\eta_h],[w_h,v_h,u_h])\eqdef
	&(\nabla\phi_h,\nabla w_h)_{\Omega_h} 
	- (\delta_t\kappa_{h},w_h)_{\Gfh}
	+\beta\left(\delta_t\phi_{h}+g\kappa_h,\alpha_fw_h+v_h\right)_{\Gfh}\\\nonumber
	-&(\delta_t\eta_{h},w_h)_{\Gsh}
	+\left(\delta_{tt}d_0\eta_{h} + \delta_t\phi_{h}+g\eta_h,u_h\right)_{\Gsh}  
	+\left( \C_\rho:\nabla^2\eta_h,\nabla^2 u_h\right)_{\Gsh}\\\nonumber
	+&\left(k_\rho\jump{\nabla\eta_h\otimes\nL},\jump{\nabla u_h\otimes\nL}\right)_{\Lj}
	-\left( \average{\C_\rho:\nabla^2\eta_h},\jump{\nabla u_h\otimes\nL}\right)_{\Lsh} \\\nonumber
	-&\left(\jump{\nabla\eta_h\otimes\n_{\Lsh}},\average{\C_\rho:\nabla^2 u_h}\right)_{\Lsh}
	+\frac{\gamma\hat{D}_{\rho}}{h}\left(\jump{\nabla\eta_h\otimes\nL},\jump{\nabla u_h\otimes\nL}\right)_{\Lsh}.
\end{align}

\begin{proposition}[Coercivity of the C/DG formulation]\label{prop:coercivity_cdg}
	If  $ \alpha_f $ is given by \eqref{eq:alpha_f} and $ \gamma>2C_I $, the bilinear form~\eqref{eq:fully_discrete_b_cgdg} is coercive for any $ \beta $ such that $ 0<\beta<1$. That is, there exists a constant $ \hat{C}_c>0 $ such that 
	\begin{equation}\label{eq:coercivity_cdg}
		\hat{B}_h([w_h,v_h,u_h],[w_h,v_h,u_h])\geq \hat{C}_c \vertiii{[w_h,v_h,u_h]}^2_{\scriptsize\mbox{CDG}},\qquad\forall[w_h,v_h,u_h]\in\hatV_h\times\hatV_{\Gf,h}\times\hatV_{\Gs,h},
	\end{equation}
	with 
	\begin{align}\label{eq:norm_cdg}
		\vertiii{[w_h,v_h,u_h]}^2_{\scriptsize\mbox{CDG}} \eqdef& \|w_h\|^2_{H^1(\Omega_h)} 
		+\left\|\left(\beta g\right)^{1/2}v_h\right\|^2_{\Gfh}
		+\left\|\left(\delta_{tt}d_0+g\right)^{1/2}u_h\right\|^2_{\Gsh} 
		\\\nonumber &
		+\left\|\C^{1/2}_\rho:\nabla^2u_h\right\|^2_{\tGsh}
		+\left\|k_\rho^{1/2}\jump{\nabla u_h\otimes\nL}\right\|^2_{\Lj}
		+\left\|\left(\frac{\hat{D}_{\rho}}{h}\right)^{1/2}\jump{\nabla\u_h\otimes\nL}\right\|^2_{\Lsh}.
	\end{align}
\end{proposition}
\begin{proof}
	Starting from the second step in equation \eqref{eq:coercivity1} we have that
	\begin{align}\label{eq:coercivity1_cdg}
		\hat{B}_h([w_h,v_h,u_h],[w_h,v_h,u_h])\geq&\ C_\Omega\|w_h\|^2_{H^1(\Omega_h)}
		+\left\|\left(\beta g\right)^{1/2}v_h\right\|^2_{\Gfh}
		+\left\|\left(\delta_{tt}d_0+g\right)^{1/2}u_h\right\|^2_{\Gsh} \\\nonumber
		&+\left\| \C^{1/2}_\rho:\nabla^2u_h\right\|^2_{\Gsh}
		+\left\|k_\rho^{1/2}\jump{\nabla u_h\otimes\nL}\right\|^2_{\Lj}\\\nonumber
		&-2\left(\jump{\nabla u_h\otimes\n_{\Lsh}},\average{\C_\rho:\nabla^2 u_h}\right)_{\Lsh}\\\nonumber
		&+\left\|\left(\frac{\gamma\hat{D}_{\rho}}{h}\right)^{1/2}\jump{\nabla\u_h\otimes\nL}\right\|^2_{\Lsh}.
%		\\\nonumber
%		\geq &\ C_c\vertiii{[w_h,v_h,u_h]}^2.
	\end{align}
	Using Young's $\epsilon$-inequality and the inverse inequality stated in Theorem \ref{theorem:trace_inequality}, one can bound the sixth term in equation \eqref{eq:coercivity1_cdg} as follows
	\begin{align}\label{eq:coercivity2_cdg}
		2\left|\left( \average{\C_\rho:\nabla^2v_h},\jump{\nabla v_h\otimes\nL}\right)_{\Lsh}\right|\leq&
		\ \epsilon\left\|h^{1/2}\average{\C^{1/2}_\rho:\nabla^2v_h}\right\|^2_{\Lsh}+\frac{1}{\epsilon}\left\|h^{-1/2}\jump{\nabla v_h\otimes\nL}\right\|^2_{\Lsh}\\\nonumber
		\leq& \sum_{E\in\Gsh} \epsilon C_I\hat{D}_\rho\left\|\C^{1/2}_\rho:\nabla^2v_h\right\|^2_{E}+\frac{1}{\epsilon}\left\|h^{-1/2}\jump{\nabla v_h\otimes\nL}\right\|^2_{\Lsh}\\\nonumber
		\underset{\epsilon=\frac{1}{2C_{I}\hat{D}_\rho}}{=}& \sum_{E\in\Gsh} \frac{1}{2}\left\|\C^{1/2}_\rho:\nabla^2v_h\right\|^2_{E}+\left(\frac{2C_{I}\hat{D}_\rho}{h}\right)\left\|\jump{\nabla v_h\otimes\nL}\right\|^2_{\Lsh}.
	\end{align}
	Introducing \eqref{eq:coercivity2_cdg} into \eqref{eq:coercivity1_cdg} we have that
		\begin{align}\label{eq:coercivity3_cdg}
		\hat{B}_h([w_h,v_h,u_h],[w_h,v_h,u_h])\geq&\ C_\Omega\|w_h\|^2_{H^1(\Omega_h)}
		+\left\|\left(\beta g\right)^{1/2}v_h\right\|^2_{\Gfh}
		+\left\|\left(\delta_{tt}d_0+g\right)^{1/2}u_h\right\|^2_{\Gsh} \\\nonumber
		&+\frac{1}{2}\left\| \C^{1/2}_\rho:\nabla^2u_h\right\|^2_{\Gsh}
		+\left\|k_\rho^{1/2}\jump{\nabla u_h\otimes\nL}\right\|^2_{\Lj}\\\nonumber
		&+(\gamma-2C_I)\left\|\left(\frac{\hat{D}_{\rho}}{h}\right)^{1/2}\jump{\nabla\u_h\otimes\nL}\right\|^2_{\Lsh}\\\nonumber
		\geq &\ \hat{C}_c\vertiii{[w_h,v_h,u_h]}^2_{\scriptsize\mbox{CDG}}.
	\end{align}
Defining the coercivity constant as $ \hat{C}_c\eqdef\min(C_\Omega,1,(\gamma-2C_I)) $, which is greater than zero provided that $ \gamma>2C_I $, we prove Proposition \ref{prop:coercivity_cdg}.
\end{proof}

\begin{proposition}[Boundedness of the C/DG formulation]\label{prop:bound_cdg}
	If  $ \alpha_f $ is given by \eqref{eq:alpha_f}, the bilinear form~\eqref{eq:fully_discrete_b_cgdg} is bounded for any $ \beta $ such that $ 0<\beta<1$. That is, there exists a constant $ \hat{C}_b>0 $ such that 
	\begin{align}\label{eq:bound_cgdg}
		\hat{B}_h([\phi_h,\kappa_h,\eta_h],[w_h,v_h,u_h])&\leq \hat{C}_b\vertiii{[\phi_h,\kappa_h,\eta_h]}_{\scriptsize\mbox{CDG}} \vertiii{[w_h,v_h,u_h]}_{\scriptsize\mbox{CDG}} ,\\\nonumber
		&\forall\ [\phi_h,\kappa_h,\eta_h],[w_h,v_h,u_h]\in\hatV_h\times\hatV_{\Gf,h}\times\hatV_{\Gs,h},
	\end{align}
	with $\vertiii{\cdot}_{\scriptsize\mbox{CDG}}$ defined in eq \eqref{eq:norm_cdg}.
\end{proposition}
\begin{proof}
	Statement \eqref{eq:bound_cgdg} follows from the same arguments used in Proposition \ref{prop:bound} incorporating the inequality derived in equation \eqref{eq:coercivity2_cdg}. The reader is referred to \cite{engel2002continuous} for a step-by-step proof.
\end{proof}

Again, using the Lax-Milgram theorem together with Proposition \ref{prop:coercivity_cdg} and Proposition \ref{prop:bound_cdg}, we can conclude that the semi-discrete problem \eqref{eq:cgdg_form} has a unique solution.

%\begin{proposition}[Energy conservation of the C/DG formulation]\label{prop:energy_cdg}
%	The semi-discrete problem~\eqref{eq:cgdg_form} is energy conserving for any $ \beta $ such that $ 0<\beta<1$. That is, 
%	\begin{equation}\label{eq:energy_cdg}
%		\frac{d\Et}{dt} = 0.
%	\end{equation}
%	With
%	\begin{align}
%		\Et\eqdef&\Ekf+\Epf+\Eks+\Ees,\\
%		\Ekf\eqdef&\frac{1}{2}\left\|\nabla\phi\right\|^2_\Omega,\\
%		\Epf\eqdef&\frac{g}{2}\left(\left\|\kappa\right\|^2_{\Gf}+\left\|\eta\right\|^2_{\Gs}\right),\\
%		\Eks\eqdef&\frac{1}{2}\left\|d_0^{1/2}\eta_t\right\|^2_{\Gs},\\
%		\Ees\eqdef&\frac{1}{2}\left\|\C^{1/2}:\nabla^2\eta\right\|^2_{\Gs}+\frac{1}{2}\left\|k^{1/2}_\rho\jump{\nabla\eta_h\otimes\nL}\right\|^2_{\Lj}.
%	\end{align}
%\end{proposition}

%\newpage
%{\color{red}
%\section{Software implementation}
%\begin{itemize}
%	\item Mixed-dimensional spaces
%	\item 2nd order ODEs
%	\item ...
%\end{itemize}
%}

\section{Numerical results}\label{sec:numerical_results}
%\begin{itemize}
%	\item 2D convergence test
%	\item 2D periodic (beam only)
%	\item 2D incoming wave single beam
%	\item 2D incoming wave multiple modules
%	\item 2D incoming irregular wave
%	\item 3D incoming regular wave single plate
%	\item 3D incoming irregular wave multiple modules
%\end{itemize}
In this section we assess the behavior of the proposed formulation for a variety of two and three-dimensional tests, analysing the accuracy, convergence and conservation properties, as well as comparing with analytical and experimental solutions that can be found in existing literature.

We start with a problem with analytical solution for the 2-dimensional case, that is the evolution of a floating infinite beam subject to an initial harmonic condition. After, we assess the performance of the method for finite beams with elastic joints in the frequency and time domains, followed by the study of a finite floating beam over irregular sea bed. Finally, we assess the behavior of a floating plate in a 3-dimensional domain and we show that the proposed approach is suitable to solve problems with structures with arbitrary shape.

\subsection{Implementation remarks}

The monolithic FE formulation and other algorithms used in the experiments below have been implemented using the Julia programming language \cite{Bezanson2017} version 1.7 and the Gridap finite element library \cite{Badia2020} version 0.17. Gridap is a free and open-source finite element library fully implemented in Julia. One of its main distinctive features is its user interface, which has a high-level syntax that resembles the notation used to define weak forms mathematically. Internally, Gridap leverages the Julia JIT compiler to generate an efficient finite assembly loop from the user input automatically \cite{Verdugo2021}, which results in efficient and easy to write user code. The formulations presented in this paper can be easily implemented using the high-level interface of Gridap in a convenient way. See, e.g., Figure \ref{fig:code_sample} that contains the implementation of the numerical example in Section \ref{sec:inf_beam_time}. Note that, even though the proposed monolithic formulation is rather complex, its implementation in Gridap can be done in few lines of code. In particular, the definition of the weak form is very compact and has a clear connection with the corresponding mathematical notation. We have taken advantage of the Gridap support for multi-field PDEs and the possibility to combine interpolation spaces defined on geometries with different spatial dimensions. A crucial feature for the implementation of the monolithic formulation is the capacity of Gridap to integrate weak forms on domains different from the ones used to define the interpolation spaces. In particular, this makes possible to integrate the jump terms on $\Lambda$. This computation is particularly challenging from an implementation point of view since it involves geometries with three different spatial dimensions. E.g, in 2D, $\Lambda$ is a 0-dimensional domain, the elevation $\eta$  is defined on the 1-dimensional domain $\Gamma_\mathrm{str}$, and the code that implements the interpolation space for $\eta$ is aware that $\Gamma_\mathrm{str}$ is on the boundary of the 2-dimensional domain $\Omega$. This last implementation ingredient is used internally to compute the terms that involve operations  between the elevation $\eta$ and the potential $\phi$ since the latter is defined on $\Omega$. To our best knowledge, other general-purpose finite element libraries are not able to handle this particular case, at least via such a compact high-level user interface. E.g., at the time of writing FEniCS is able to integrate weak forms on geometries with $d$ and $d-1$ space dimensions at most \cite{DaversinCatty2021} and, thus, cannot be used to implement the formulation of this paper easily, which requires geometries with $d$, $d-1$, and also $d-2$ space dimensions. The numerical results below have been computed on a laptop with an Intel(R) Core(TM) i7-8665U CPU at 1.90GHz with approximately 16GiB of RAM and should be reproducible on a machine with similar characteristics. The software used to generate the results presented in this section is available at the registered Julia package \textit{MonolithicFEMVLFS.jl} \cite{Colomes_MonolithicFEMVLFS_2022}.

\begin{figure}[pos=h!]
	\centering
	\includegraphics[width=\textwidth]{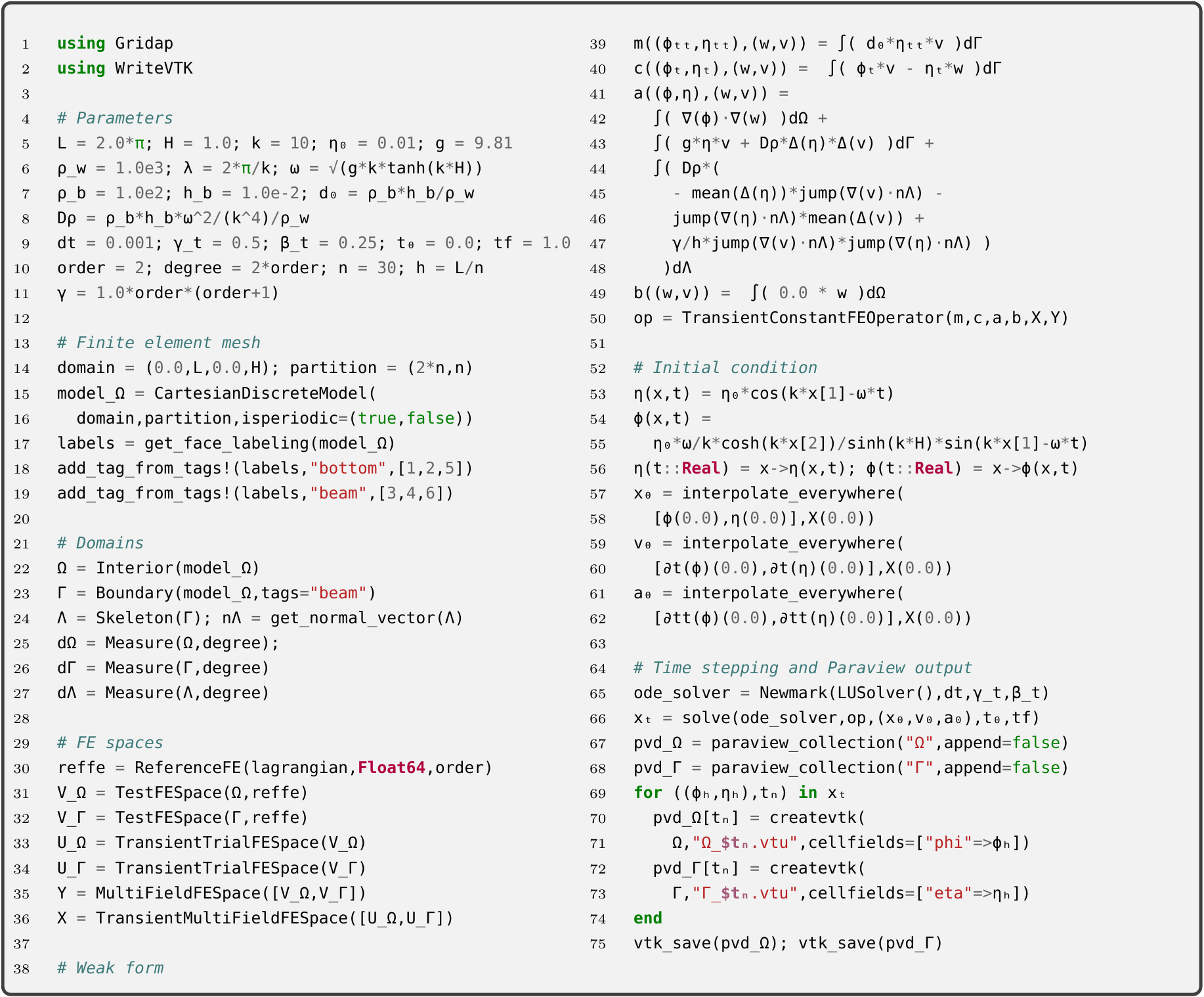}
	\caption{Implementation of the numerical example in Section \ref{sec:inf_beam_time} using Julia and the high-level user interface of Gridap.}
	\label{fig:code_sample}
\end{figure}

\subsection{Infinite beam in time domain} \label{sec:inf_beam_time}
Let us consider an infinite beam floating on top of an infinite 2-dimensional potential flow domain. This setting is achieved by considering periodic boundary conditions in the vertical boundaries, $ \Gamma_{\scriptsize periodic} $, see Figure \ref{fig:infinite_beam}. Here we set $ L=2\pi\ \mbox{m} $ and $ H=1.0\ \mbox{m} $. 
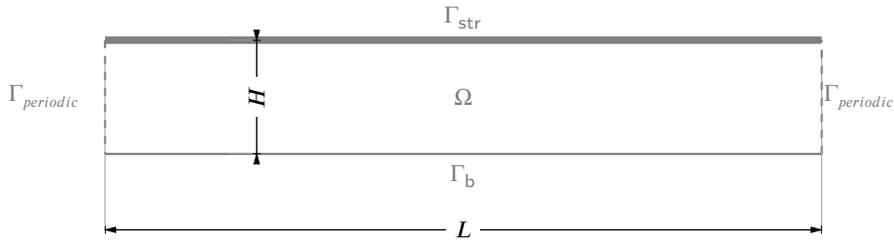
\begin{figure}[pos=h!]
	\centering
	\begin{tikzpicture}
	\draw [dashed, line width=0.03cm] (0,1.5) -- (0,0);
	\draw [line width=0.03cm] (0,0) -- (1.5*6.3,0);
	\draw [dashed, line width=0.03cm]  (1.5*6.3,0) -- (1.5*6.3,1.5);
	\draw[line width=0.1cm](1.5*6.3,1.5) -- (0,1.5);
	\draw (0.75*6.3,-0.3) node{$ \Gb $};
	\draw (1.5*6.3+0.5,0.75) node{$  \Gamma_{\scriptsize periodic} $};
	\draw (-0.8,0.75) node{$  \Gamma_{\scriptsize periodic} $};
	\draw (0.75*6.3,1.5+0.3) node{$\Gs$};
	\draw (0.75*6.3,0.75) node{$\Omega$};
	\node at (0,0) (na) {};
	\node at (1.5*6.3,0) (nb) {};
	\node at (2,0) (nc) {};
	\node at (2,1.5) (nd) {};
	\dimline[line style = {line width=0.5},
	extension start length=-0.24,
	extension end length=-0.24] {($ (na) + (0,-1) $)}{($ (nb) + (0,-1) $)}{$L$};
	\dimline[line style = {line width=0.5,
	arrows=dimline reverse-dimline reverse},
	extension start length=-0.24,
	extension end length=-0.24] {(nc)}{(nd)}{$H$};
	\end{tikzpicture}
	\caption{Infinite beam geometry definition.}
	\label{fig:infinite_beam}
\end{figure}
Let us also consider a traveling wave given by the following surface elevation and potential flow expressions:
\begin{subequations}\label{eq:exact}
	\begin{align}
		\label{eq:exact_phi}
		\phi((x,y),t) =& -\frac{\eta_0\omega}{k_\lambda}\frac{\cosh(k_\lambda y)}{\sinh(k_\lambda  H)}\sin(k_\lambda  x-\omega t),\\
		\label{eq:exact_eta}
		\eta((x,y),t)=&\eta_0\cos(k_\lambda x-\omega t).
	\end{align}
\end{subequations}
Where $ k_\lambda $ is the wavenumber, $ \omega $ the wave frequency.
With these definitions, it can be seen that the equations \Eq{eq:exact_phi}-\Eq{eq:exact_eta} satisfy equations \Eq{eq:pflow} and \Eq{eq:dynamic_linear_bc_s} when $ D=\frac{\rho_bh_b\omega^2}{k_\lambda^4} $. Moreover, if $ \omega=\sqrt{gk_\lambda\tanh(k_\lambda H)} $, the kinematic boundary conditions \Eq{eq:kinematic_bc} are also satisfied. Therefore, in this section we use the previous definitions for $ D $ and $ \omega $, with the remaining parameter values as given in Table~\ref{tab:exact_solution_properties}. In Figure~\ref{fig:periodic_beam} we show the velocity potential and surface elevation fields, $ \phi_{h} $ and $ \eta_{h} $, at four different times, $ t={0.0,\frac{T}{4},\frac{T}{2},\frac{3T}{4}} $, with $ T\eqdef\frac{2\pi}{\omega} $ the wave period.
 \begin{table}[pos=h]
 	\centering
 	\label{tab:exact_solution_properties}
 	\caption{Infinite beam test parameters.}
 	\begin{tabular}{lccc}
 		Parameter&Symbol&Value&Units\\ \hline
 		Water density&$\rho_w$&$1.0e3$&$\mbox{kg}/\mbox{m}^3$\\
 		Structure density&$\rho_b$&$1.0e2$&$\mbox{kg}/\mbox{m}^3$\\
 		Structure thickness&$h_b$&$1.0e\text{-}2$&$\mbox{m}$\\
% 		Structure rigidity&$D$&$1.0e\text{-}2$&$\mbox{Nm}$\\
 		Gravity acceleration&$g$&9.81&$\mbox{m}/\mbox{s}^2$\\
 		Surface elevation&$\eta_0$&0.01&$\mbox{m}$\\
 	\end{tabular}
 \end{table}
\begin{figure}[pos=h!]
\centering
\begin{tabular}{rcl}
	\multirow{ 4}{*}{\includegraphics[clip=true,trim=0 0 50cm 0,width=0.15\textwidth]{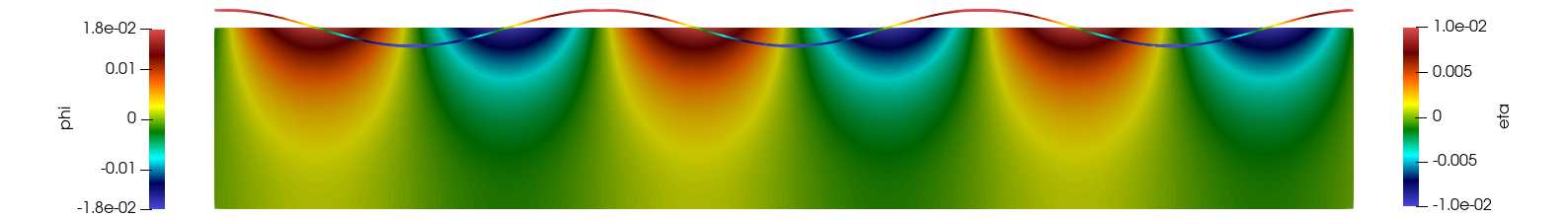}}
	&\includegraphics[clip=true,trim=8cm 0 8cm 0,width=0.6\textwidth]{PeriodicBeam_0.0.png} &\multirow{ 4}{*}{\includegraphics[clip=true,trim=50cm 0 0 0,width=0.15\textwidth]{PeriodicBeam_0.0.png}}\\
	&\includegraphics[clip=true,trim=8cm 0 8cm 0,width=0.6\textwidth]{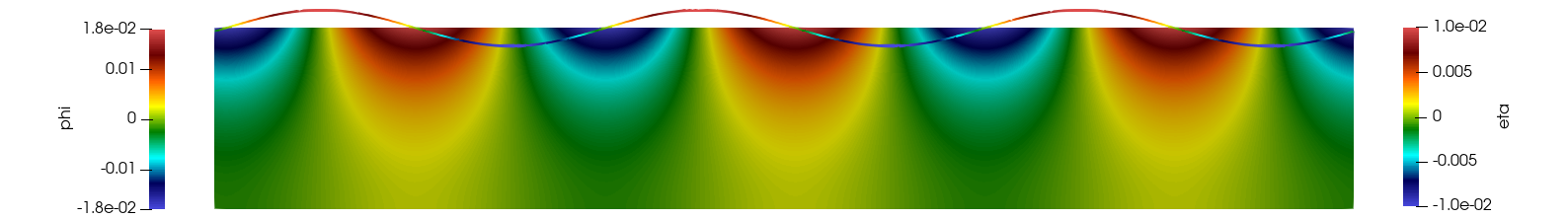}&\\
	&\includegraphics[clip=true,trim=8cm 0 8cm 0,width=0.6\textwidth]{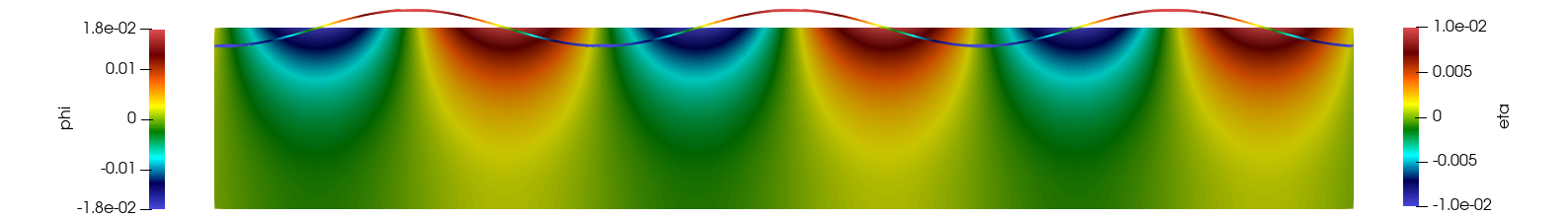}&\\
	&\includegraphics[clip=true,trim=8cm 0 8cm 0,width=0.6\textwidth]{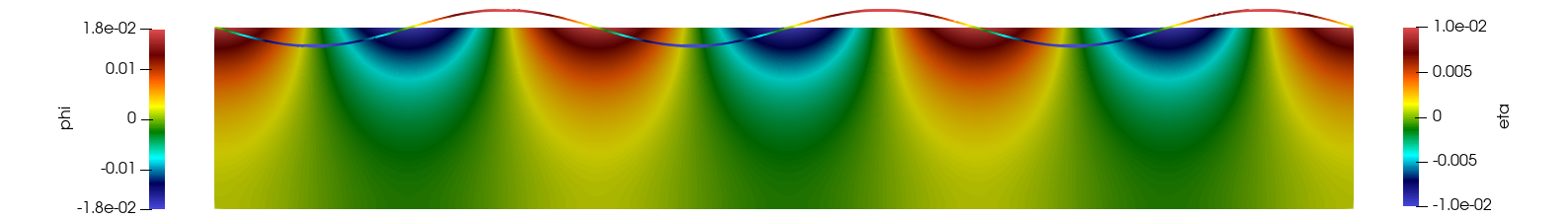}&
\end{tabular}
\caption{From top to bottom: velocity potential and surface elevation (magnified by a factor of 10), $ \phi_{h} $ and $ \eta_{h} $, at  $ t=0.0 $, $t=\frac{T}{4}$, $t=\frac{T}{2}$ and $t=\frac{3T}{4} $.}
\label{fig:periodic_beam}	
\end{figure}
The results shown in Figure~\ref{fig:periodic_beam} have been computed using order $ r=4 $, $ n_x=20 $ elements in the horizontal direction, $ n_y=10 $ elements in the vertical direction and a time step size of $ \Delta t=\frac{T}{50} $. 

\subsubsection{Convergence in space}
\label{subsec:convergence_space}
We first assess the convergence rate of the method by evaluating the $ L^2 $-norm of the solution error for different mesh sizes and polynomial orders. Given the solution to problem \Eq{eq:cgdg_t_form} at time $ t^{n+1}$, $ [\phi_h^{n+1},\eta_h^{n+1}] $, the error of the potential and surface elevation are
\begin{align}
\label{eq:l2norm_phi}
e_\phi^{n+1}&\eqdef\|\phi(t^{n+1})-\phi_h^{n+1}\|_\Omega,\\
\label{eq:l2norm_eta}
e_\eta^{n+1}&\eqdef\|\eta(t^{n+1})-\eta_h^{n+1}\|_{\Gs}.
\end{align}

In Figure \ref{fig:convergence_hp} we plot the potential and surface elevation errors, $ e_\phi^{n+1}  $ and $ e_\eta^{n+1} $, respectively, with respect to the number of elements in the horizontal direction. Here we use a uniform mesh with twice the number of elements in horizontal direction than in the vertical direction. The wavenumber  is set to $ k_\lambda=15 $ and we select a very small time step size, $ \Delta t = 1.0e\text{-}6 $, with a final time $ t=1.0e\text{-}4 $. The choice of such a small time step is to avoid pollution of the error by the time discretization, especially for the finest mesh and higher polynomial degree.

As expected, in Figure \ref{fig:convergence_hp}, we see that both errors, $ e_\phi $ and $ e_\eta $, converge with the expected order of convergence, i.e. $ \mathcal{O}(h^{k+1}) $ for reference FE of order $ k $.
\begin{figure}[pos=h!]
	\centering
	\includegraphics[width=0.45\textwidth]{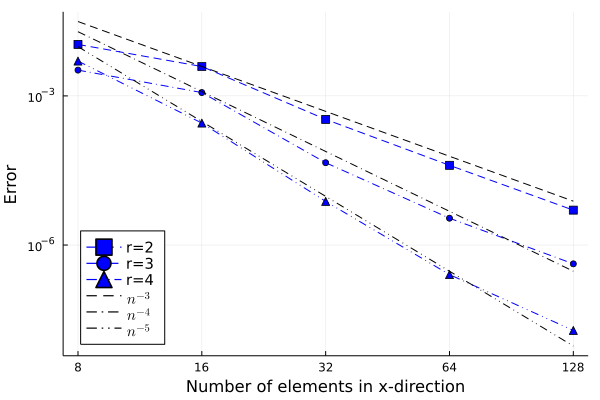}
	\includegraphics[width=0.45\textwidth]{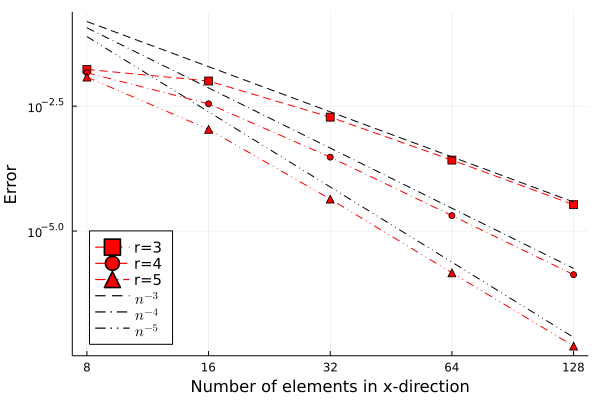}
	\caption{Evolution of the L$2$-norm of the potential error $ e_\phi $ (left), and surface elevation error $ e_\eta $ (right) for different element sizes and element orders, $r=\{2,3,4\} $.}	
	\label{fig:convergence_hp}
\end{figure}

In some cases, to reduce the computational burden, one might be tempted to reduce the order of the FE space for the velocity potential at the interior of the fluid domain and keep a higher order FE space for the surface elevation. The formulation proposed in this work enables different order of interpolation for the different spaces, as long as the trace of the FE space of the velocity potential belongs to the FE space of the free surface elevation, as noted in the proof of proposition~\ref{prop:energy}. Here we stress this case by selecting a 2nd order piece-wise polynomial for $ \hat{\V}_h $ and varying polynomial order for $ \hatV_{\Gf,h} $ and $ \hatV_{\Gs,h} $. In Figure~\ref{fig:convergence_hp_fixed_order} we see that, even when keeping fixed the polynomial order for the velocity potential to 2nd order, the order of convergence for the surface elevation is not affected.

\begin{figure}[pos=h!]
	\centering
	\includegraphics[width=0.45\textwidth]{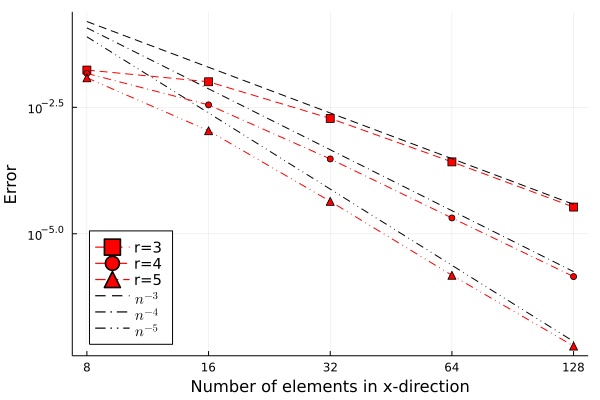}
	\caption{Evolution of the L$2$-norm of surface elevation error $ e_\eta $ (right) for different element sizes and element orders, $r=\{2,3,4\} $, keeping the order of $ \hatV_h $ fixed to $ r=2 $.}	
	\label{fig:convergence_hp_fixed_order}
\end{figure}

\subsubsection{Convergence in time}
We use the same setting as defined in sub-section \ref{subsec:convergence_space} to assess the convergence properties of the time discretization. Here we use a 4th order polynomial space with a mesh of 128 elements in the horizontal direction and 64 in the vertical direction. Since we want to minimize the spatial error, here we use a wave number of $ k_\lambda=1 $. In Figure \ref{fig:convergence_dt} we depict the potential and surface elevation errors, $ e_\phi $ and $ e_\eta$ at $ t=1.0 $ using different time step sizes.
\begin{figure}[pos=h!]
	\centering
	\includegraphics[width=0.45\textwidth]{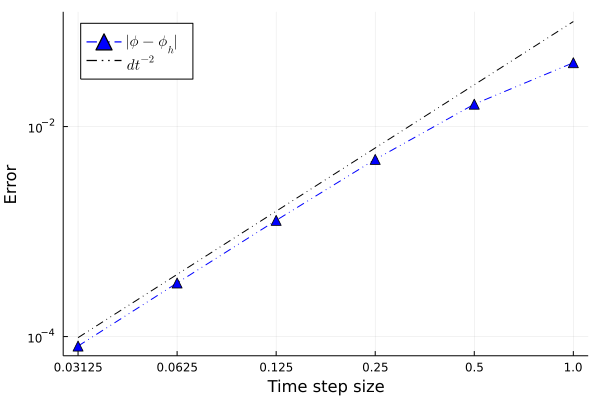}
	\includegraphics[width=0.45\textwidth]{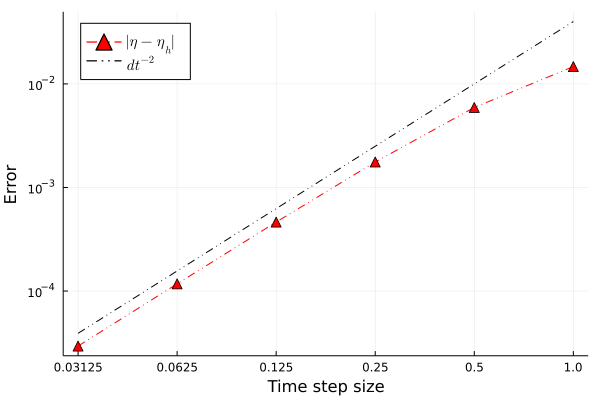}
	\caption{Evolution of the L$2$-norm of the potential error $ e_\phi $ (left), and surface elevation error $ e_\eta $ (right) for different time step sizes.}	
	\label{fig:convergence_dt}
\end{figure}
Again, Figure \ref{fig:convergence_dt} shows that the convergence rate of the solution is $ \mathcal{O}(\Delta t^2) $, as prescribed by the Newmark-beta method with the pair of parameters $ \gamma=0.5 $ and $ \beta=0.25 $.

\subsubsection{Energy conservation}
We also assess the energy conservation properties of the proposed approach. Here we evaluate the relative energy error, $ e_E\eqdef|\Et-\Eth|/\Et $ with $ \Eth $ the total energy computed using the discrete solution, during $ 10 $ wave periods. That is,a final time $ t=10T=5.18 $s. The total energy of the initial condition can be computed from the kinetic, potential and elastic contributions from the velocity potential, free surface elevation and beam deflection, \textit{i.e.}
\begin{align*}
	\Ekf=&\frac{1}{2}\|\nabla\phi\|_{\Omega}^2 = \frac{1}{4}g\eta_0^2L,\\
	\Eks=&\frac{1}{2}d_0\|\eta_t\|_{\Gs}^2=\frac{1}{4}d_0\omega^2\eta_0^2L,\\
	\Epf=&\frac{1}{2}g\|\eta\|_{\Gf\cup\Gs}^2=\frac{1}{4}g\eta_0^2L,\\
	\Ees=&\frac{1}{2}\Drho\|\Delta\eta\|_{\Gs}^2=\frac{1}{4}\Drho k^4\omega^2\eta_0^2L,\\
	\Et\eqdef&\Ekf+\Eks+\Epf+\Ees=\frac{1}{2}g\eta_0^2L+\frac{1}{4}\eta_0^2L\left(d_0\omega^2+\Drho k^4\right) = \frac{1}{2}(g+d_0\omega^2)\eta_0^2L.
\end{align*}

We use the same setting as Section \ref{subsec:convergence_space} and we select two different cases modifying the time step size and the mesh size, with: 
\begin{itemize}
	\item case 1: $ n_x=\left\{16,32,64,128\right\} $, $ r=4 $ and $ \Delta t=1.0e\text{-}3  $;
	\item case 2: $ n_x=128 $, $ r=4 $, $ \Delta t=\left\{10T/4, 10T/8, 10T/16, 10T/32, 10T/64\right\} $ with a final time of $ t=T $.
	\end{itemize}
\begin{figure}[pos=h!]
	\centering
	\includegraphics[width=0.5\textwidth]{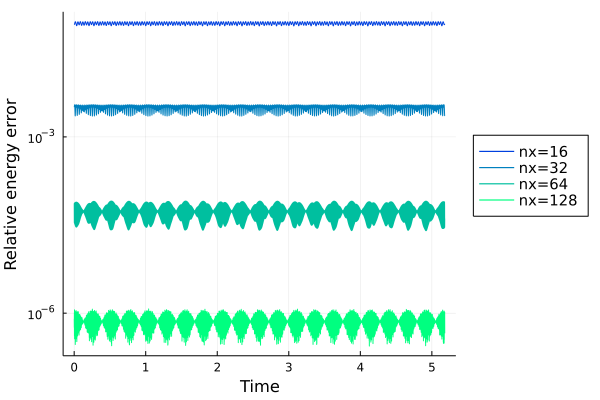}
	\includegraphics[width=0.45\textwidth]{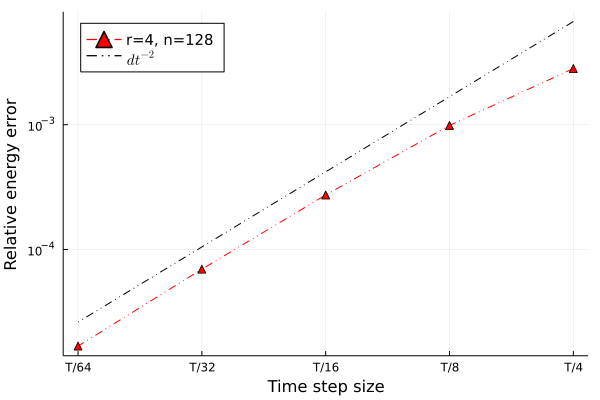}
	\caption{Relative error, $e_E$, evolution in time for case 1 (left) and error convergence with time step size for case 2 (right). }	
	\label{fig:energy_error}
\end{figure} 
In Figure~\ref{fig:energy_error} we depict the evolution of the relative energy error in time for the case 1 (Figure~\ref{fig:energy_error} left) and the convergence of the error with respect to the time step size (Figure~\ref{fig:energy_error} right). We see that when refining the mesh the energy error decreases. We also see that the relative error does not increase as time evolves, denoting that energy is conserved with the proposed formulation. We note that the oscillatory behaviour in time seen in the figure on the left is caused by the numerical error introduced by the gradient jump terms appearing in equation~\eqref{eq:cgdg_form}. This oscillations can be reduced by increasing the penalty parameter $ \gamma $, which is set to $ \gamma=10.0r(r+1) $ in this test. We also see that when refining the time step size, increasing the order and increasing the number of elements, the total energy error is reduced. 

\subsubsection{Energy conservation in a finite beam}
In the previous subsection we have assessed the energy conservation properties of an infinite floating beam. Here we extend this analysis to the case where we have a boundary composed by a finite beam and a free surface. To this end, we use the same periodic setting as in the previous section, but with a beam of size $ L_b=\pi $ located at the center of the domain.
\begin{figure}[pos=h!]
	\centering
	\includegraphics[width=0.5\textwidth]{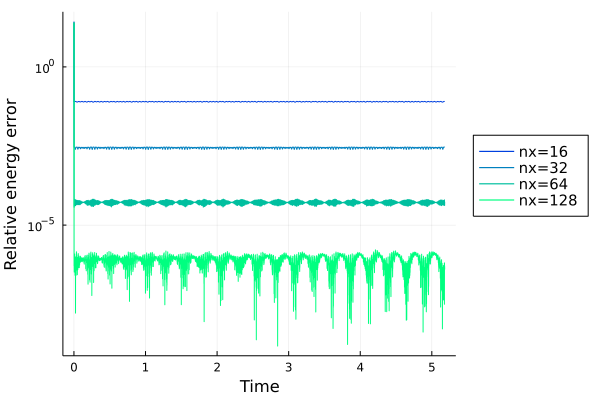}
	\includegraphics[width=0.45\textwidth]{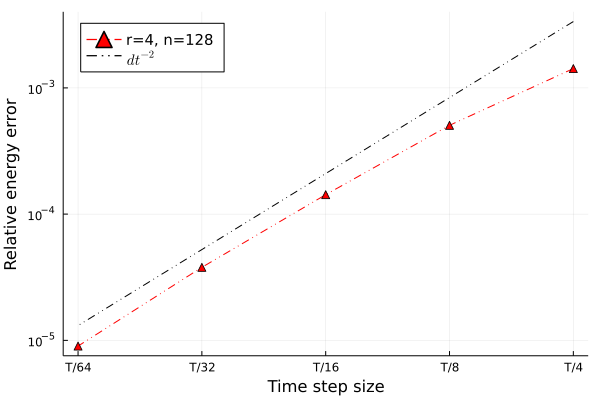}
	\caption{Relative error, $e_E$, evolution in time for case 1 (left) and error convergence with time step size for case 2 (right). }	
	\label{fig:energy_error_finite_beam}
\end{figure} 
Again, in Figure~\ref{fig:energy_error_finite_beam}  we see that error is not increasing in time and decreases when we refine the mesh (left), while it also converges with the expected rate as the time step is decreased (right).

\subsection{Floating beam with elastic joint} \label{sec:elastic_joint}
Once analysed the behaviour of the proposed formulation in time domain for infinite and finite beams, we now proceed to assess the formulation for the case of a floating beam with a joint and varying stiffness. Here we will solve the setting proposed by \cite{khabakhpasheva2002hydroelastic} and also tested by \cite{riyansyah2010connection} in the frequency domain. The geometry of this test is given in Figure \ref{fig:khabakhpasheva_geometry} and the input parameters are defined in Table \ref{tab:khabakhpasheva_properties}.
\begin{figure}[pos=h!]
	\centering
	\includegraphics[width=0.6\textwidth]{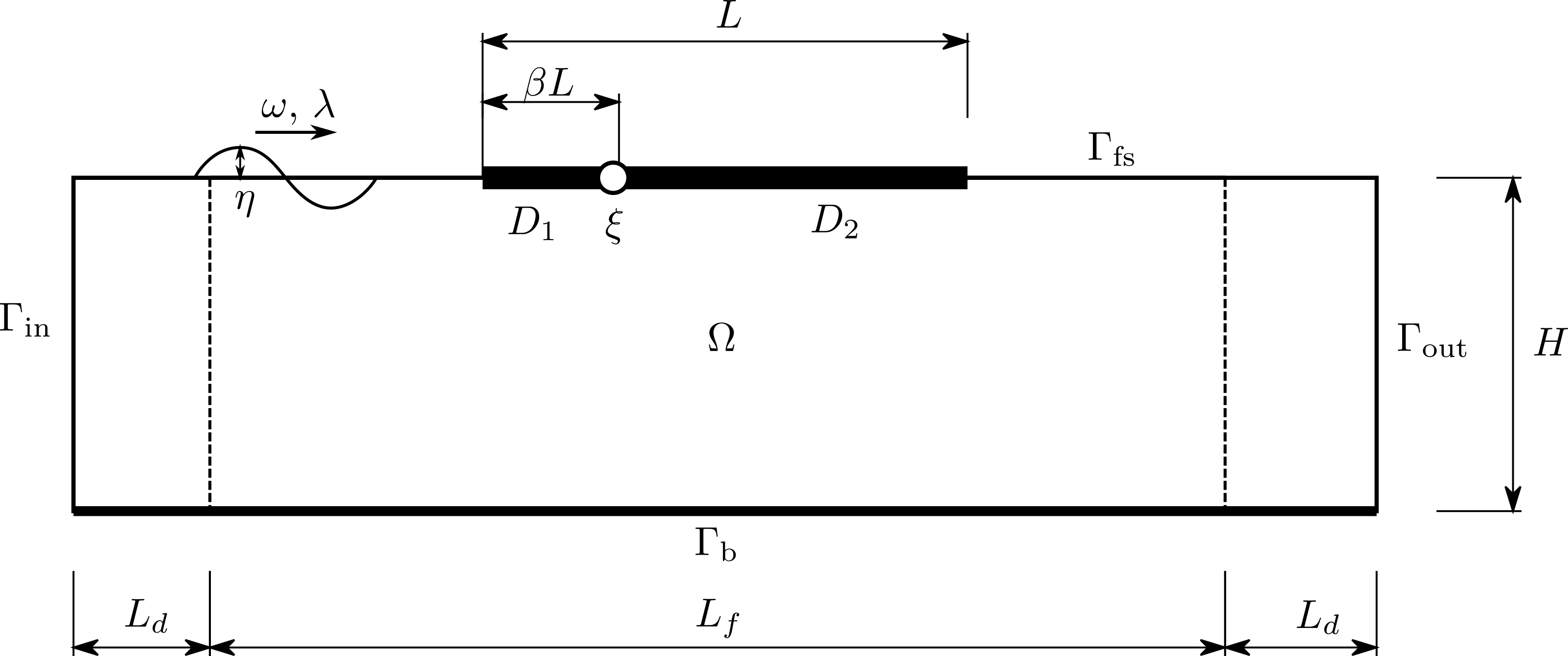}
	\caption{Sketch with the definition of the geometry used in \cite{khabakhpasheva2002hydroelastic}.}	
	\label{fig:khabakhpasheva_geometry}
\end{figure}
\begin{table}[pos=h]
	\centering
	\caption{Khabakhpasheva \textit{et al.} test parameters.}
	\label{tab:khabakhpasheva_properties}
	\begin{tabular}{lccc}
		Parameter&Symbol&Value&Units\\ \hline
		Draft&$d_0$&$8.1561e\text{-}3$&$\mbox{m}$\\
		Structure length&$L$&$12.5$&$\mbox{m}$\\
		Fluid domain length&$L_f$&$25$&$\mbox{m}$\\
		Tank depth&$H$&$1.1$&$\mbox{m}$\\
		Connection location parameter&$ \beta $&$ 0.2 $& - \\
		Connection rotational stiffness parameter&$ \xi $& $ 0 $ and $ 625 $& - \\
		Structure 1 rigidity&$D_1$&$47100$&$\mbox{N m}$\\
		Structure 2 rigidity&$D_2$&$471$&$\mbox{N m}$\\
		Gravity acceleration&$g$&9.81&$\mbox{m}/\mbox{s}^2$\\
		Wavelength-to-beam length ratio&$\alpha$&0.249&-\\
	\end{tabular}
\end{table}
In this test we use the formulation for floating beams defined in~\eqref{eq:bilinear_modified_h_cgdg_beam}, with the following kinematic boundary conditions:
\begin{subequations}
	\label{eq:kinematic_bc_khabakhpasheva}
	\begin{align}
		\label{eq:kinematic_bc_b_khabakhpasheva}
		\n\cdot\nphi&=0\quad\mbox{on }\Gb,\\
		\label{eq:kinematic_bc_i_khabakhpasheva}
		\n\cdot\nphi&=-\omega\eta_0\frac{\cosh(k_\lambda y)}{\sinh(k_\lambda H)}\cos(k_\lambda x-\omega t)\quad\mbox{on }\Gi,\\
		\label{eq:kinematic_bc_o_khabakhpasheva}
		\n\cdot\nphi&=0\quad\mbox{on }\Go.
	\end{align}
\end{subequations}

Condition \eqref{eq:kinematic_bc_i_khabakhpasheva} enforces an incoming wave, as defined by \eqref{eq:exact}, on $ \Gi $. For this problem an incoming wave length of $ \lambda\eqdef\alpha L $ is defined, with $ \alpha=0.249 $, and a wave frequency of $ \omega=\sqrt{gk_\lambda\tanh(k_\lambda H)} $. In addition, we define a damping zone at the inlet of the tank and at the outlet of the tank of length $ L_d=L\approx4\lambda $, where damping terms are added to the free surface dynamic and  kinematic boundary conditions according to \cite{kim2014numerical}, resulting in 
\begin{subequations}\label{eq:damping_bcs}
	\begin{align}\label{eq:kinematic_bc_fs_damping}
		&\n\cdot\nabla\phi=\eta_t+\mu_2(\eta-\eta^*)\qquad\mbox{on }\Gf,\\
		&\phi_t+g\eta+\mu_1(\nabla\phi\cdot\n-\nabla\phi^*\cdot\n)=0\qquad\mbox{on }\Gf.
	\end{align}
\end{subequations}
With 
\begin{align*}
\mu_1(x) &=\begin{cases}
\mu_0\left[1-\sin\left(\frac{\pi}{2}\frac{x}{L_d}\right)\right] & \mbox{if }x_{d,in}<x,\\
\mu_0\left[1-\cos\left(\frac{\pi}{2}\frac{x-x_d}{L_d}\right)\right] & \mbox{if }x>x_{d,out},\\
0&\mbox{otherwise},
\end{cases} \\
\mu_2(x)&=k_\lambda\mu_1(x).
\end{align*}
We select $ \mu_0 = 2.5 $. The variables $ \phi^* $ and $ \eta^* $ are the values that we want to enforce at each damping zone, these are given by equation~\eqref{eq:exact} at the inlet and zero at the outlet.

The joint rotational stiffness is parametrized by an adimensional parameter, $ \xi $, such that $ k_\rho=\xi\Drho/L $. We consider two cases : a first case with a hinge, \textit{i.e.}  $ \xi=0 $, and a second case with a stiff elastic joint with $ \xi=625 $. 
In this test we use a mesh with elements of 4th order. We define 20 elements through the beam, \textit{i.e} 80 elements in the horizontal direction, and 5 elements in the vertical direction with exponential refinement close to the free surface.

\subsubsection{Results in frequency domain}
The problem sketched in \figurename~\ref{fig:khabakhpasheva_geometry} is first solved using the frequency domain approach, as defined in Section~\ref{subsubsec:freq_domain}. In Figure \ref{fig:khabakhpasheva_experiment} we plot the relative surface elevation, $ \eta/\eta_0 $ along the beam for the two cases, comparing with the results given by \cite{khabakhpasheva2002hydroelastic} and \cite{riyansyah2010connection}. We can see that for both cases, $ \xi=0 $ and $ \xi=635 $, the results that we obtain with the proposed monolithic scheme are in very good agreement with the other two works.
\begin{figure}[pos=h!]
	\centering
	\includegraphics[width=0.49\textwidth]{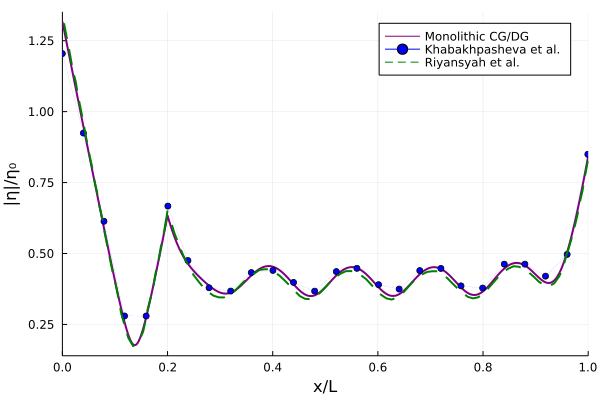}
	\includegraphics[width=0.49\textwidth]{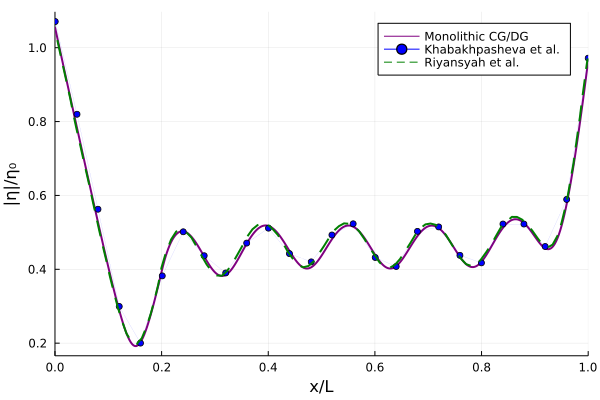}
	\caption{Relative surface elevation at the beam for the hinged case, $ \xi=0 $, (left) and the stiff joint, $ \xi=625 $, (right).}	
	\label{fig:khabakhpasheva_experiment}
\end{figure}

\begin{figure}[pos=h!]
\centering
\includegraphics[width=0.7\textwidth]{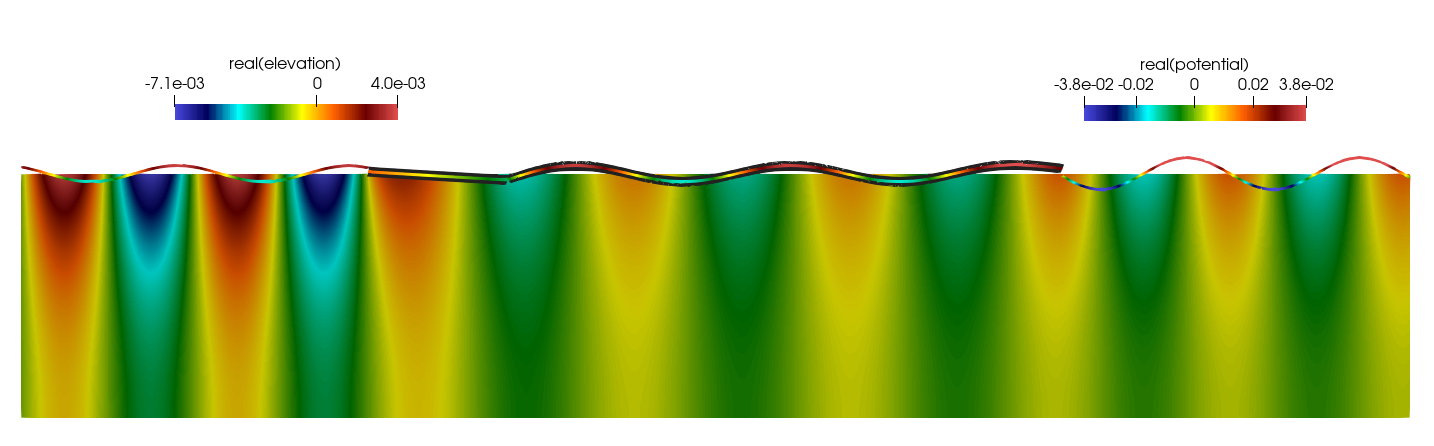}\\
\includegraphics[width=0.7\textwidth]{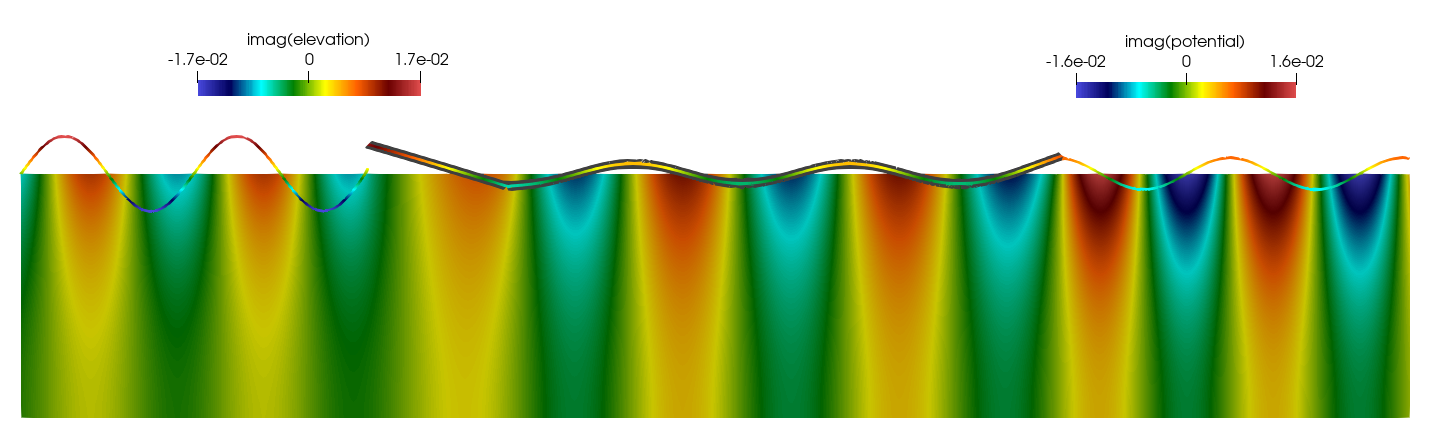}\\
\includegraphics[width=0.7\textwidth]{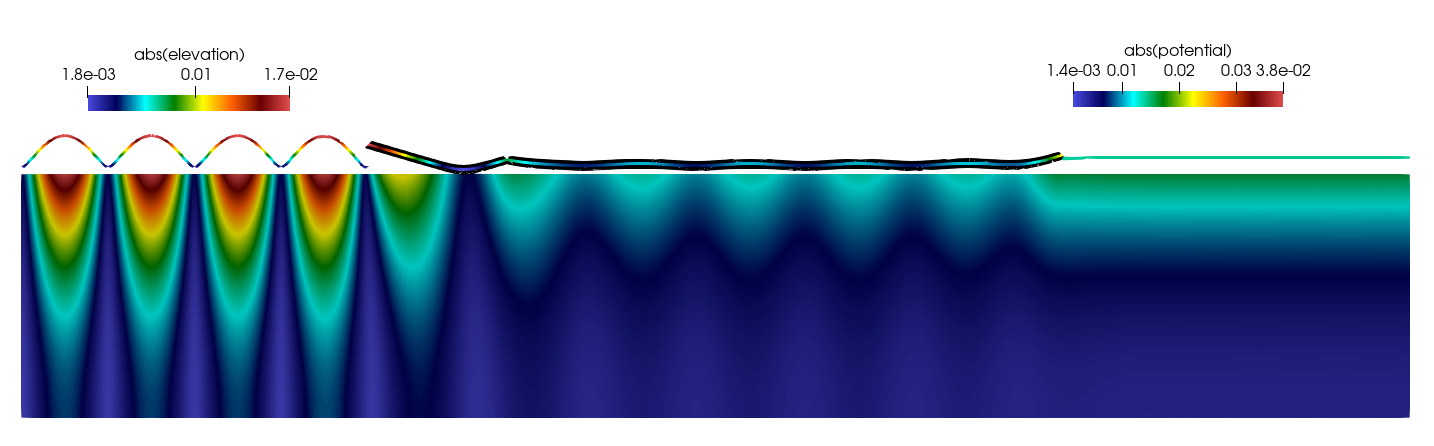}
\caption{Velocity potential, $ \phi_{h} $, and surface elevation, $ \eta_h $, for the case $ \xi=0 $. Real part (top), imaginary part (center) and absolute values (bottom). The vertical direction of the domain is scaled 4:1 and the surface elevation is scaled by 40. The beam region is shadowed in black.}	
\label{fig:khabakhpasheva_experiment_solution_xi0}
\end{figure} 
\begin{figure}[pos=h!]
\centering
\includegraphics[width=0.7\textwidth]{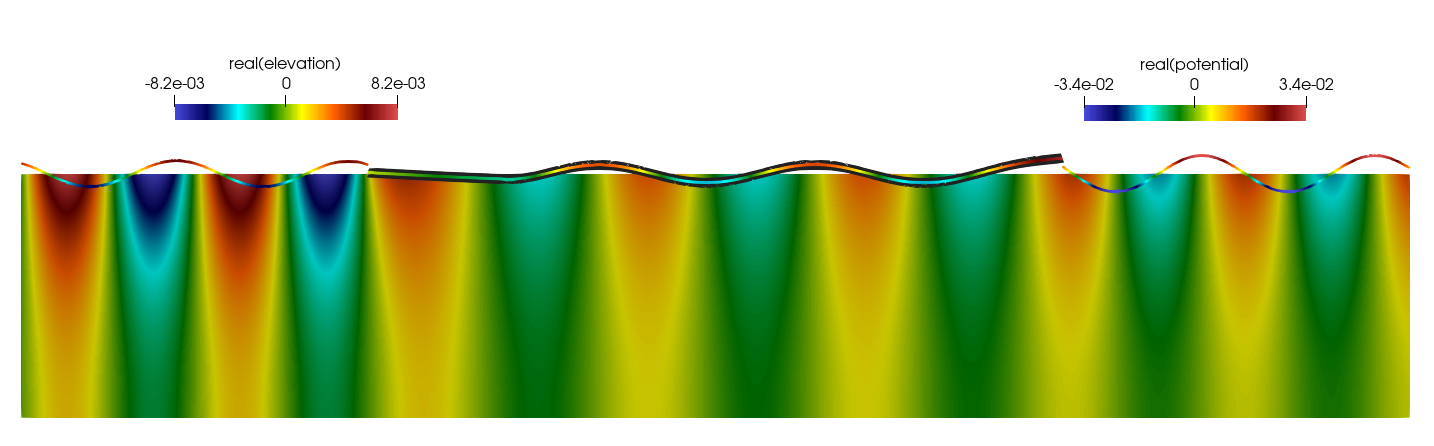}\\
\includegraphics[width=0.7\textwidth]{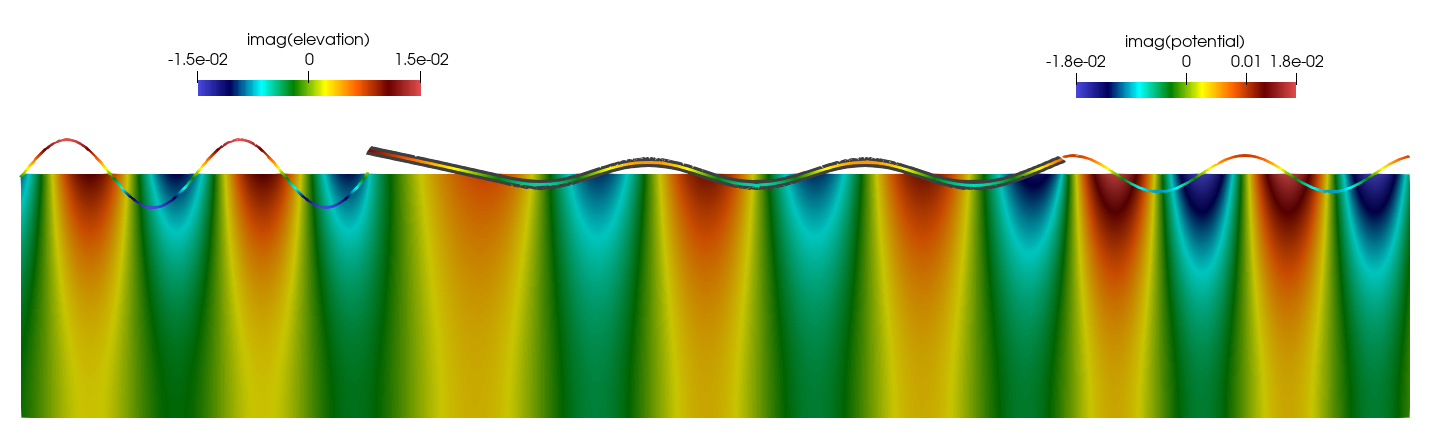}\\
\includegraphics[width=0.7\textwidth]{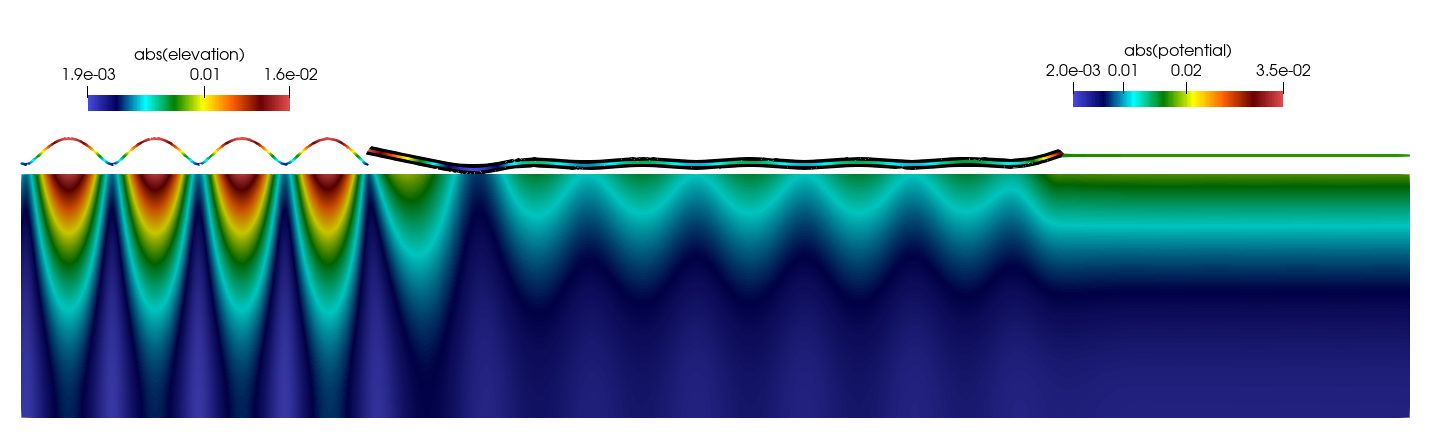}
\caption{Velocity potential, $ \phi_{h} $, and surface elevation, $ \eta_h $, for the case $ \xi=625 $. Real part (top), imaginary part (center) and absolute values (bottom). The vertical direction of the domain is scaled 4:1 and the surface elevation is scaled by 40. The beam region is shadowed in black.}
\label{fig:khabakhpasheva_experiment_solution_xi635}	
\end{figure} 
In Figures~\ref{fig:khabakhpasheva_experiment_solution_xi0} and \ref{fig:khabakhpasheva_experiment_solution_xi635} we show the real, imaginary and absolute values of the velocity potential and surface elevation fields for the case $ \xi=0 $ and $ \xi=625 $, respectively, excluding the damping zones. It is seen that the wave is stretched under the structure, specially in the most stiffer zone at the left of the joint. We also see that the hinged case, Figure~\ref{fig:khabakhpasheva_experiment_solution_xi0}, results in higher reflection at the left of the structure and a smaller transmitted wave amplitude at the right of the structure compared to the elastic joint, Figure~\ref{fig:khabakhpasheva_experiment_solution_xi635}.

\subsubsection{Results in time domain}
In this section we solve the floating beam with elastic joint test in the time domain. Here we will assess the behavior of the proposed formulation, as given in~\eqref{eq:bilinear_modified_h_t_cgdg_beam}. We use the same setting as defined in the frequency domain case. The problem is solved for $ t=[0,50T] $, where $ T=2\pi/\omega $,  with a time step size of $ \Delta t=T/40 $.

In \figurename~\ref{fig:khabakhpasheva_time_experiment} we plot the envelope of the normalized absolute value of the beam deflection for the two cases, $ \xi=0 $ and $ \xi=624 $, comparing with results from literature, \textit{i.e.} \cite{khabakhpasheva2002hydroelastic,riyansyah2010connection}. The envelope is computed accounting only for the results from $ t=[25T,50T] $, to avoid the transient effects from the initial stages of the simulation. In the same figure, we also depict the absolute value of the normalized beam deflection at different times, $ t=\{35.716,35.895,36.074,36.252\} $, to visualize the beam deformation along time.
\begin{figure}[pos=h!]
	\centering
	\includegraphics[width=0.49\textwidth]{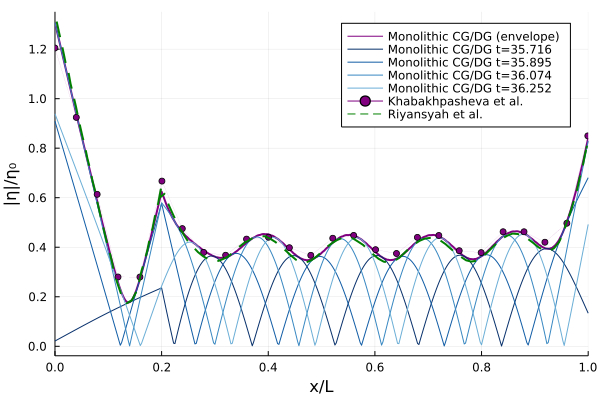}
	\includegraphics[width=0.49\textwidth]{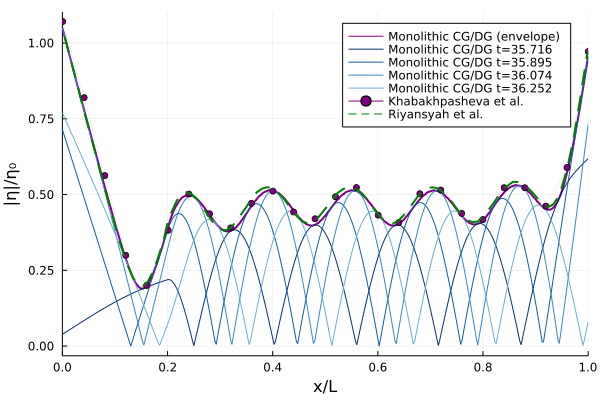}
	\caption{Relative surface elevation at the beam for the hinged case, $ \xi=0 $, (left) and the stiff joint, $ \xi=625 $, (right).}	
	\label{fig:khabakhpasheva_time_experiment}
\end{figure} 

Again, in \figurename~\ref{fig:khabakhpasheva_time_experiment} we see that the results obtained with the proposed monolithic formulation are in very good agreement with the results appearing in the literature. We can also clearly observe the effect of the joint and different beam rigidities. 

In Figures~\ref{fig:khabakhpasheva_experiment_solution_xi0_time} and \ref{fig:khabakhpasheva_experiment_solution_xi625_time} we plot the velocity potential and surface elevation fields at different times, $ t={25.2\mbox{s},25.6\mbox{s},26.0\mbox{s}} $. Looking at the color scales of both figures, which are bounded by the overall simulation time maxima and minima, we can see that the hinged case, $ \xi=0 $ results in larger wave elevations at the front of the platform. This is caused by higher reflected wave amplitude, resulting in a smaller transmitted wave. Another phenomena that can be observed in Figures~\ref{fig:khabakhpasheva_experiment_solution_xi0_time} and \ref{fig:khabakhpasheva_experiment_solution_xi625_time} is the wave stretching under the platform, where we see that the wavelength is greater than the incoming wavelength.
\begin{figure}[pos=h!]
	\centering
	\includegraphics[width=0.7\textwidth]{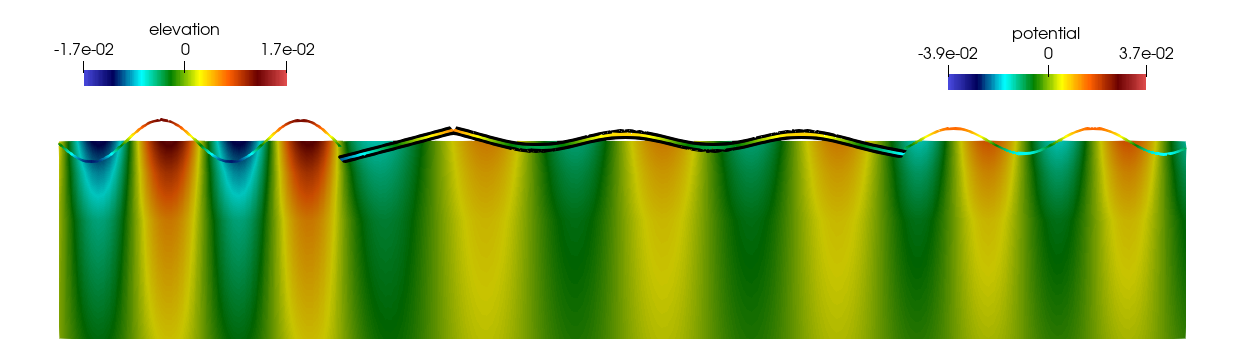}\\
	\includegraphics[clip=true,trim=0 0 0 4cm,width=0.7\textwidth]{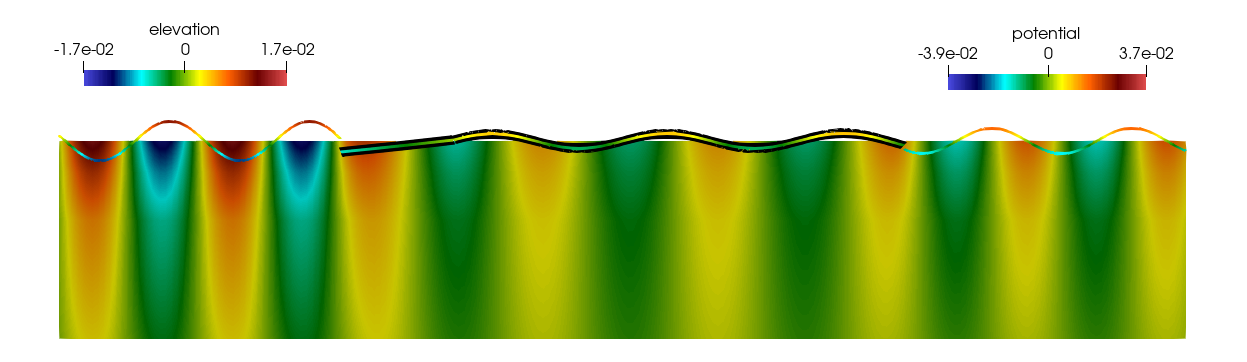}\\
	\includegraphics[clip=true,trim=0 0 0 4cm,width=0.7\textwidth]{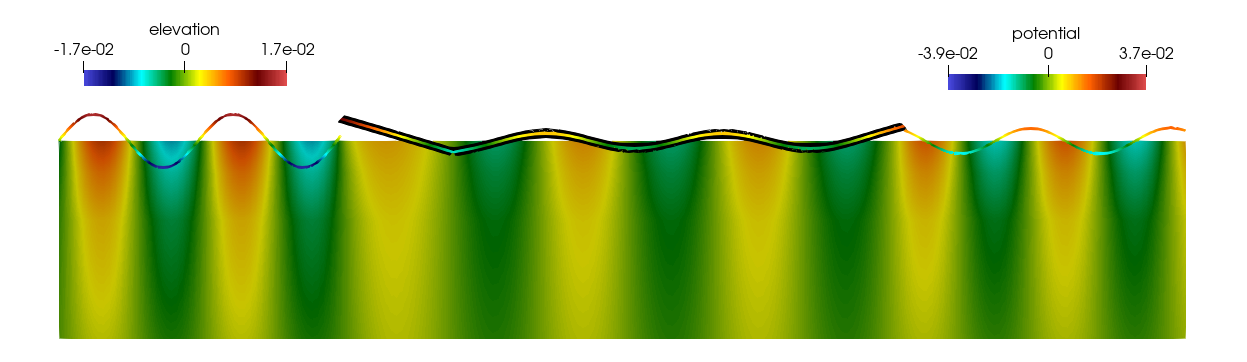}
	\caption{Velocity potential, $ \phi_{h} $, and surface elevation, $ \eta_h $, for the case $ \xi=0 $ at $ t=25.2 $s (top), $ t=25.6 $s (center) and $ t=26.0 $s (bottom). The vertical direction of the domain is scaled 4:1 and the surface elevation is scaled by 40. The beam region is shadowed in black.}	
	\label{fig:khabakhpasheva_experiment_solution_xi0_time}
\end{figure} 
\begin{figure}[pos=h!]
	\centering
	\includegraphics[width=0.7\textwidth]{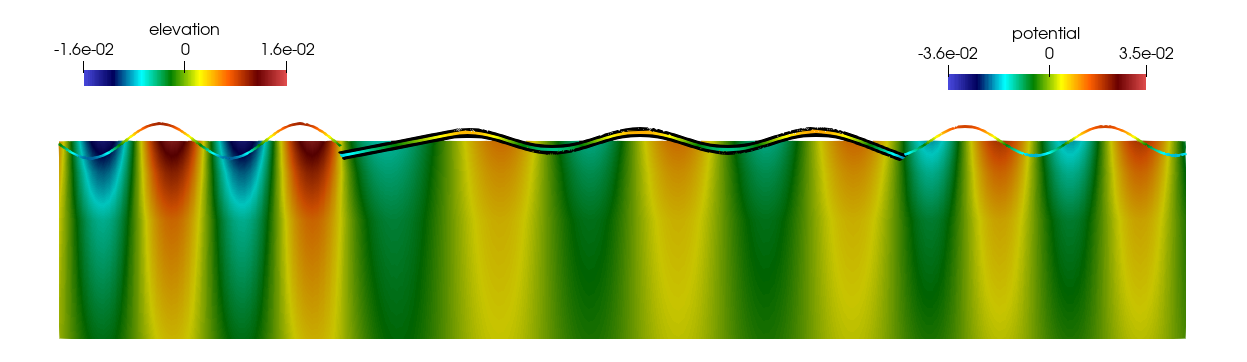}\\
	\includegraphics[clip=true,trim=0 0 0 4cm,width=0.7\textwidth]{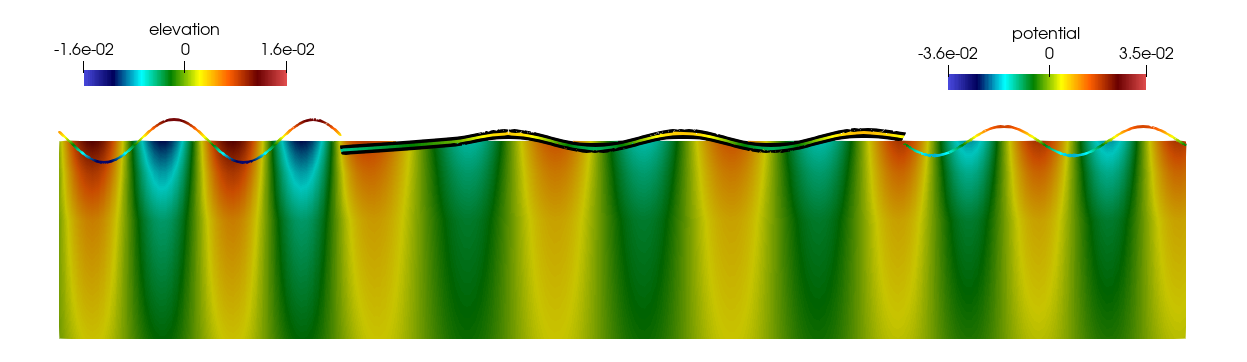}\\
	\includegraphics[clip=true,trim=0 0 0 4cm,width=0.7\textwidth]{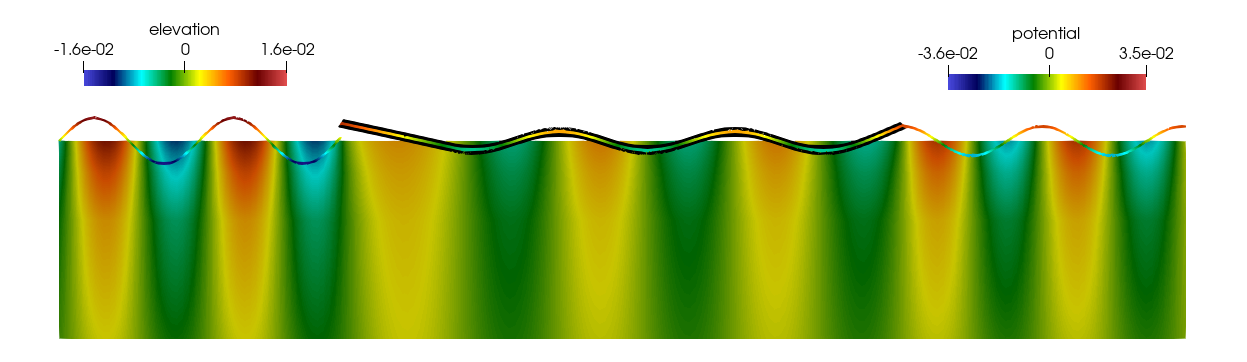}
	\caption{Velocity potential, $ \phi_{h} $, and surface elevation, $ \eta_h $, for the case $ \xi=625 $ at $ t=25.2 $s (top), $ t=25.6 $s (center) and $ t=26.0 $s (bottom). The vertical direction of the domain is scaled 4:1 and the surface elevation is scaled by 40. The beam region is shadowed in black.}	
	\label{fig:khabakhpasheva_experiment_solution_xi625_time}
\end{figure} 

\subsection{Floating beam in irregular sea bed}
In this section we assess the behavior of the proposed approach for a case with non-flat sea bed. We show that the formulation defined in Section~\ref{sec:formulation} is not limited to the case of constant bathymetry by analysing the test proposed in~\cite{liu2020dmm}. In addition, we also demonstrate that the proposed formulation is suitable for domains discretized using unstructured grids. This is specially relevant for the case of non-constant bathymetry and/or structures with arbitrary shape.

In particular, here we will solve the test case for a floating beam over a sloping seabed. In this case, the bathymetry is constant over the domain, except for the region where the floating beam is located, where a linearly varying depth with constant slope of $ \beta $ is considered. Hence, the water depth is given by
\begin{equation}\label{eq:liu_water_depth}
	H_x(x)=\begin{cases}
		H_l&\mbox{if }x\leq x_{\scriptsize\mbox{b,l}},\\
		H_l-\frac{x-x_{\scriptsize\mbox{b,0}}}{L}(H_l-H_r)&\mbox{if } x_{\scriptsize\mbox{b,l}} < x < x_{\scriptsize\mbox{b,r}}, \\
		H_r&\mbox{if }x_{\scriptsize\mbox{b,r}\leq x }.
	\end{cases}
\end{equation}
Where $ x_{\scriptsize\mbox{b,l}} $ is the most left coordinate of the beam and $ x_{\scriptsize\mbox{b,r}} $ the beam end point on the right. For clarity, in Figure~\ref{fig:liu_geometry} we plot an sketch of the geometry used in the test. 
\begin{figure}[pos=h!]
	\centering
	\includegraphics[width=0.6\textwidth]{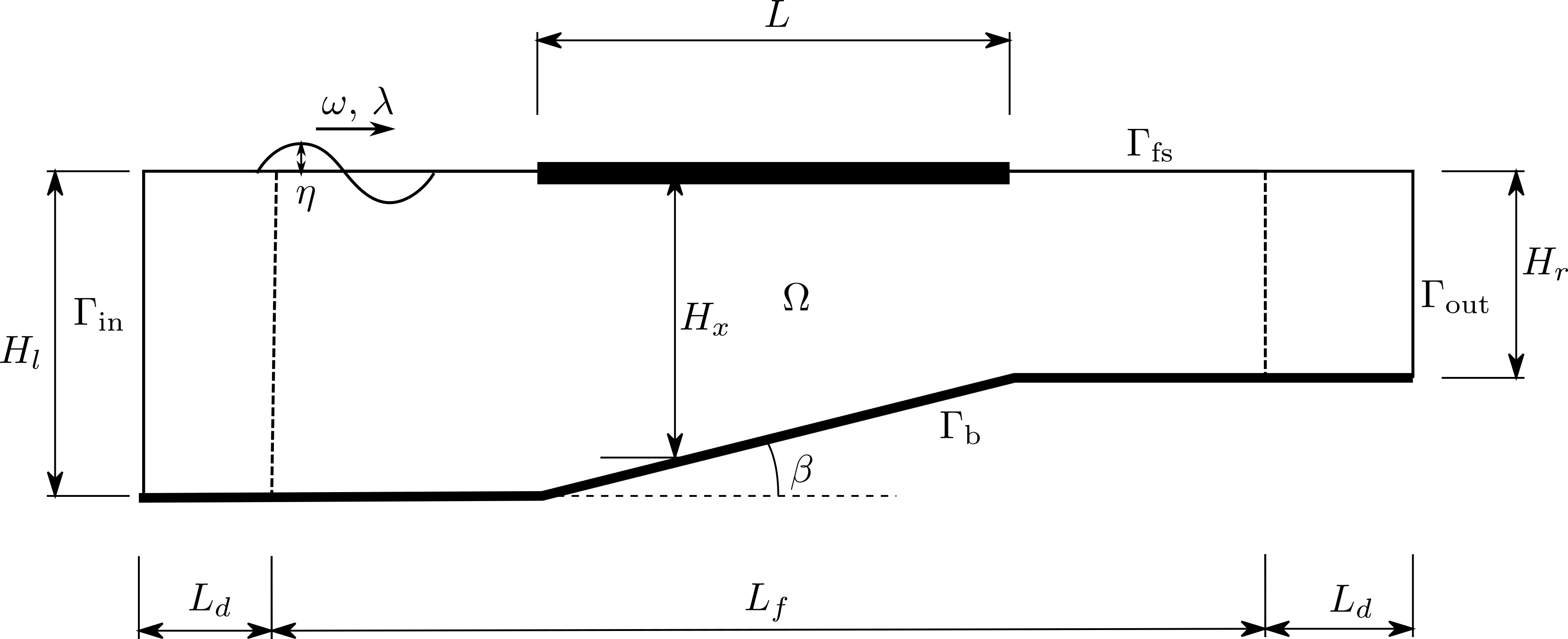}
	\caption{Sketch with the definition of the geometry used in \cite{liu2020dmm}.}	
	\label{fig:liu_geometry}
\end{figure}

The same incoming wave conditions as defined in equations~\eqref{eq:exact}, together with the damping zones defined in equations~\eqref{eq:damping_bcs} are used in this test, with the only difference of the choice of $ \mu_0=10 $. The damping regions at the inlet and outlet are taken as four times the structure length, $ L_d=4L $. The final set of parameters used for this test are summarized in Table~\ref{tab:Liu_properties}.
\begin{table}[pos=h]
	\centering
	\caption{Liu \textit{et al.} test parameters.}
	\label{tab:Liu_properties}
	\begin{tabular}{lccc}
		Parameter&Symbol&Value&Units\\ \hline
		Draft&$d_0$&$0.4878$&$\mbox{m}$\\
		Structure length&$L$&$300$&$\mbox{m}$\\
		Fluid domain length&$L_f$&$1500$&$\mbox{m}$\\
		Tank depth left side&$H_l$&$60$&$\mbox{m}$\\
		Tank depth right side&$H_r$&$30$&$\mbox{m}$\\
		Structure rigidity&$D$&$1.0e10$&$\mbox{N m}$\\
		Gravity acceleration&$g$&9.81&$\mbox{m}/\mbox{s}^2$\\
		Wave frequency&$\omega$&$ 0.4 $ and $ 0.8 $&rad/s\\
	\end{tabular}
\end{table}

As previously mentioned, the geometry is discretized with an unstructured grid. We define an element size of $ h=\frac{L}{50}=6 $m at the free surface and a characteristic element size of $ h=\frac{L}{25}=12 $m at the sea bed. A close-up view of the mesh used around the floating beam is depicted in Figure~\ref{fig:liu_mesh}.
\begin{figure}[pos=h!]
	\centering
	\includegraphics[width=0.8\textwidth,clip=true,trim=0 0 0 3.9cm]{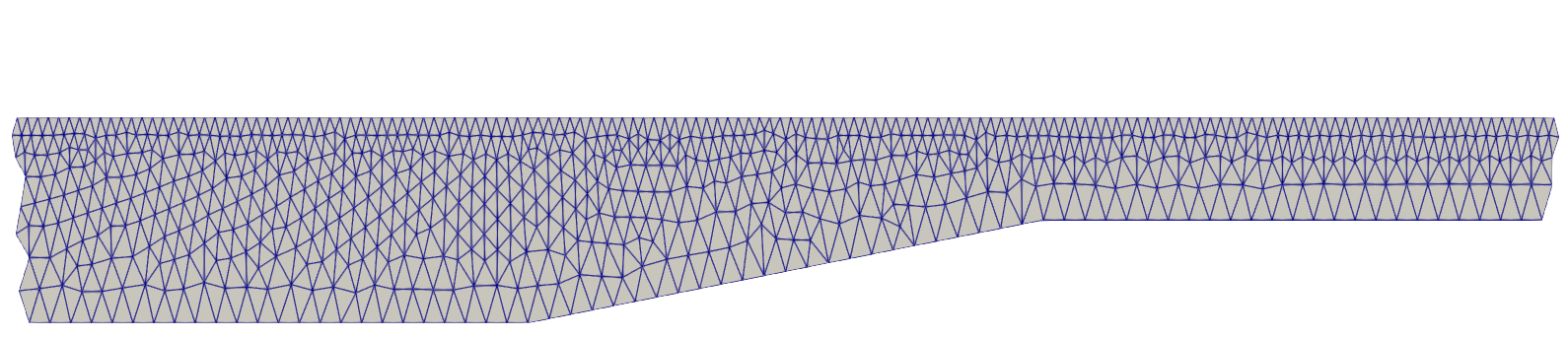}
	\caption{Close-up view of the mesh used to solve the Liu \textit{et al.} test. The vertical direction is scaled 4:1.}	
	\label{fig:liu_mesh}
\end{figure}

In Figure~\ref{fig:Liu_experiment} we plot the normalized surface elevation for the cases $ \omega=0.4 $ and $ \omega=0.8 $. We compare the results with those reported in \cite{liu2020dmm} for the same cases. It is observed that the results of the proposed monolithic formulation match very well the results from literature.
\begin{figure}[pos=h!]
	\centering
	\includegraphics[width=0.49\textwidth]{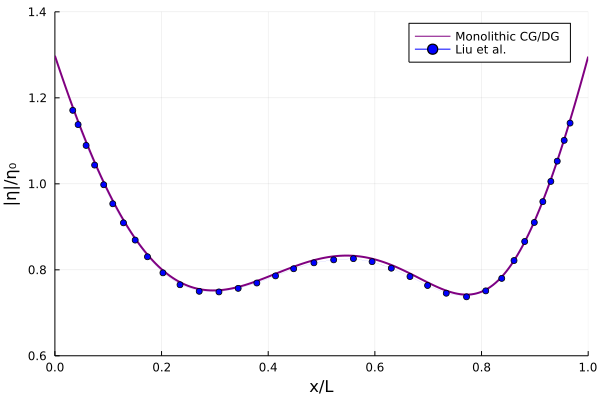}
	\includegraphics[width=0.49\textwidth]{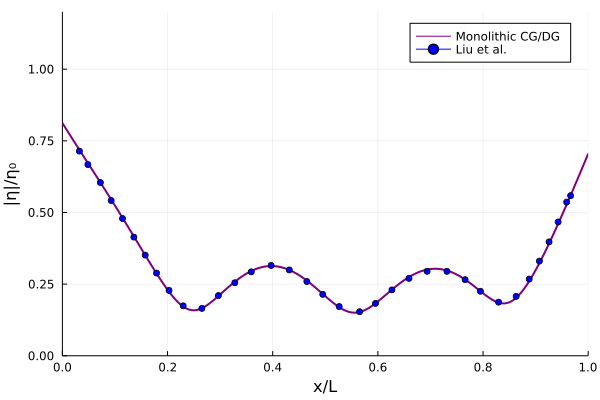}
	\caption{Relative surface elevation at the beam for the case $ \omega=0.4 $ (left) and $ \omega=0.8 $ (right).}	
	\label{fig:Liu_experiment}
\end{figure} 

To have a better understanding of the behavior of the floating beam, we also plot the real, imaginary and absolute values of the velocity potential and the surface elevation fields, see Figure~\ref{fig:liu_experiment_solution_omega04} for the case $ \omega=0.4 $ and Figure~\ref{fig:liu_experiment_solution_omega08} for the case $ \omega=0.8 $.
\begin{figure}[pos=h!]
	\centering
	\includegraphics[width=0.7\textwidth]{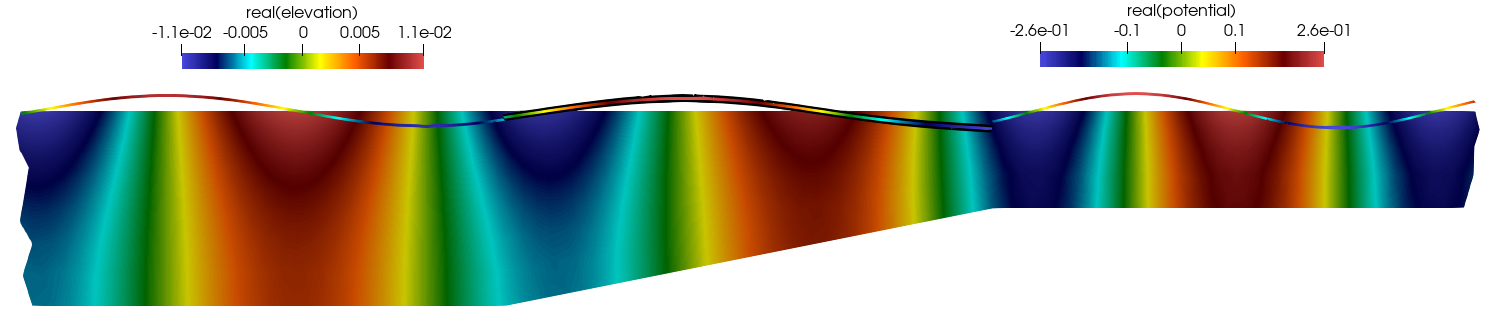}\\
	\includegraphics[width=0.7\textwidth]{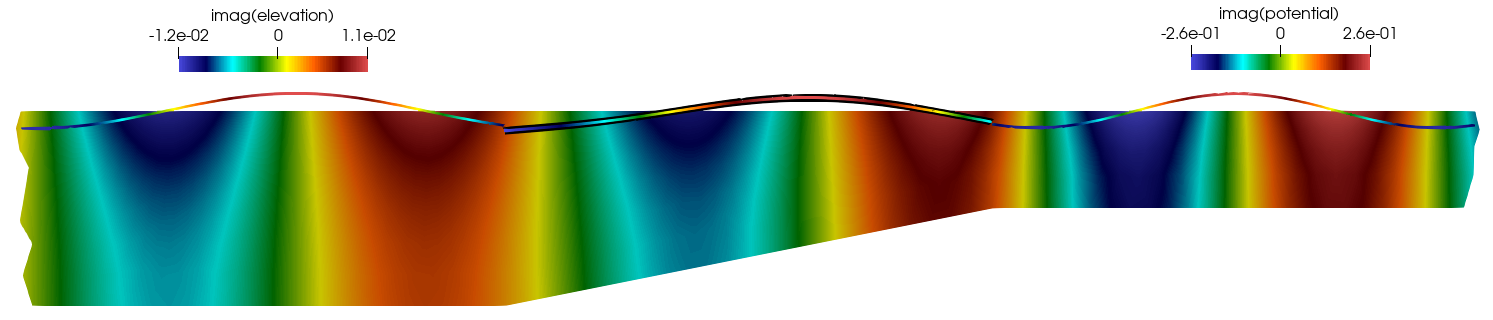}\\
	\includegraphics[width=0.7\textwidth]{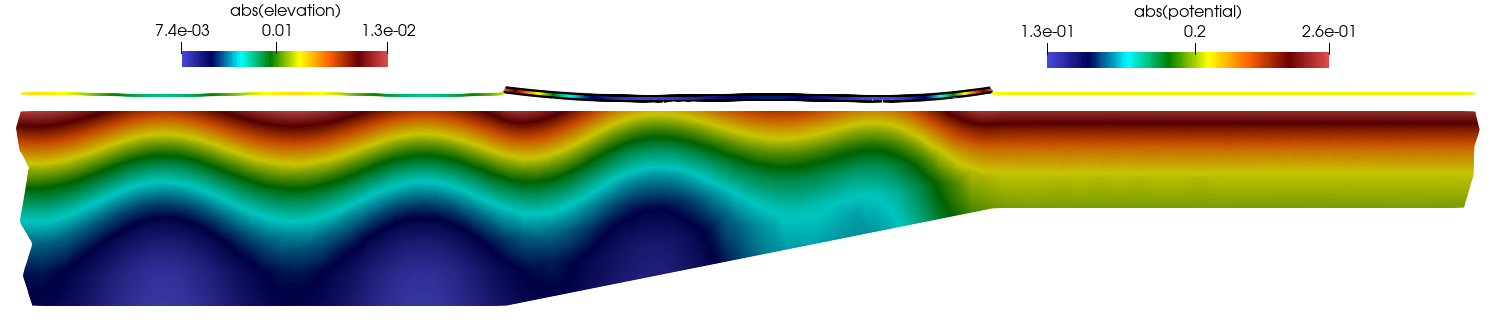}
	\caption{Close-up view of the velocity potential, $ \phi_{h} $, and surface elevation, $ \eta_h $, for the case $ \omega=0.4 $. Real part (top), imaginary part (center) and absolute values (bottom). The vertical direction of the domain is scaled 4:1 and the surface elevation is scaled by 1000. The beam region is shadowed in black.}
	\label{fig:liu_experiment_solution_omega04}	
\end{figure} 
\begin{figure}[pos=h!]
	\centering
	\includegraphics[width=0.7\textwidth]{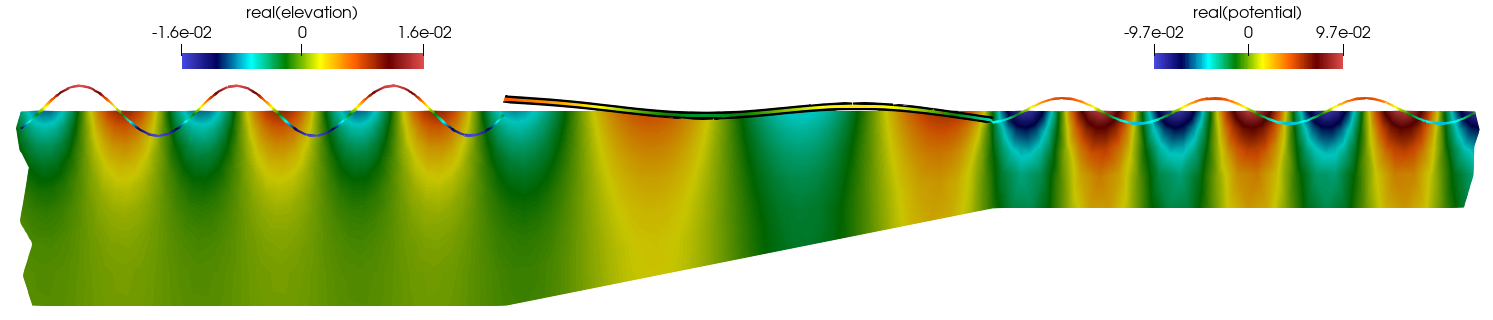}\\
	\includegraphics[width=0.7\textwidth]{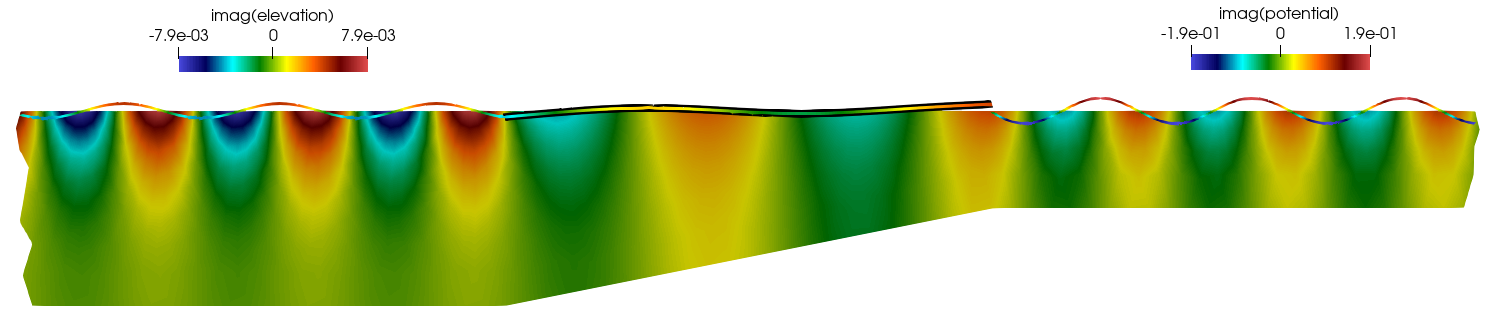}\\
	\includegraphics[width=0.7\textwidth]{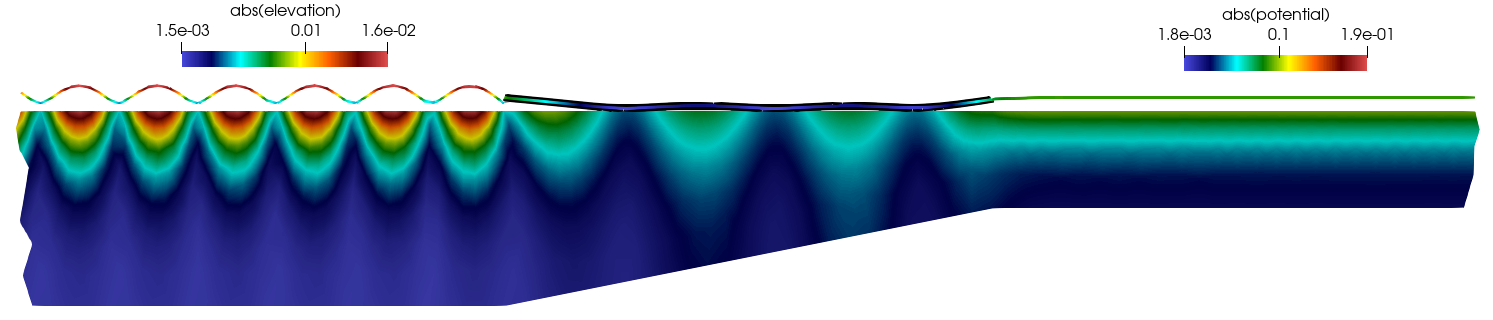}
	\caption{Close-up view of the velocity potential, $ \phi_{h} $, and surface elevation, $ \eta_h $, for the case $ \omega=0.8 $. Real part (top), imaginary part (center) and absolute values (bottom). The vertical direction of the domain is scaled 4:1 and the surface elevation is scaled by 1000. The beam region is shadowed in black.}
	\label{fig:liu_experiment_solution_omega08}	
\end{figure} 

\subsection{Floating plate}
Once analysed the formulation for an infinite-dimensional floating beam, we assess the behaviour of the novel method for the simulation of finite floating thin 2-dimensional structures in  3-dimensional domains. Here we use the setting used in the test reported in the numerical study of Fu \textit{et al.}~\cite{fu2007hydroelastic}, based on the experimental study of Yago \textit{et al.}~\cite{yago1996hydoroelastic} were a scaled mat-like structure model is assessed in a  wave tank. In Table \ref{tab:yago_properties} we summarize the parameter values used in this test.
 \begin{table}[pos=h]
	\centering
	\caption{Yago \textit{et al.} test parameters.}
	\label{tab:yago_properties}
	\begin{tabular}{lccc}
		Parameter&Symbol&Value&Units\\ \hline
		Structure length&$L$&$300$&$\mbox{m}$\\
		Structure width&$B$&$60$&$\mbox{m}$\\
		Structure height&$h_b$&$2$&$\mbox{m}$\\
		Structure density&$\rho_b$&$256.25$&$\mbox{kg}/\mbox{m}^3$\\
		Structure Draft&$d_0$&$0.5$&$\mbox{m}$\\
		Structure Young modulus&$E$&$1.19e10$&$\mbox{Pa}$\\
		Structure Poisson coefficient&$\nu$&$0.13$&$-$\\
		Tank length&$2L_d+L_f$&$3000$&$\mbox{m}$\\
		Tank width&$B_f$&$840$&$\mbox{m}$\\
		Tank depth &$H$&$58.5$&$\mbox{m}$\\
		Gravity acceleration&$g$&9.81&$\mbox{m}/\mbox{s}^2$\\
		Wavelength&$\lambda$&$ 0.4L $, $ 0.6L $ and $ 0.8L $&$ \mbox{m} $\\
	\end{tabular}
\end{table}

The computational domain, sketched in Figure \ref{fig:yago_geometry}, has size $ (2L_d+L_f)\times B_f\times H $, with the front edge of the structure located at a distance $ 4.5L $ from the inlet, and the plate side edges at a distance $ 6.5B $ from the left wall of the tank. At the inlet and outlet of the tank we define a damping region of length $ L_d=4\lambda $ with the same damping terms as described in equation~\eqref{eq:damping_bcs}. In this case, we select $ \mu_0=6.0 $ and the variables $ \phi^* $ and $ \eta^* $ are given by equation~\eqref{eq:exact_3d} at the inlet and zero at the outlet.
\begin{figure}[pos=h!]
	\centering
	\includegraphics[width=0.6\textwidth]{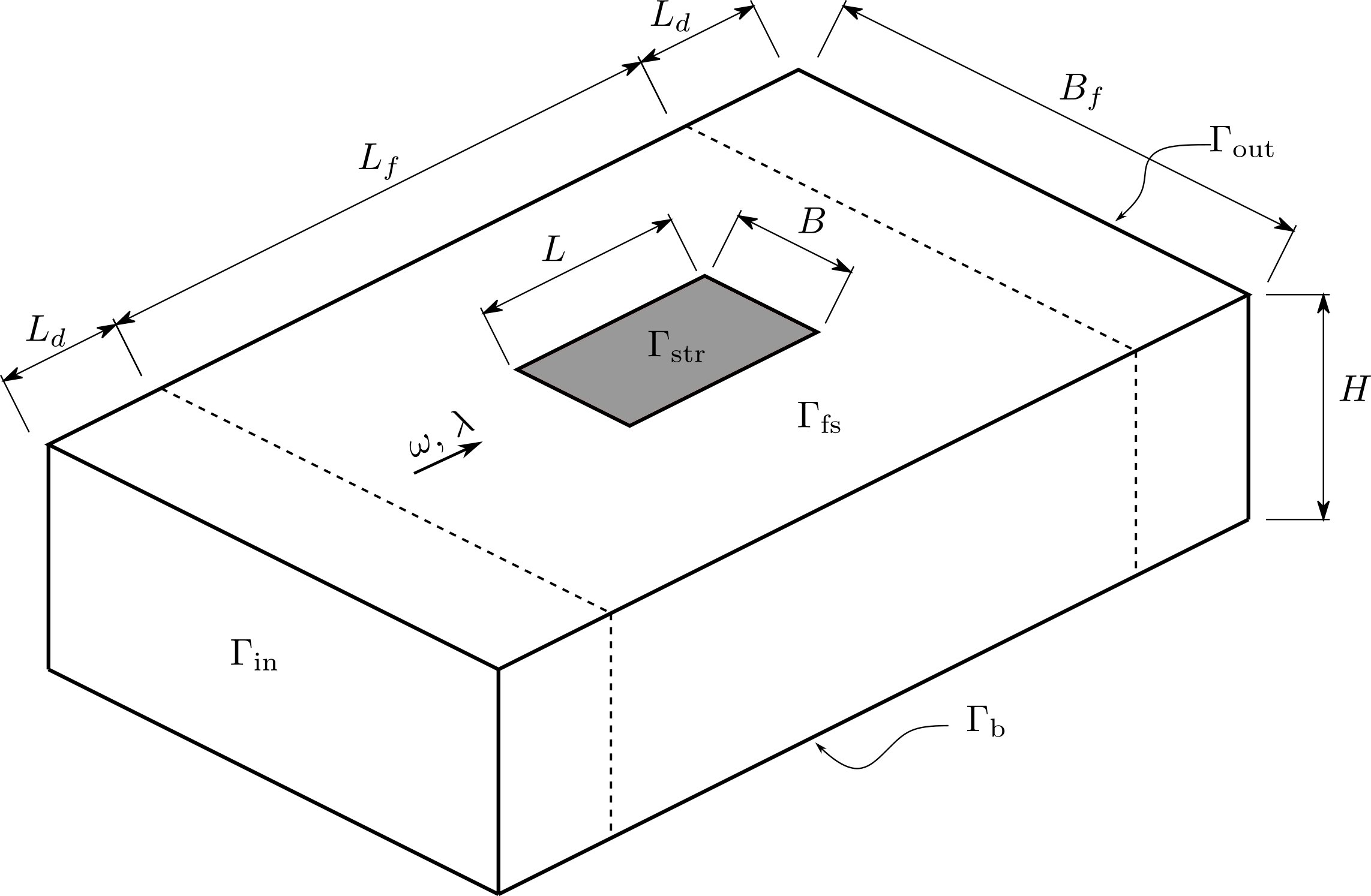}
	\caption{Sketch with the definition of the geometry used in the floating plate test.}	
	\label{fig:yago_geometry}
\end{figure}
\begin{subequations}\label{eq:exact_3d}
	\begin{align}
		\label{eq:exact_phi_3d}
		\phi((x,y,z),t) =& -\frac{\eta_0\omega}{k_\lambda}\frac{\cosh(k_\lambda z)}{\sinh(k_\lambda  H)}\sin(k_\lambda  x-\omega t),\\
		\label{eq:exact_eta_3d}
		\eta((x,y,z),t)=&\eta_0\cos(k_\lambda x-\omega t).
	\end{align}
\end{subequations}

In this test we use a different composition of the FE spaces for the potential and surface elevation. In order to reduce the overall computational time, we use Lagrange polynomials of order 2 for the potential FE space, $ \hat{\V}_{h} $, and 4th order Lagrange polynomials for the surface elevation FE space, $ \hat{\V}_{\Gamma,h} $. We use 32 elements through the $ x $-direction of the plate, \textit{i.e.} 320 elements in total in the $ x $-direction of the domain, 4 elements in the $ y $-direction of the plate, resulting in 56 elements in the $ y $-direction of the domain, and 4 elements in the vertical direction with exponential refinement close to the free surface.
\begin{figure}[pos=h!]
	\centering
	\includegraphics[width=0.49\textwidth]{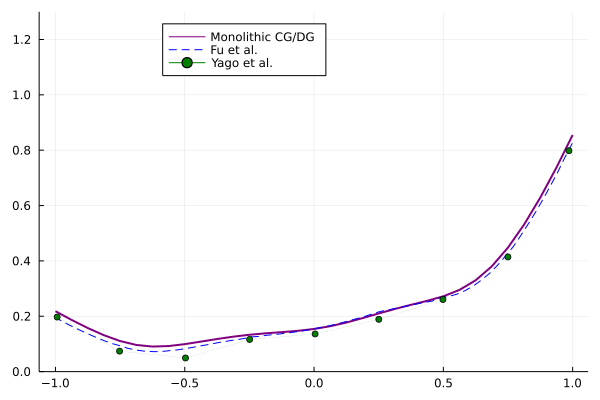}
	\includegraphics[width=0.49\textwidth]{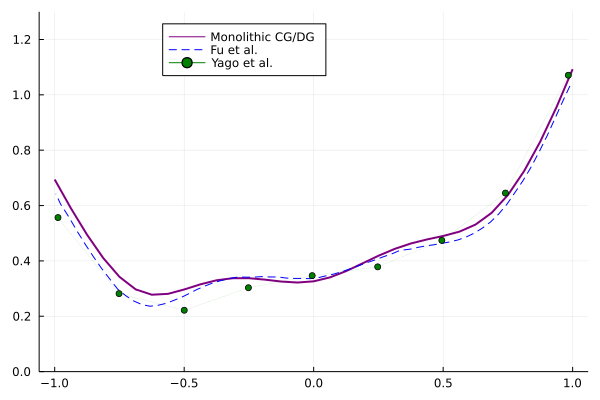}\\
	\includegraphics[width=0.49\textwidth]{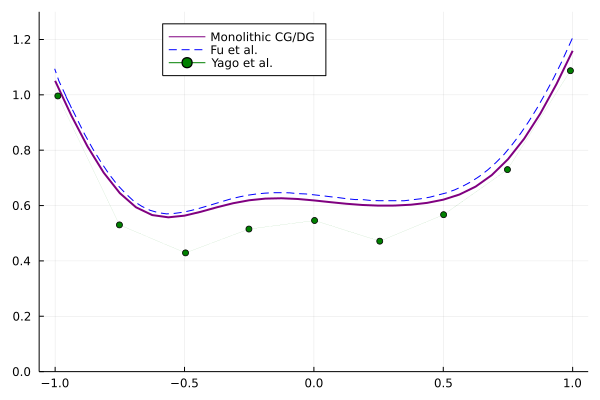}
	\caption{Relative surface elevation at the plate centerline for the case $ \lambda=0.4L $ (top left), $ \lambda=0.6L $ (top right) and $ \lambda=0.8L $ (bottom).}	
	\label{fig:Yago_experiment}
\end{figure} 

Looking at the relative surface elevation at the centerline of the plate shown in Figure~\ref{fig:Yago_experiment}, we see that the proposed approach is in good agreement with the numerical results of  Fu \textit{et al.}~\cite{fu2007hydroelastic}. This is the case for the three wave settings $ \lambda=0.4L $, $ \lambda = 0.6L $ and $ \lambda=0.8L $. In Figure~\ref{fig:yago_omega04} we depict the real, imaginary and absolute value of the surface elevation for the case $ \lambda=0.4L $.
\begin{figure}[p]
	\centering
	\includegraphics[clip=true,trim=13cm 0 15cm 0,width=0.65\textwidth]{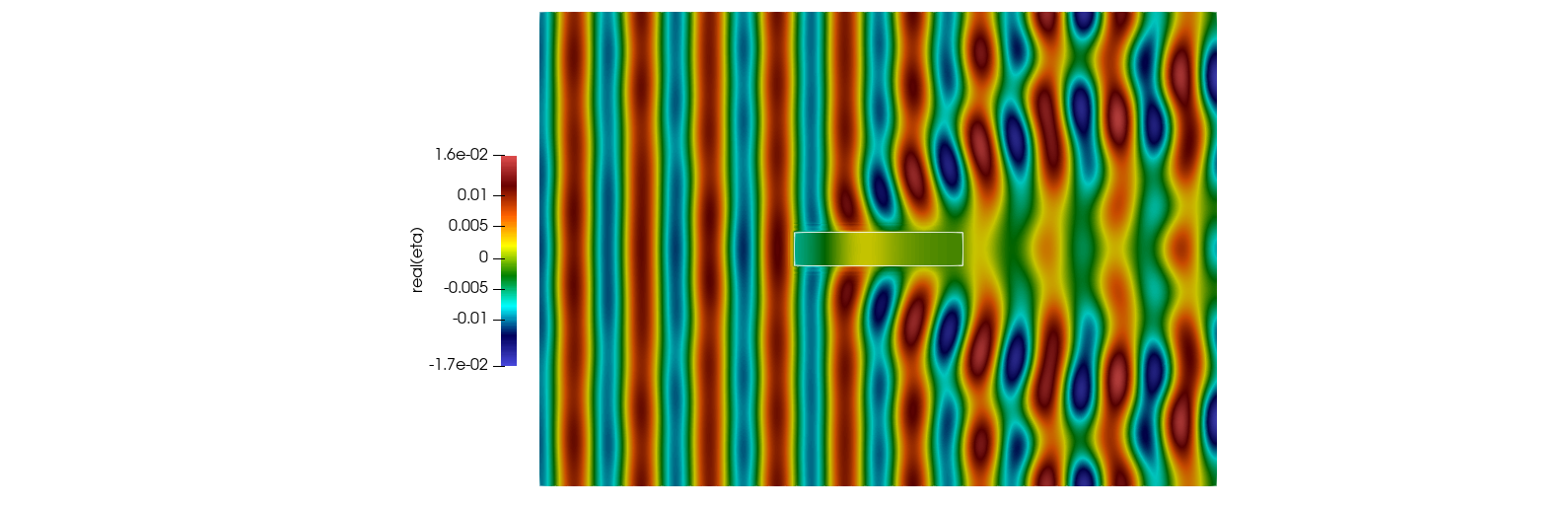}\\
	\includegraphics[clip=true,trim=13cm 0 15cm 0,width=0.65\textwidth]{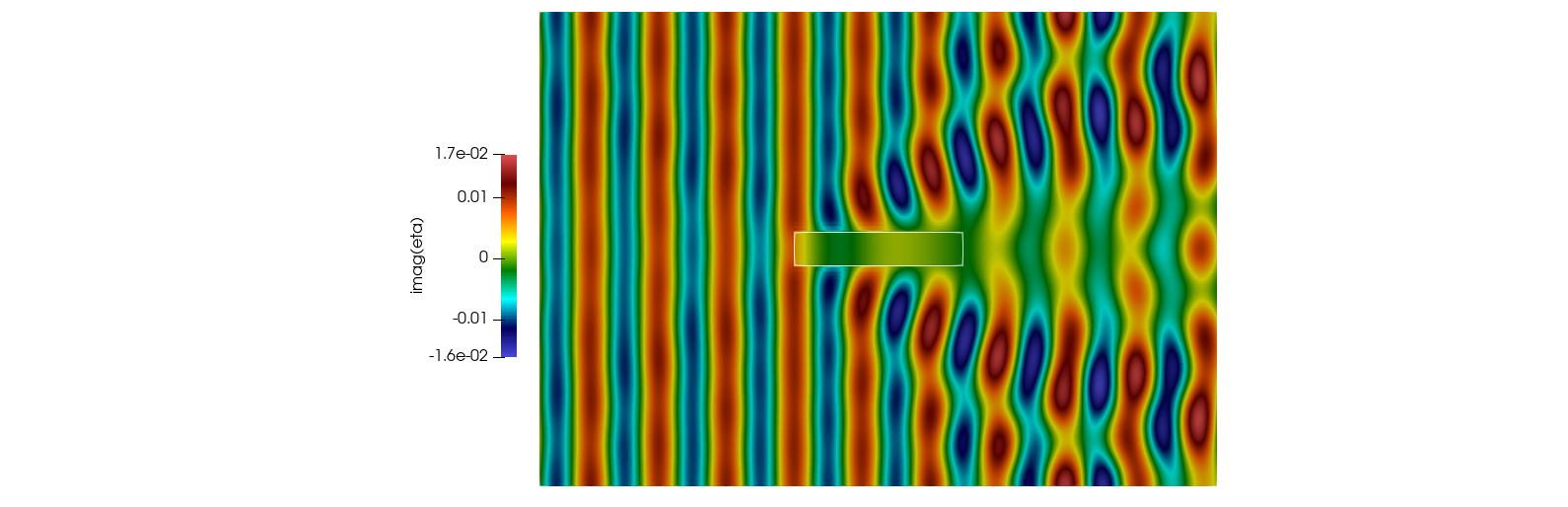}\\
	\includegraphics[clip=true,trim=13cm 0 15cm 0,width=0.65\textwidth]{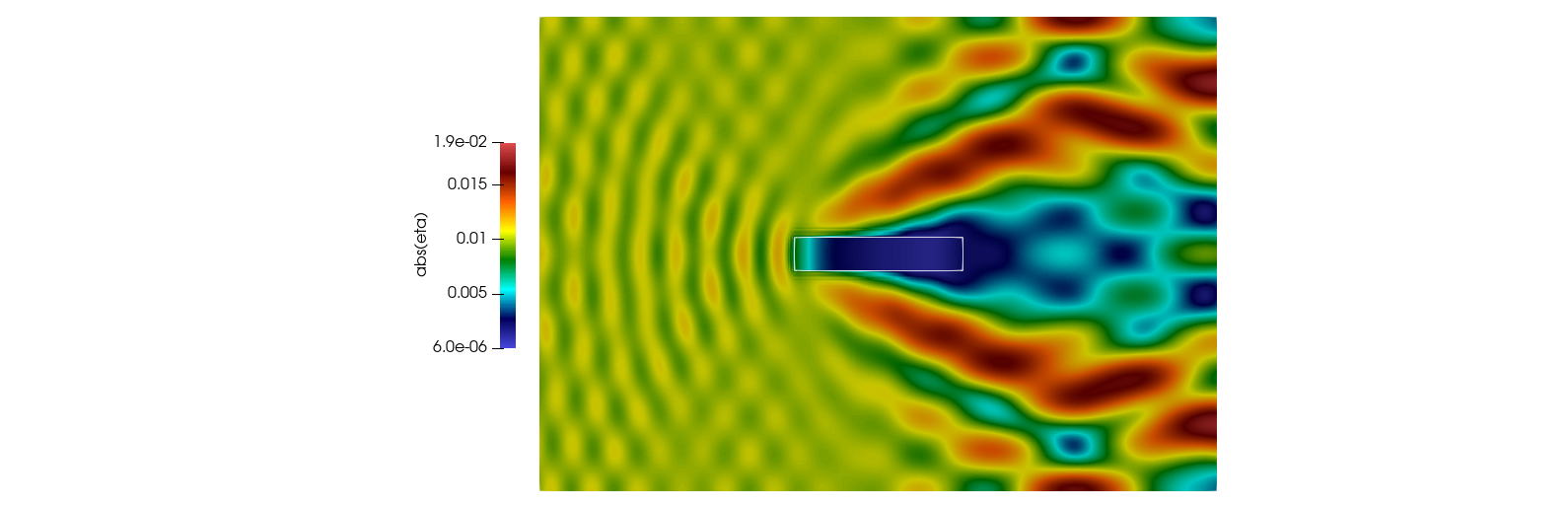}
	\caption{Surface elevation, $ \eta_h $ and $ \kappa_{h} $, for the case $ \lambda=0.4L $. Real part (top), imaginary part (center) and absolute values (bottom). The plate boundaries are marked in white.}
	\label{fig:yago_omega04}	
\end{figure}

\section{Conclusions}\label{sec:conclusions}
In this manuscript we present a novel monolithic FE formulation for the hydroelastic analysis of thin floating structures that can be described by the Euler-Bernoulli beam theory, for the 2-dimensional case, or by the Poisson-Kirchhoff plate theory. We define the new formulation for both, $C^1$ and $C^0$ FE spaces. The later uses a continuous/discontinuous Galerkin approach for the displacements and rotations, respectively, resulting in consistent, stable and energy-conserving formulations. We show a practical implementation of the monolithic formulation in the pure Julia library Gridap.jl. We have shown that the statements proven in the numerical analysis section are supported by the numerical results. We also see that the method proposed in this manuscripts leads to results that are in good agreement with other experimental and numerical works in the literature. We test the method for a wide variety of cases, including 2 and 3-dimensional geometries, structures with elastic joints or domains with variable bathymetry.

\section*{Acknowledgment}

F. Verdugo acknowledges support from the “Severo Ochoa Program for Centers of Excellence in
R\&D (2019-2023)” under the grant CEX2018-000797-S funded by MCIN/AEI/10.13039/501100011033.

\printcredits

\section*{Appendix A. Formulation for the floating Euler-Bernoulli beam}\label{sec:appendix_A}
\subsection*{Weak form}
Following the same steps as described in Section~\ref{subsec:weak_form}, one can derive the weak form for the 2-dimensional case, where the floating structure is modeled as a 1-dimensional Euler-Bernoulli beam. In that case, the equivalent bilinear form to \Eq{eq:bilinear_modified} would read
\begin{align}\label{eq:bilinear_modified_EB}
	B_{\scriptsize\mbox{EB}}([\phi,\eta],[w,v])\eqdef&(\nphi,\nabla w)_\Omega - (\eta_t,w)_{\Gf\cup\Gs}\\\nonumber
	+&\beta\left(\phi_t+g\eta,\alpha_fw+v\right)_{\Gf} +\left(d_0\eta_{tt} + \phi_t+g\eta,v\right)_{\Gs} + \left( \Drho\Delta\eta,\Delta v\right)_{\Gs}.
\end{align}

Equivalently, the bilinear form for a structure with joints is given by	
\begin{align}\label{eq:bilinear_modified_joints_beam}
	B_{\scriptsize\mbox{j,EB}}([\phi,\eta],[w,v])\eqdef&B_{\scriptsize\mbox{EB}}([\phi,\eta],[w,v])+\left(k_\rho\jump{\nabla\eta\cdot\nL},\jump{\nabla v\cdot\nL}\right)_{\Lj}.
\end{align}

\subsection*{Spatial discretization}
The equivalent semi-discrete problem using a C/DG formulation for the case of a floating Euler-Bernoulli beam reads: find $ [\phi_h,\eta_h] \in\hatV_h\times\hatV_{\Gamma,h}$ such that
\begin{equation}\label{eq:cgdg_form_beam}
	\hat{B}_{h,{\scriptsize\mbox{EB}}}([\phi_h,\eta_h],[w_h,v_h])=L_h([w_h,v_h])\quad\forall[w_h,v_h]\in\hatV_h\times\hatV_{\Gamma,h}.
\end{equation}
Where 
\begin{align}\label{eq:bilinear_modified_h_cgdg_beam}
	\hat{B}_{h,{\scriptsize\mbox{EB}}}([\phi_h,\eta_h],[w_h,v_h])\eqdef
	&(\nabla\phi_h,\nabla w_h)_{\Omega_h} - (\eta_{h,t},w_h)_{\Gfh\cup\Gsh}\\\nonumber
	+&\beta\left(\phi_{h,t}+g\eta_h,\alpha_fw_h+v_h\right)_{\Gfh} +\left(d_0\eta_{h,tt} 
	+ \phi_{h,t}+g\eta_h,v_h\right)_{\Gsh} \\\nonumber
	+&\left( \Drho\Delta\eta_h,\Delta v_h\right)_{\Gsh} +\left(k_\rho\jump{\nabla\eta_h\cdot\nL},\jump{\nabla v_h\cdot\nL}\right)_{\Lj} \\\nonumber
	-&\left( \average{\Drho\Delta\eta_h},\jump{\nabla v_h\cdot\nL}\right)_{\Lsh} 
	- \left(\jump{\nabla\eta_h\cdot\nL},\average{\Drho\Delta v_h}\right)_{\Lsh}\\\nonumber
	+&\frac{\gamma}{h}\left(\jump{\Drho\nabla\eta_h\cdot\nL},\jump{\nabla v_h\cdot\nL}\right)_{\Lsh}.
\end{align}

\subsection*{Time discretization}
The equivalent fully discrete problem in the frequency domain, as it has been defined in Section~\ref{subsubsec:freq_domain}, for the case of a 1-dimensional floating Euler-Bernoulli beam in a 2-dimensional domain would be given by: find $ [\phi_h,\eta_h] \in\hatV^\omega_h\times\hatV^\omega_{\Gamma,h}$ such that
\begin{equation}\label{eq:cgdg_form_freq_domain_beam}
	\hat{B}^\omega_{h,{\scriptsize\mbox{EB}}}([\phi_h,\eta_h],[w_h,v_h])=L^\omega_h([w_h,v_h])\quad\forall[w_h,v_h]\in\hatV^\omega_h\times\hatV^\omega_{\Gamma,h}.
\end{equation}
With
\begin{align}\label{eq:bilinear_modified_h_cgdg_freq_domain_beam}
	\hat{B}^\omega_{h,{\scriptsize\mbox{EB}}}([\phi_h,\eta_h],[w_h,v_h])\eqdef&(\nabla\phi_h,\nabla w_h)_{\Omega_h} 
	+ (i\omega\eta_{h},w_h)_{\Gfh\cup\Gsh}\\\nonumber
	+&\beta\left(g\eta_h-i\omega\phi_{h},\alpha^\omega_fw_h+v_h\right)_{\Gfh} +\left((g-\omega^2d_0)\eta_{h} - i\omega\phi_{h},v_h\right)_{\Gsh}\\\nonumber
	+&\left( \Drho\Delta\eta_h,\Delta v_h\right)_{\Gsh} 
	+ \left(k_\rho\jump{\nabla\eta_h\cdot\nL},\jump{\nabla v_h\cdot\nL}\right)_{\Lj} \\\nonumber
	-&\left( \average{\Drho\Delta\eta_h},\jump{\nabla v_h\cdot\n_{\Lsh}}\right)_{\Lsh} 
	- \left(\jump{\Drho\nabla\eta_h\cdot\n_{\Lsh}},\average{\Delta v}\right)_{\Lsh}\\\nonumber
	+&\frac{\gamma}{h}\left(\jump{\Drho\nabla\eta_h\cdot\n_{\Lsh}},\jump{\nabla v_h\cdot\n_{\Lsh}}\right)_{\Lsh}.
\end{align}

Alternatively, the fully discrete problem in the time domain, as it has been defined in Section~\ref{subsubsec:time_domain}, for the case of a 1-dimensional floating Euler-Bernoulli beam in a 2-dimensional domain would be given by: find $ [\phi_h^{n+1},\eta_h^{n+1}] \in\hatV_h\times\hatV_{\Gamma,h}$ such that
\begin{equation}\label{eq:cgdg_t_form_beam}
	\hat{B}_{h,{\scriptsize\mbox{EB}}}^{n+1}([\phi_h^{n+1},\eta_h^{n+1}],[w_h,v_h])=L_h^{n+1}([w_h,v_h])\quad\forall[w_h,v_h]\in\hatV_h\times\hatV_{\Gamma,h}.
\end{equation}
Where
\begin{align}\label{eq:bilinear_modified_h_t_cgdg_beam}
	\hat{B}_{h,{\scriptsize\mbox{EB}}}^{n+1}([\phi_h^{n+1},\eta_h^{n+1}],[w_h,v_h])\eqdef
	&(\nabla\phi_h^{n+1},\nabla w_h)_{\Omega_h} 
	- (\delta_t\eta_{h}^{n+1},w_h)_{\Gfh\cup\Gsh}\\\nonumber
	+&\beta\left(\delta_t\phi_{h}^{n+1}+g\eta_h^{n+1},\alpha_fw_h+v_h\right)_{\Gfh}
	+\left(\delta_{tt}d_0\eta_{h}^{n+1} + \delta_t\phi_{h}^{n+1}+g\eta_h^{n+1},v_h\right)_{\Gsh}\\\nonumber 
	+&\left( \Drho\Delta\eta_h^{n+1},\Delta v_h\right)_{\Gsh} 
	+\left(k_\rho\jump{\nabla\eta_h^{n+1}\cdot\nL},\jump{\nabla v_h\cdot\nL}\right)_{\Lj}\\\nonumber
	-&\left( \average{\Drho\Delta\eta_h^{n+1}},\jump{\nabla v_h\cdot\nL}\right)_{\Lsh} 
	- \left(\jump{\Drho\nabla\eta_h^{n+1}\cdot\nL},\average{\Delta v}\right)_{\Lsh}\\\nonumber
	+&\frac{\gamma}{h}\left(\jump{\Drho\nabla\eta_h^{n+1}\cdot\nL},\jump{\nabla v_h\cdot\nL}\right)_{\Lsh} .
\end{align}

\begin{comment}

\section*{Appendix B. Definitions, lemmas and theorems}\label{sec:appendix_B}
\begin{definition}
	For any $w\in C^\infty(\Omega)$, define the trace operator $ \gamma|_{\partial\Omega} $ by   \begin{equation}\label{eq:trace}
		\gamma|_{\partial\Omega}w(x)=w(x),\qquad x\in\partial\Omega.
	\end{equation}
\end{definition}
\begin{theorem}[Trace theorem of Sobolev spaces]\label{theorem:trace2}
	Let $ \Omega $ be a bounded simply connected Lipschitz domain. Then, the trace operator $ \gamma|_{\partial\Omega} $ is a bounded linear operator from $ H^1(\Omega) $ to $ H^{\frac{1}{2}}(\partial\Omega) $. That is, \[ \left\|\gamma|_{\partial\Omega}w\right\|_{H^{\frac{1}{2}}(\partial\Omega)}\leq C_{\partial\Omega} \left\|w\right\|_{H^1(\Omega)}. \]
	With $ C $ a constant that only depends on $ \partial\Omega $.
\end{theorem}
See \cite{ding1996proof} for a proof of Theorem~\ref{theorem:trace2}.
\begin{theorem}\label{theorem:generalized_poincare}
	Let $ \Omega $ be a bounded connected Lipschitz domain and $ f $ be a linear form from $ H^1(\Omega) $ with a non-zero restriction on constant functions.
	 %whose restriction on constant functions is not zero. 
	 %The restriction of $f$ on constant functions is non-zero. 
	 Then, there is a constant $ C_\Omega > 0 $ such that 
	\begin{equation}\label{key}
		C_\Omega\|w\|_{H^1(\Omega)}\leq \|\nabla w\|_{L^2(\Omega)} + |f(w)|, \qquad\forall w\in H^1(\Omega).
	\end{equation}
\end{theorem}
\end{comment}

\bibliographystyle{cas-model2-names}

\bibliography{refs.bib}
%\printbibliography

\end{document}